\documentclass[11pt]{article}

\usepackage[latin1]{inputenc}
\usepackage[T1]{fontenc}
\usepackage{color}
\usepackage{eurosym}
\usepackage[top=4cm,bottom=4cm,left=3cm,right=3cm]{geometry}
\usepackage{fancyhdr}
\usepackage{float}
\usepackage{amssymb,amsmath, amsthm}
\usepackage{url,graphicx}

\usepackage{booktabs}
\usepackage{stmaryrd} 
\usepackage{dsfont}

\usepackage{color}
\usepackage{comment}

\usepackage[bbgreekl]{mathbbol} 
\usepackage{authblk}
\usepackage{multirow}
\usepackage{enumerate} 
\usepackage{ragged2e}
\newtheorem{theorem*}{Theorem}[section]
\newtheorem{ass*}[theorem*]{Assumption}
\newtheorem{note*}[theorem*]{Note}
\newtheorem{lemma*}[theorem*]{Lemma}
\newtheorem{definition*}[theorem*]{Definition}
\newtheorem{proposition*}[theorem*]{Proposition}
\newtheorem{corollary*}[theorem*]{Corollary}
\newtheorem{remark*}[theorem*]{Remark}
\newtheorem{example*}[theorem*]{Example}

\numberwithin{equation}{section}

\def\bd{\begin{description}}
\def\ed{\end{description}}

\def\D2{\bbD_{2,\infty-}}

\def\D{{\bf D}}

\def\F{{\bf F}}

\def\cala{{\cal A}}
\def\calb{{\cal B}}

\def\cald{{\cal D}}
\def\cale{{\cal E}}
\def\calf{{\cal F}}

\def\calh{{\cal H}}

\def\calm{{\cal M}}
\def\caln{{\cal N}}

\def\calt{{\cal T}}
\def\calu{{\cal U}}

\def\cadlag{c\`adl\`ag\ }
\def\caglad{c\`agl\`ad\ }

\def\be{\begin{equation}}
\def\ee{\end{equation}}
\def\bea{\begin{eqnarray}}
\def\eea{\end{eqnarray}}
\def\beas{\begin{eqnarray*}}
\def\eeas{\end{eqnarray*}}
\def\bi{\begin{itemize}}
\def\ei{\end{itemize}}

\def\bd{\begin{description}}
\def\ed{\end{description}}
\def\l{\left}
\def\r{\right}

\newcommand{\bbD}{{\mathbb D}}

\newcommand{\reels}{\mathbb{R}}
\newcommand{\naturels}{\mathbb{N}}

\newcommand{\esp}{\mathbb{E}}
\newcommand{\proba}{\mathbb{P}}

\newcommand{\ito}{It\^{o}}

\begin{document}
 \title{Cointegration in high frequency data\footnote{We would like to thank an anonymous referee for helpful comments and advice. The research of Yoann Potiron is supported by Japanese Society for the Promotion of Science Grant-in-Aid for Young Scientists (B) No. 60781119 and a special grant from Keio University. The research of Simon Clinet is supported by Japanese Society for the Promotion of Science Grant-in-Aid for Young Scientists No. 19K13671.}}
\author{Simon Clinet\footnote{Faculty of Economics, Keio University. 2-15-45 Mita, Minato-ku, Tokyo, 108-8345, Japan. Phone:  +81-3-5427-1506. E-mail: clinet@keio.jp website: http://user.keio.ac.jp/\char`\~clinet/}   and Yoann Potiron\footnote{Faculty of Business and Commerce, Keio University. 2-15-45 Mita, Minato-ku, Tokyo, 108-8345, Japan. Phone:  +81-3-5418-6571. E-mail: potiron@fbc.keio.ac.jp website: http://www.fbc.keio.ac.jp/\char`\~ potiron}}
\date{This version: \today}

\maketitle

\begin{abstract}
In this paper, we consider a framework adapting the notion of cointegration when two asset prices are generated by a driftless It\^{o}-semimartingale featuring jumps with infinite activity, observed regularly and synchronously at high frequency. We develop a regression based estimation of the cointegrated relations method and show the related consistency and central limit theory when there is cointegration within that framework. We also provide a Dickey-Fuller type residual based test for the null of no cointegration against the alternative of cointegration, along with its limit theory. Under no cointegration, the asymptotic limit is the same as that of the original Dickey-Fuller residual based test, so that critical values can be easily tabulated in the same way. Finite sample indicates adequate size and good power properties in a variety of realistic configurations, outperforming original Dickey-Fuller and Phillips-Perron type residual based tests, whose sizes are distorted by non ergodic time-varying variance and power is altered by price jumps. Two empirical examples consolidate the Monte-Carlo evidence that the adapted tests can be rejected while the original tests are not, and vice versa.

\end{abstract}
\textbf{Keywords}: cointegration; deflation; high frequency data; It\^{o}-semimartingale; residual based test; truncation; unit root test
\\

\section{Introduction}
It is often the case that time series cruise componentwise, but that a linear combination of the components does not drift apart. Since the seminal papers of \cite{granger1981some} and \cite{engle1987co}, cointegration has spread across and way beyond the field of econometrics. The authors bring forward a residual based two-step strategy to test for the presence of cointegration. The first step corresponds to estimation of cointegrating relations via regression. The second step is closely related to test of unit roots. More specifically, residual based tests are designed to test the null of no cointegration by relying on a unit root test on the residuals. If the null of unit root test is not rejected, then the null of no cointegration is also not rejected. Among the class of unit root tests, standard Dickey-Fuller (DF) tests originally from \cite{dickey1981likelihood} and augmented procedure from \cite{said1984testing}, and Phillips-Perron tests from \cite{phillips1988testing}, are among the most popular. \cite{phillips1987time} shows that this type of unit root tests is typically robust to many weakly dependent and heterogeneously distributed time
 series. As for cointegration tests, a large sample limit theory for the residual based tests is investigated in \cite{phillips1990asymptotic}. In particular, the critical values tabulated for the pure unit root tests are altered due to the error made in the first step. In the cited paper, and to the best of our knowledge in the rest of the literature on cointegration, the asymptotics is low frequency in the sense that the time gap between two observations $\Delta $ is fixed while the horizon time $T \rightarrow \infty$.  
 
 \smallskip
There is a solid body of empirical work employing cointegration in  financial economics. In that field, the most common specification of cointegration definition is that all the components of the observed vector $x_t$ are unit roots (e.g. random walks) and that there exists a vector $\alpha$ so that $\alpha' x_t$ is stationary. Examples of application include and are not limited to nominal dollar spot exchange rates, e.g. \cite{baillie1989common} and \cite{diebold1994cointegration}, price discovery, e.g. \cite{hasbrouck1995one}, and pairs trading, e.g. \cite{caldeira2013selection} and the review of \cite{krauss2017statistical}. Although the major part of earlier empirical studies was confined to data observed on a daily basis, nowadays they frequently incorporate data observed on the intraday basis, see e.g. \cite{elyasiani2001interdependence}, \cite{hasbrouck2003intraday}, \cite{pati2011intraday} or \cite{yang2012intraday} among others. This increasing use of high frequency data is perfectly natural, and should even be encouraged given the basic statistical principle that one shall not throw away data. As a matter of fact, it echoes similar changes in other areas of the literature, the best example being the diversification and sophistication of efficient variance estimation measures over the past two decades.

\smallskip
Our concern in this paper is to back up the empirical use of high frequency data by theoretically validated high frequency robust estimation of cointegrating relations and tests. More specifically, the difference between the classical \emph{time series} framework and the \emph{high frequency} framework we consider is that in the latter $\Delta  \rightarrow 0$ while keeping $T \rightarrow \infty$, and the observed time series $x_t$ will be generated by a driftless It\^{o}-semimartingale. This environment, quite standard in high frequency financial econometrics, typically accommodates for a variety of stylized facts, such as time-varying variance featuring jumps, leverage effect and price jumps, all of which are salient in high frequency data. Price jumps can happen at deterministic times (e.g. macroeconomic news announcements) or not, are usually of random size, and quite frequent, at least ten a year on average according to Table 8 (p. 484) in \cite{huang2005relative}. For the purpose of simplicity, we restrict ourselves to the two dimensional case.

\smallskip
The question of testing for no cointegration with high frequency data is of practical relevance since some of its stylized facts can lead to significant distortions of the power of the usual tests. In an extensive Monte Carlo experiment, \cite{krauss2017power} implement ten leading cointegration tests in a variety of set-up tailored with high-frequency features. They use an AR(1) with normal innovations as benchmark, but also look at non-normality effects employing t-distributed innovations, GARCH effects, nonlinearities and price jumps. They typically find that among a cohort of high frequency stylized facts, price jumps deteriorate the power the most.

\smallskip
In fact, the current theoretical framework to test for no cointegration is quite restricted and only accommodates for ergodic time-varying variance. Spurious cointegration can occur in the presence of a mere jump in long term equilibrium variance, as documented by \cite{noh2003behaviour}. In other words, this means that non ergodic time-varying variance can distort the size of the tests. Obviously, time-varying variance is not solely a high frequency stylized fact, as \cite{sensier2004testing} report that most of the real and price variables in the dataset from \cite{stock1999comparison} reject the null hypothesis of constant variance against one regime shift alternative. As far as we know, the only attempt to accommodate for time-varying variance in a more flexible cointegration model is \cite{kim2010cointegrating}, but unfortunately in that paper the authors test the reverse null of cointegration.
 
 \smallskip
Our aim is to adapt the residual based DF procedure developed in \cite{phillips1990asymptotic} to test for the null hypothesis of no cointegration in high frequency data. More specifically, the residual based tests of the cited paper are not theoretically robust to the two aforementioned high frequency stylized facts pointed by the literature as distorting size and power -time-varying variance and price jumps-. Accordingly, we develop two separate adaptations, one for each feature, that we eventually combine with each other, as the two problems are not of the same nature and cannot simply be tackled with a unique straightforward adaptation. As far as we know, no cointegration test has been tuned to those two high frequency stylized facts, yet there are accordingly two very related fields from the literature. 
 
 \smallskip
The first active field is about testing for the presence of a unit root incorporating time-varying variance. \cite{cavaliere2008time} apply a time transformation to the data prior to use the standard unit root tests and retrieve the same asymptotics. \cite{cavaliere2007testing} employ the time deformation to simulate valid critical values for standard unit root tests. \cite{cavaliere2008bootstrap} and \cite{cavaliere2009heteroskedastic} consider a related method involving the wild bootstrap. Wild bootstrap tests formed from a feasible generalized least squares are developed in \cite{boswijk2018adaptive}. Finally, \cite{beare2018unit}  preestimate volatility and exploit it to deflate the returns prior the use of the standard tests. Unfortunately, we were not able to find track of the use of those nice methods in the context of residual based cointegration test. In addition, the theoretical framework of all those papers is restricted to deterministic volatility and no price jumps, except for \cite{cavaliere2009heteroskedastic} who allow for random variance.
 
 \smallskip
Another very dynamic field is about residual based cointegration tests with regime changes, which is typically related to price jumps in the It\^{o}-semimartingale setup. In a univariate context, unit root tests were designed in \cite{perron1989great}, \cite{perron1990testing}, \cite{banerjee1992recursive}, \cite{perron1992nonstationarity} and \cite{zivot1992further}. \cite{gregory1996residual} extend the tests in a cointegration framework. In all those papers, the number of shifts is at most equal to unity. Moreover, \cite{hatemi2008tests} allows for two possible breaks. Finally, \cite{maki2012tests} proposes tests which permits an arbitrary number of breaks in principle, although the method is seldom used with more than five breaks in practice. The strong limitation of the regime change approach is that jumps are required to happen at deterministic times, with deterministic size, known number of jumps, and not very often, all of which are not consistent with the aforementioned high-frequency stylized facts.
 
 \smallskip
 To obtain the robustness of the residual based test to time-varying variance, the most natural candidate consists in a residual based test based on the aforementioned time-varying variance robust unit root tests. Unfortunately, a transformation of data solely at the unit root test step entails hybrid-behaved tests to the extent that the limit distribution is unidentifiable, at least to us. To circumvent this difficulty, we deflate the returns prior to the two steps. As for that to price jumps, we simply consider a truncation method as in \cite{mancini2009non}, which is commonly used in the literature on high frequency econometrics when dealing with jumps with random size and times.
 
 \smallskip
Our theoretical contribution is divided into two parts. First, we provide a regression procedure based on deflated and truncated asset price returns in order to estimate the cointegrated system and establish the related consistency and central limit theory when there is cointegration. This is related to Theorem \ref{thmH1} in Section \ref{sectionLimitTheoryOLS}. Second, we develop a DF type residual-based cointegration test, also based on deflated and truncated observations, along with its central limit theory in the absence of cointegration and under local alternatives. We also provide with an equivalent of the expression in the presence of cointegration (see Theorem \ref{thmH0DF} and Proposition \ref{propH1}). In particular we show that under the null hypothesis of no cointegration the limit distribution of the statistic is that of a classical DF test, and that no local power is lost due to the truncation and the deflation.

 \smallskip
 The finite sample from this paper corroborates that from cited papers and our theoretical findings. On the one hand, a typical environment of cointegrated assets observed at high frequency distorts badly size and power of the DF or Phillips-Perron type residual based unit root tests, and as expected we can isolate the effect of time-varying variance as impacting the former, and that of price jumps as altering the latter. On the other hand, the proposed test size is appropriate and its power is good in the simultaneous presence of both stylized facts. Two empirical examples illustrate the fact that the proposed test can be in disagreement with the original tests, indicating the practical relevance of asset prices truncation and deflation. 
 
 \smallskip
One big limitation of our technique is that we have no drift. Such a simple assumption is quite restrictive since the concept of cointegration is mainly used to analyze the long run comovements of multiple time series and the role of stochastic drift can be very important in the long run relationship. Correspondingly, we examine the case of \emph{linear} drifting \ito-semimartingales and show that, even though theoretical derivations are possible, the limit distribution of the statistic is quite complex, indicating that combining drift and deflation is a difficult task. We are well aware that linear drift is rather a strong assumption, but overall there is a lack of literature on drift. Some exceptions include \cite{perron2016residuals}, which deals with optimal methods for unit-root and cointegration tests when a trend is present (but no time-varying volatility), and \cite{beare2018unit} (Section 4.3), which gives a general detrending method for time-linear trends in the presence of time-varying variance for the unit-root test. Theorem 4.3 from the latter work gives the related theoretical asymptotic results. In particular, the author already noticed that the modified limit distribution under the null strongly depends on the shape of volatility and so has to be computed numerically every time the method is used. Our own method detailed in Section 4.2 uses a similar detrending device (adapted to our own high-frequency scheme). Our Theorem 4.5 essentially generalizes the result of \cite{beare2018unit} to the case of cointegration and presents the same feature, i.e that the detrended data yields limit distribution under the null which explicitly depends on the volatility curve. To our knowledge, there is no method (even in unit root literature) which deals with time linear trends in the presence of time-varying volatility and which yields a limit distribution independent of the volatility process. In addition, finite sample sheds light on the fact that drift does not seem to affect both size and power of all the considered residual based tests, at least in a range of realistic configurations.

\smallskip
 Another obvious limitation of our approach is that the framework considered rules out market microstructure noise. This is mitigated by the fact that the vast majority of empirical studies related to cointegration operating with financial time series does not sample with a frequency higher than five minutes so that if the asset is liquid enough the data are reasonably free of market microstructure noise (see, e.g., \cite{ait2019hausman}). To temper as much as possible the effect of microstructure noise, we do not sample faster than every ten minutes in both our numerical and empirical studies.
 
 \smallskip
 The remaining of this paper is structured as follows. The framework which is a natural adaptation to high frequency data is given in Section 2. Estimation of cointegrated relations, based on regression on the truncated and deflated price, and its related limit theory under cointegration is provided in Section 3. A DF-type test of the null hypothesis that there is no cointegration against the alternative that there is cointegration, together with its limit theory, is developed in Section 4. In particular, this test's robustness to time-varying variance and price jumps, which is based respectively on deflation and truncation, is established. Section 5 is devoted to a Monte Carlo experiment, to shed light on the size and power twist of the original DF and Phillips-Perron based cointegration tests and the good behavior of the adapted DF test in a variety of realistic configurations. In Section 6, a brief empirical study, which corroborates the fact that the proposed test is not always in accordance with the original tests, is conducted. We conclude in Section 7. Proofs and part of the numerical results have been relegated to the Appendix for the sake of clarity.

 \section{The framework: a natural adaptation to high frequency data} \label{sectionCointModel}
In this section, we introduce our general framework, which in particular accommodates the definition of \emph{no cointegration} (\cite{engle1987co}) when two driftless It\^{o}-processes including jumps with infinite activity are observed synchronously and regularly. Our introduced framework also specifies the notion of \emph{cointegration} and \emph{weak cointegration}.

\smallskip
We introduce a few key concepts from the aforementioned paper and the existing literature on cointegration for time series that will help motivate the framework of our own study. For the sake of clarity, we restrict ourselves to the case of a pair of processes, but all the definitions can be extended to the multivariate case with no major difficulty. Let us consider two unit root time series $x_t$ and $y_t$ (i.e. with explosive variance and with stationary increment processes $\Delta x_t$ and $\Delta y_t$). As pointed out in \cite{engle1987co}, in general, any linear combination $y_t - \alpha x_t$ will be again a unit root process drifting away from zero. In that case, $x_t$ and $y_t$ are said not to be cointegrated. However, it may also happen that for some $\alpha$, the time series $y_t- \alpha x_t$ does not wander far from zero, and not only its increments but also the series itself is stationary. The couple
$(x_t,y_t)$ is then said to be cointegrated with cointegration vector $(1,-\alpha)^T$. When conducting tests, and for the sake of tractability, both notions are often naturally embedded (see e.g. model (4.7) of \cite{engle1987co}) in an AR(1) model specification as follows. We assume that $x_t$ is a unit root process, and that
\bea \label{cointTimeSeries} 
y_t = \alpha x_t + \epsilon_t, \textnormal{ with } \epsilon_t = \rho \epsilon_{t-1} + u_t,
\eea 
where $u_t$ is a non-trivial stationary process possibly correlated with $\Delta x_t$. Then, $0 < \rho < 1$ yields a cointegrated system, whereas $\rho = 1 $ implies that $\epsilon_t$ is a unit root process, hence no cointegration. Moreover, the regime $0 < \rho < 1 $ with $\rho \to 1$, yields a process $\epsilon_t$ which is a nearly integrated process (as first introduced in \cite{nabeya1994local}), and accordingly the system (\ref{cointTimeSeries}) can be said weakly cointegrated. Residual based tests for the null of no cointegration usually pre-estimate $\alpha$ by, for instance, an ordinary least squares (OLS) regression, and then run a unit root test (i.e $\rho=1$ versus $\rho<1$) on the estimated residual $\widehat{\epsilon}_t = y_t - \widehat{\alpha} x_t$. As detailed in \cite{phillips1990asymptotic}, it is of importance to mention that such a two-step procedure affects the limit distribution of the test statistic so that testing for cointegration does not amount to directly testing for a unit root in $\widehat{\epsilon}_t$. In particular and of practical relevance, the critical values are altered. This deviation from the unit root case stems from the inconsistency of $\widehat{\alpha}$ (hence $\widehat{\epsilon}_t$) under the null of no cointegration.

\smallskip
In view of the above discussion, our first goal consists in introducing a model similar to (\ref{cointTimeSeries}) where now $\Delta x_t$ and $u_t$ are replaced by increments of driftless \ito-semimartingales. We return to the case of drifting processes in a detailed discussion at the end of Section \ref{SectionDrift}. We assume that we observe regularly (i.e. at times $t_0 := 0, t_1 := \Delta, \cdots, t_n := n \Delta$ with $\Delta := T/n$) between 0 and the horizon time $T$ (depending on $n$) two \cadlag (right continuous with left limits) processes $X$ and $Y$.\footnote{A full specification of the model actually involves the stochastic basis $\calb=(\Omega,\proba,\calf,\F)$, where $\F=(\calf_t)_{t\geq0}$ is a filtration satisfying the usual conditions and $\calf = \vee_{t \geq0} \calf_t$. We assume that all the processes are $\F$-adapted. Also, when referring to It\^{o}-semimartingale, we automatically mean that the statement is relative to $\F$.} For any process $A$, we use the following conventions for $i\in\{0,\cdots,n\}$:
\bea
A_i & := & A_{t_i},\label{notation1}\\ 
\Delta A_i & := & A_i - A_{i-1} \text{ for } i \neq 0,\label{notation2}\\ 
\Delta A_0 & := & 0.\label{notation3}
\eea 
In what follows, we assume that $X$ and $Y$ may be further decomposed as
\bea \label{defXYBasic}
X_t = X_t^c + J_t^X \textnormal{ and } Y_t = Y_t^c + J_t^Y, \textnormal{ }t\in[0,T]
\eea 
where $X^c$ and $Y^c$ are the continuous parts of $X$ and $Y$, and $J^X$ and $J^Y$ are pure jump processes such that, for $U \in \{X,Y\}$, $t \in[0,T]$,
$$ J_t^U = \int_{[0,t] \times E} \delta^U (s,z) \mu^U(ds, dz),$$
where $\mu^U$ is a Poisson random measure on $\reels_+ \times E$ for $E$ some auxiliary Polish space, $\nu^U$ is the compensator of $\mu^U$ of the form $\nu^U(ds,dz)=ds \otimes \lambda^U(dz) $ where $\lambda^U$ is a $\sigma$-finite measure, and where $\delta^U$ is a predictable function on $\Omega \times \reels_+ \times E$. Moreover, we assume that there exists $r \in [0,1)$ such that 
\bea \label{assJumps}
\sup_{t \in \reels_+} \esp\int_E (|\delta^U(t,z)|^r \vee |\delta^U(t,z)|^8) \lambda^U(dz) < +\infty.  
\eea 
In particular, although the jump processes may feature an infinite number of jumps on bounded time intervals, (\ref{assJumps}) ensures that the jumps are summable on $[0,T]$ and of order at most $T$, since it implies that
\bea \label{momentJumps}
\esp \sum_{0 <s \leq T} |\Delta J_s^U| \leq  \esp \sum_{0 <s \leq T} |\Delta J_s^U|^r \vee |\Delta J_s^U|^8 \leq KT
\eea 
for some constant $K \geq 0$. The summability property of the jumps is used extensively in our proofs (see e.g Lemma \ref{lemmaJumps} in the Appendix) in order to control most deviations involving jumps. Unfortunately, results such as those in Lemma \ref{lemmaJumps} may not hold when the jump processes are not of finite variation. As far as we know, it is not clear whether relaxing (\ref{assJumps}) to the case of non-summable jumps (for instance assuming $r \in [0,2)$ only) is possible in the present framework, which is why we leave it for future research. In particular, note that it is well-known that when the jumps are not summable, slower rates of convergence for even simple quantities such as threshold realized volatility should be expected (\cite{jacod2014remark}), indicating that this case is non-standard.   

\smallskip
We now assume that $X$ and $Y$ satisfy a relation of the same nature as (\ref{cointTimeSeries}). Assuming first that $J^X = J^Y = 0$, we naturally adapt (\ref{cointTimeSeries}) as follows. We assume that  there exist $c_0, \alpha_0 \in \reels$ such that for any $i \in \{1,\cdots,n\}$, we have

\bea \label{cointContinuous}
Y_i^c = c_0 + \alpha_0 X_i^c +  \epsilon_i, \textnormal{ with } \epsilon_i = \rho \epsilon_{i-1} + \Delta Z_i, \textnormal{ }\epsilon_0 = Z_0 = 0
\eea 
where $\rho \in [0,1]$, and may depend on $n$, and where $X^c$ and $Z$ are two continuous It\^{o}-martingales of the form 
\begin{eqnarray}
\label{XtcZt}
X_t^c & = & X_0 +   \int_0^t \sigma_{s/T}^M\sigma_s^X dW_s^X  \textnormal{ and } Z_t  =    \int_0^t\sigma_{s/T}^M \sigma_s^Z dW_s^Z, t \in [0,T],
\end{eqnarray}
where $\sigma^X$ and $\sigma^Z$ are \cadlag adapted processes, and $W^X$ and $W^Z$ are Brownian motions featuring possibly non-trivial high frequency correlation $d\langle W^X, W^Z \rangle_t = r_t dt$. Therefore, at time $t \in [0,T]$, and up to the multiplicative term $(\sigma_{t/T}^M)^2$, the two dimensional process $(X,Z)$ features a squared volatility equal to
$$ \Sigma_t = \left(\begin{matrix} (\sigma_t^X)^2 & r_t \sigma_t^X \sigma_t^Z\\ r_t \sigma_t^X \sigma_t^Z &  (\sigma_t^Z)^2\end{matrix}\right).$$
Finally, $\sigma^M$, which is the common deterministic volatility component, is a \caglad (left-continuous with right limits) function from $[0,1]$ to $\reels_+ - \{0\}$.  The component $\sigma^M$ can be interpreted as the market volatility (i.e. common to all the stocks), whereas $\sigma^X$ and $\sigma^Z$ correspond to the idiosyncratic component of volatility. For example, $\sigma^M$ can be a linear trend, while $\sigma^X$ and $\sigma^Z$ can both be a product of daily U-shape and random stochastic component such as Heston model. Further examples are considered in our finite sample analysis.

\begin{remark*}
The components $\sigma^X$ and $\sigma^Z$ may differ, will be assumed ergodic and typically account for the long time regularities (e.g. seasonality) of $X$ and $Z$. Having ergodic returns is in line with most of the literature on cointegration, see for instance \cite{phillips1990asymptotic}, or the more recent work of \cite{perron2016residuals}. On the other hand, $\sigma^M$ encompasses possibly non ergodic trends in volatility and is assumed to be a common factor in $X$ and $Y$. As far as we know, and as detailed in Section \ref{sectionTest}, if the non ergodic components differ in $X$ and $Y$, then constructing a test statistic which is numerically reliable and whose distribution is identifiable under the null of no cointegration remains an open and difficult question that we set aside in this paper. Note that adding such a non ergodic component scaled in time from 0 to $T$ is common practice in the literature on tests for unit root processes, as in \cite{cavaliere2005unit}, \cite{cavaliere2007testing} and more recently \cite{beare2018unit} among others. In our case, however, $\sigma^M$ is taken \cadlag whereas earlier works on unit root processes assumed the function to satisfy a Lipschitz condition except for, at most, a finite number of points of discontinuity. Note also that, as mentioned in the aforementioned papers, the fact that $\sigma^M$ is assumed deterministic can be easily relaxed to random and independent of the main filtration $\F$. Then all the convergences can be taken conditionally to $\sigma^M$. Finally, to our knowledge, there is no existing literature on a test of no cointegration which is robust to the presence of a common non ergodic volatility component. \\
\end{remark*}
Just as $\Delta X_i^c$ is the continuous time counterpart of $\Delta x_t$, $\Delta Z_i$ now plays the role of $u_t$ in (\ref{cointTimeSeries}). Note also that the presence of an intercept $c_0$ in the regression is just a convenient way to center the residual process $\epsilon$ without loss of generality. Moreover, as $\rho$ controls how close the residual is from a unit root process in (\ref{cointTimeSeries}), $\rho$ now controls how close $\epsilon$ is from an \ito-martingale in (\ref{cointContinuous}), with the two extreme cases being $\rho = 1$ where $\epsilon_i = Z_i$, and $\rho = 0$ where $\epsilon_i = \Delta Z_i$ for $i \in \{1,\ldots,n\}$. When $0 < \rho < 1$, $\epsilon_i$ lies somewhere between  an It\^{o}-martingale and an increment of an It\^{o}-martingale. 

\begin{remark*}
(limitation in terms of economic/statistical modeling) Note that as the model was designed with the (necessary and non straightforward) development of asymptotic theories in mind, there is (at least) an important side-effect feature in terms of economic/statistical modeling. If $0 \leq \rho < 1$, $\epsilon_i$ and $Y^c_i$ are $\Delta$ dependent. 
Highly related literature working under the cointegration error, $\epsilon_i$, being
independent of $\Delta$ includes \cite{bandi2014nonparametric}, \cite{kanaya2011nonparametric} and \cite{kim2018unit}.
\end{remark*}

In general, when $J^X \neq 0$ or $J^Y \neq 0$, it would be natural to simply replace $X^c$ and $Y^c$ by $X$ and $Y$ in (\ref{cointContinuous}), but it turns out that imposing such a constraint on the jump processes would imply that $\Delta J^Y = \alpha_0 \Delta J^X$. If the jumps are seen as structural breaks in the processes $X$ and $Y$, it is a very strong assumption which would require substantial support from empirical data. More importantly, letting $J^X$ and $J^Y$ free of constraint does not affect our strategy (and the related limit theory) for analyzing $X$ and $Y$, which will consist in first getting rid of the jump components using the truncation approach of \cite{mancini2009non} and then directly work with the estimated continuous components. Accordingly, for the sake of generality, we keep (\ref{cointContinuous}) even in the presence of jumps. The cointegration relation between $X$ and $Y$ yields for $i\in \{1,\ldots,n\}$ that
\bea \label{cointJumpContinuous}
Y_i = c_0 + J_i^Y - \alpha_0J^X_i + \alpha_0 X_i + \epsilon_i,
\eea 
which can be seen as cointegration (if $\rho <1$) with multiple \textit{level shifts}. Cointegration with breaks has been studied in \cite{gregory1996residual} (see Model 2) in the case of a single shift, and extended to the case of an arbitrary large (but known) number of deterministic breaks in \cite{maki2012tests}. In Equation (\ref{cointJumpContinuous}), the shifts are $\Delta J_s^Y - \alpha_0 \Delta J_s^X$ for $s \in [0,T]$, which may be in infinite number, are of random sizes, and can feature endogeneity.   

\medskip 
We now adapt the notion of no cointegration introduced in \cite{engle1987co} and discussed above to the case of driftless \ito-semimartingales.
\begin{definition*} \label{defNoCoint}(no cointegration)
Two \cadlag processes $X$ and $Y$ are said not cointegrated if any linear combination $Y - \alpha X$, $\alpha \in \reels$, is a driftless It\^{o}-semimartingale whose volatility component $\sigma$ is such that $\proba-\liminf_{T \to +\infty} T^{-1}\int_0^T \sigma_s^2ds > 0$.    
\end{definition*}

Definition \ref{defNoCoint} is thus a straightforward adaptation where time series have been replaced by \cadlag processes and unit root processes have been replaced by driftless \ito-semimartingales with non-trivial volatility components so that they are indeed explosive as $T \to +\infty$. Let us now get back to the model (\ref{cointContinuous}) and turn our attention to the description of different settings of $\rho$ and their impact on the relationship between $X$ and $Y$. Following the time series framework, $Y$ and $X$ are not cointegrated if $\rho = 1$ and, of course, if for any $x \in \reels^2 - \{0\}$, $\proba-\liminf_{T \to +\infty} T^{-1}x^T\int_0^T \Sigma_s^2dsx > 0$, since in that case, (\ref{cointContinuous}) reads for any $i \in \{1,\ldots,n\}$
\beas 
Y_i^c = c_0 + \alpha_0 X_i^c + Z_i.
\eeas 
We now turn our attention to the notion of cointegration. We follow again the time series case and say that $X$ and $Y$ satisfying (\ref{cointContinuous}) are cointegrated if $\rho \in [0,1)$ and does not depend on $n$. In that case, note that this implies the existence of $c_0, \alpha_0 \in \reels$ such that 

\beas 
Y_i^c = c_0 + \alpha_0 X_i^c + \epsilon_i,  
\eeas
where $\epsilon_i$ is not the value of an \ito-martingale (and is of one order of magnitude smaller). Finally, we will also consider the intermediary situation where $X$ and $Y$ are weakly cointegrated, which corresponds to the case where $\rho = 1 - \beta/n$ for some $\beta > 0$. Henceforth we will accordingly always assume that $X$ and $Y$ are generated according to one of the following setting:     
\begin{enumerate}[{(}i{)}]
    \item \textbf{Cointegration} $0 \leq \rho < 1 $ (independent of $n$).  
    \item \textbf{Weak cointegration} $\rho = 1- \beta/n$, with $\beta > 0$. 
    \item \textbf{No cointegration} $\rho = 1$. 
\end{enumerate}

We end this section with a brief remark on the connection between the above definition of cointegration and continuous-time mean reverting residuals. It sounds indeed reasonable to expect that if $X^c$ and $Y^c$ are such that for any $t \in [0,T]$
$$ Y_t^c = c_0 + \alpha_0 X_t^c + \cale_t $$
where $\cale$ is mean-reverting around $0$, i.e if for instance
$$ d\cale_t = - \theta \cale_t dt + \sigma^\cale(t)dW_t$$
for some $\theta >0$ and where $W$ is a standard Brownian motion, then $\cale$ remains of order $O_\proba(1)$ for any $t \in [0,T]$ and $Y^c$ and $X^c$ should be cointegrated to a certain degree. The next remark shows that this mean-reverting residual framework precisely corresponds asymptotically to the case of weak cointegration, where $X^c$ and $Y^c$ satisfy (\ref{cointContinuous}) for $\rho = 1-\beta/n$ with $\beta = \theta T$, and where the observation frequency $n \to +\infty$. 
\begin{remark*} \label{rmkOU} 
Assume that $X$ and $Y$ satisfy (\ref{cointContinuous}) with cointegration parameter $\rho = 1 - \theta T/n$ with $\theta >0$. Then, the interpolating residual 
\beas 
\epsilon_t = \epsilon_{i} \textnormal{ for } t \in [t_i, t_{i+1}) 
\eeas
is such that when $n \to +\infty$
$$ \epsilon \to^{u.c.p} \cale = \left(\int_0^{t} e^{-\theta (t-s)} dZ_s\right)_{t \in [0,T]} $$
where $\to^{u.c.p}$ stands for the uniform convergence in probability on any compact. Therefore, $\cale$ enjoys the mean-reverting dynamics
$$ d\cale_t = -\theta \cale_tdt + \sigma^Z(t) \sigma^M(t/T) dW_t^Z, t \in [0,T].$$
In particular, when $\sigma^Z(t) = \sigma^Z$ and $\sigma^M=1$, $\cale$ is an Ornstein-Uhlenbeck process.
\end{remark*}
\begin{proof}
We have the representation $\epsilon_t = \sum_{j=1}^n \rho^{\Delta^{-1}(t_i-t_j)} (Z_{t_i \wedge t_j} - Z_{t_i \wedge t_{j-1}} )$ for $t \in [t_i, t_{i+1})$. Therefore, the convergence towards $\cale$ is a direct consequence of the convergence $\rho^{ \Delta^{-1} s} = (1-\theta  T/n)^{ns/T}
\to e^{-\theta s}$  for any $s \in [0,T]$ along with Theorem I.4.31(iii) from \cite{jacod2013limit} (p. 47). 
\end{proof}

\section{Estimation of cointegrated systems}
\subsection{Construction of the estimator}
We now focus on estimating the couple $(\alpha_0, c_0)$ based on the discrete observations of $X$ and $Y$. Naturally, one can expect $(\alpha_0, c_0)$ to be identifiable only when $X$ and $Y$ are cointegrated. Indeed, we will see that when $\rho = 1$ (i.e. non cointegration case), our proposed estimator is inconsistent. Accordingly, any estimation of the cointegration parameters is untrustworthy without performing a test of no cointegration, question that we set aside in this section, and that is treated in Section \ref{sectionTest}. In any case, constructing the estimator does not require any knowledge whatsoever about the cointegration level $\rho \in [0,1]$. 

\smallskip
We first consider the case where $J^Y = J^X = 0$ and $\sigma^M = 1$. We adapt the classical OLS estimator proposed in \cite{engle1987co}, and resulting from (\ref{cointContinuous}) seen as a linear regression where $\epsilon$ is the noise process. Recall that $X$ and $Z$ may be correlated, so that the regression model induced by (\ref{cointContinuous}) features endogeneity. In particular, this rules out the alternative regression based on the high-frequency \textit{returns} of $X$ and $Y$
\bea \label{olsReturn} 
\Delta Y_i^c = c_0 + \alpha_0 \Delta X_i^c +  \Delta \epsilon_i, \textnormal{ with } \Delta \epsilon_i = (\rho-1) \epsilon_{i-1} + \Delta Z_i,
\eea 
because $\Delta Y_i^c$, $\Delta X_i^c$ and $\Delta \epsilon_i$ being of the same order $\sqrt{\Delta}$, the OLS estimator based on (\ref{olsReturn}) would be inconsistent due to non-zero correlation between $\Delta \epsilon_i$ and $\Delta X_i^c$. Conversely, the regression based on (\ref{cointContinuous}) is robust to endogeneity because $X_i$ and $Y_i$ are of order $\sqrt{T}$ whereas when $\rho < 1$, $\epsilon_i$ remains of order $\sqrt{\Delta}$.

\smallskip
In general, $X$ and $Y$ contain jumps and a non-constant common volatility component. Accordingly, we now estimate $(\alpha_0, c_0)$ in a three-step procedure consisting in first getting rid of those two features and then applying the aforementioned OLS estimation. At this point, in view of the representation (\ref{cointJumpContinuous}), it seems natural to adapt the methodology of \cite{gregory1996residual} and derive a modified OLS estimator which estimates and cancels the effect of the jumps seen as level shifts in (\ref{cointJumpContinuous}). However, the time required to run the break-robust OLS estimation greatly increases with the number of breaks. \cite{maki2012tests} proposed an alternative and less time-consuming method, but it also presents several drawbacks. First, the number of breaks (or at least an upper bound $k$) must be known. Second, the limit distribution of the test statistic depends on $k$, meaning that it may be necessary to calculate critical values if $k$ is larger than five, which is the highest value for which they have been reported. Finally, the method is supported by a numerical study only (again, for models with five breaks at most). Since we allow for a potentially high number of jumps in (\ref{cointJumpContinuous}), both approaches are inadapted, which is why we henceforth adopt the truncation method of \cite{mancini2009non}. As explained in the asymptotic theory section, it is independent of the number of shifts and robust to all the aforementioned features of the jumps, both for estimation and test. Accordingly, we remove the increments of $X$ and $Y$ such that \textit{at least} one of them is greater than a given threshold in absolute value. More precisely, for $U \in \{X,Y\}$, we compute the truncated process $\calt(U) = (\calt(U)_0,\ldots, \calt(U)_n)$ such that $\calt(U)_0 = U_0$, and for $i \in \{1,\ldots,n\}$
\beas 
\Delta \calt(U)_i = \Delta U_i \mathbb{1}_{\{ |\Delta X_i| \leq a 
\Delta^{\overline{\omega}}\} \cap \{ |\Delta Y_i| \leq a 
\Delta^{\overline{\omega}}\}},
\eeas 
that is
\beas
\calt(U)_i = U_0 + \sum_{j=1}^i \Delta U_j \mathbb{1}_{\{ |\Delta X_j| \leq a 
\Delta^{\overline{\omega}}\} \cap \{ |\Delta Y_j| \leq a 
\Delta^{\overline{\omega}}\}},
\eeas 
for some constant $a > 0$ and some exponent $\overline{\omega} \in (0,1/2)$ satisfying additional constraints stated in the next section. 
\begin{remark*} The practitioner should keep in mind that the choice of the threshold tuning parameters $(a,\overline{\omega})$ is very important in practice. On the one hand, if the threshold is too loose, then some small jumps will
not be properly removed. On the other hand, when the truncation is tight, some
returns will be mistakenly truncated away. From a theoretical perspective, both constraints are crucial, and the power of the test procedure will be badly affected while the cointegration estimator may become inconsistent if they are not satisfied. From a numerical perspective, however, it seems that not being able to remove jumps distorts the test procedure far more than removing too many returns. Details about the constraints that the truncation exponent $\overline{\omega}$ must satisfy can be found in Assumption $\textnormal{\textbf{[C]}}$. 
\end{remark*}

Second, we deflate the returns of both truncated processes $\calt(X)$ and $\calt(Y)$ by a consistent estimator $\sqrt{C_i}$ of $\sigma_{t_i-}^M$ up to some multiplicative constant. The procedure is similar to what was proposed in \cite{beare2018unit} for the case of a unit root process. Hereafter, we take $C_{i}$ as the standard local realized volatility on the truncated returns of $X$ (Using the returns of $Y$ would yield an estimator of $\sigma_{t_i-}^M$ up to a coefficient which depends on $\rho$). First, for two indices $0 \leq l < k < i$, $i\in \{1,\ldots,n\}$, we define 
\bea \label{defRVLocal}
RV_{i,k,l} = \sum_{j=(i-k) \vee 1}^{(i-l-1) \vee 1} \Delta X_j^2 \mathbb{1}_{\{ |\Delta X_j| \leq a \Delta^{\overline{\omega}} \}}, 
\eea 
where $a$ and $\overline{\omega}$ were defined before. Next, we take $k = [T^\gamma \Delta^{-1}]$, $l = [T^{\gamma'} \Delta^{-1}]$ (which both implicitly depend on $n$), where $[x]$ is the floor of $x$, and $0 < \gamma' < \gamma < 1$. The local window considered for (\ref{defRVLocal}) is thus such that the number of observations $k-l \to+\infty$ and at the same time the length of the window $(k-l)\Delta = o(T)$. Moreover, realized volatility is calculated over the interval $[t_{i-k},t_{i-l-1}]$ and not $[t_{i-k},t_{i-1}]$ in order to preserve the martingale structure of some transformations of $\epsilon$ when $\rho < 1$ and thus circumvent some technical difficulties that arise in the proofs. We then define for $i \in  \{1,\ldots,n\}$ 
		\beas
		C_{i} &=& T^{-\gamma}RV_{i,k,l} \textnormal{ if } RV_{i,k,l} > 0 \textnormal{ and } i > 2k , \\
		C_{i} &=&+\infty \textnormal{ otherwise. }   
		\eeas 

 Given that $l$ is negligible with respect to $k$, it does not affect the limit theory of $C_i$. We then compute the deflated version of $\calt(U) \in \{\calt(X),\calt(Y)\}$, $\calt( U)^{def}$ such that $\calt (U)_0^{def} = U_0$, and for $i \in \{1,\ldots,n\}$,
$$ \Delta  \calt (U)_i^{def} = \frac{\Delta \calt (U)_i}{\sqrt{C_{i}}},$$
that is
\beas 
\calt (U)_i^{def} = U_0 + \sum_{j=1}^i \frac{\Delta U_j}{\sqrt{C_{j}}} \mathbb{1}_{\{ |\Delta X_j| \leq a 
\Delta^{\overline{\omega}}\} \cap \{ |\Delta Y_j| \leq a 
\Delta^{\overline{\omega}}\}}.
\eeas
Key to our analysis is that both operations $\calt$ and '$def$' naturally preserve the cointegration relationship, in the sense that for any $i \in \{1,\ldots,n\}$
\bea \label{cointModified}
\calt(Y)_i^{def} = c_0 + \alpha_0 \calt(X)_i^{def} +  \calt(\epsilon)_i^{def}, 
\eea
with the new residual 
\beas 
\calt(\epsilon)_i^{def} =  \sum_{j=1}^i \frac{\Delta \epsilon_{j}}{\sqrt{C_j}} \mathbb{1}_{\{ |\Delta X_j| \leq a 
\Delta^{\overline{\omega}}\} \cap \{ |\Delta Y_j| \leq a 
\Delta^{\overline{\omega}}\}}. 
\eeas 
Finally, for two processes $A,B$, their associated OLS estimator is defined as
$$ \textnormal{OLS}[A,B] = \left(  \frac{\sum_{i=1}^n (B_i - \overline{B})(A_i - \overline{A})}{\sum_{i=1}^n (A_i - \overline{A})^2},\overline{B} -\frac{\sum_{i=1}^n (B_i - \overline{B})(A_i - \overline{A})}{\sum_{i=1}^n (A_i - \overline{A})^2}\overline{A}\right)$$
with $\overline{A} = n^{-1}\sum_{i=1}^n A_i$ and $\overline{B} = n^{-1}\sum_{i=1}^n B_i$. The general cointegration estimator is defined as
\begin{eqnarray}
\label{alphahatexplicitdef}
(\widehat{\alpha},\widehat{c})  = \textnormal{OLS}[\calt (X)^{def}, \calt (Y)^{def}].
\end{eqnarray}

\begin{remark*}
When $J^X = J^Y = 0$ and in the absence of truncation, $(\widehat{\alpha},\widehat{c})$ is the OLS estimator of a the linear transformation of $X$ and $Y$ where their respective returns have been multiplied by the weights $w_i = C_i^{-1/2}$. This is similar (yet not equal) to the GLS of \cite{kim2010cointegrating}, where the authors pre-estimate the time-varying variance of the noise $\epsilon$ by a standard OLS, and then construct the associated GLS cointegration estimator consisting in putting similar weights directly in front of $X_i$ and $Y_i$.   
\end{remark*}

\subsection{Assumptions and high-frequency framework}

We now proceed to give an asymptotic framework along with reasonable conditions under which the OLS estimator introduced in (\ref{alphahatexplicitdef}) is consistent (assuming, of course, $\rho <1$, i.e. cointegration). We also give a stronger setting on the jump processes which ensures a central limit theory for $(\widehat{\alpha},\widehat{c})$.
We will use the same framework when testing for no cointegration in the next section. Our first assumption specifies the high frequency asymptotics that is considered in this paper.\\

\par\noindent\textbf{Assumption }\textbf{[A]:} $n \to + \infty$, $T \to +\infty$, $\Delta = T/n \to 0$. \\

\par\noindent Such a double asymptotic is consistent with the high-frequency context ($\Delta \to 0$) and the fact that cointegration is a long-run phenomenon ($T \to +\infty$). Next, we assume that the volatility matrix $\Sigma$ has bounded moments up to some order $p_0$ and is ergodic.\\

\par\noindent\textbf{Assumption }\textbf{[B]:}
There exists $p_0 \geq 8$ such that $\sup_{t \in \reels_+} \esp \|\Sigma_t \|^{p_0} < +\infty$, where for a matrix $M$, $\|M\| := \sum_{i,j} |m_{i,j}|$. Moreover, there exists a positive definite matrix $\Omega = (\omega_{ij})_{1\leq i,j \leq 2}$ such that 
\bea  \epsilon(T) := \sup_{u\in [0,1]} \esp\left| \frac{1}{T} \int_{uT}^{T} \Sigma_tdt - (1-u)\Omega \right|^2  \to 0, \textnormal{ } T \to +\infty. \label{ergoSigma} \eea 
Moreover, $\sigma^X$ is asymptotically bounded from below with probability $1$, i.e. there exists $\underline{\sigma}^X >0$ such that $\proba - \liminf_{t \to + \infty} \sigma_t^X \geq \underline{\sigma}^X$.
\begin{remark*}
The definition of ergodicity stated in $\textnormal{\textbf{[B]}}$ is quite flexible. For instance, it encompasses most combinations of stationary ergodic processes and periodic processes. The asymptotic boundedness of $\sigma^X$ away from $0$ is assumed to avoid degenerate behaviors of the statistics due to the deflation operation. Similar long-run high-frequency asymptotics and ergodic settings can be found in the recent literature, see e.g. \cite{christensen2018diurnal} and \cite{andersen2019time} where the volatility process is the product of a stationary mixing component and a periodic component. Finally, note that $\Sigma_t$ may be correlated with $(W^X,W^Z)$ so that the Brownian integrals feature leverage effect.
\end{remark*}

\par\noindent Now we turn to our third assumption, which states conditions on the truncation parameters, and an additional condition on the relationship between $T$ and $n$.\\
\par\noindent\textbf{Assumption }\textbf{[C]:} We have $\frac{1}{4-r} < \overline{\omega} < \frac{1}{2} - \frac{3}{2p_0}$. Moreover, let $e_1 = \frac{1}{2\overline{\omega}(1-r)} \vee \frac{1}{\overline{\omega}(4-r)-1}$ if $r > 0$ and $e_1 > \frac{1}{4\overline{\omega}-1}$ if $r = 0$, and $e_2 = \frac{4}{p_0(1/2 -\overline{\omega})} \vee \frac{1}{p_0(1/2-\overline{\omega}) + 2\overline{\omega} - 3/2}$. We have $T^{1+e_1\vee e_2}n^{-1} \to 0$.  

\begin{remark*}
In particular, the second condition in $\textnormal{\textbf{[C]}}$ implies that $T$ must tend to infinity slowly enough compared to $n$. In the case where $\Sigma_t$ admits moments of any order  ($p_0 = +\infty$), we note that Condition $\textnormal{\textbf{[C]}}$ can be simplified as $1/(4-r) < \overline{\omega} < 1/2$, and taking $\overline{\omega}$ arbitrarily close to $1/2$, the condition on $T$ becomes $T^{1+ \eta + \frac{1}{1-r}}n^{-1} \to 0$ for $\eta >0$ arbitrarily small. Finally, note that for jumps of finite activity $r=0$, this can be further simplified as $T^{2+ \eta} n^{-1} \to 0$ which is stronger than the condition $\Delta \to 0$ stated in Assumption $\textnormal{\textbf{[A]}}$. Of course, if $J^X = J^Y = 0$ and the truncation step is ignored in the estimators, all the stated results hold with $e_1 = 0$ in $\textnormal{\textbf{[C]}}$. If moreover $p_0 = +\infty$ then $e_2 = 0$ too and $\textnormal{\textbf{[C]}}$ reduces to $\Delta \to 0$.   
\end{remark*}
Finally, we state an additional more restrictive assumption on the jumps, that we will use only to derive a central limit theorem under cointegration for $(\widehat{\alpha}, \widehat{c})$ (Theorem \ref{thmH1}) and an equivalent of the Dickey-Fuller statistics under the alternative of cointegration (Proposition \ref{propH1}). It plays no role in the derivation of the consistency of the OLS and of the Dickey-Fuller test under any of the hypotheses.  \\
\par\noindent\textbf{Assumption }\textbf{[D]:}  $J^X$ and $J^Y$ are two sequences of jump processes such that for $U \in \{X,Y\}$, $\sup_{j \in \{1,\ldots,n\}, u \in [0,1]} |J_{t_j \wedge uT_n}^U-J_{t_{j-1} \wedge uT_n}^U| = O_\proba(\Delta^{1/2})$, and $T^{-1} \sum_{0<s \leq T} |\Delta J_s^Y - \alpha_0 \Delta J_s^X| = o_\proba(T^{-1/2}n^{-1})$.    \\
\par\noindent The above assumption states that the jump processes are asymptotically negligible in the regression, and satisfy an additional cointegration condition. Hereafter, we will always implicitly assume that \textbf{[A]}-\textbf{[C]} hold. As for \textbf{[D]}, we will explicitly state whether it is assumed or not. 

\subsection{Asymptotic theory of the OLS estimator under cointegration} \label{sectionLimitTheoryOLS}

We now give the asymptotic properties of $C_i$, and of the OLS estimator under the cointegration regime $\rho < 1$. We start with $C_i$. Since $0 < \gamma < 1$, note that the local time window $T^{\gamma} \to +\infty$ so that the ergodic theory for $\Sigma$ kicks in, and at the same time the scaled time window $T^{\gamma-1} \to 0$ so that for $t \in [t_{i-k_n}, t_{i-l_n}]$, $\sigma_{t/T}^M \approx \sigma_{(t_i/T)-}^M$ by left continuity of $\sigma^M$. 

\begin{proposition*} \label{lemmaCiWeak}
For any $i \in \{1,\ldots,n\}$, when $n \to +\infty$,
$$\esp |C_i - (\sigma_{(t_{i}/T)-}^M)^2 \omega_{11}|^2 \to 0.$$
\end{proposition*}
As proved in the Appendix (see Lemma \ref{lemmaRVLocal}), it turns out that the above $\mathbb{L}^2$ convergence is even uniform outside a set of indices whose cardinality is negligible with respect to $n$. A full uniformity as in Lemma 4.1 of \cite{beare2018unit} is impossible here because $\sigma^M$ may have jumps (whereas in the aforementioned paper $\sigma^M$ is assumed differentiable). We now focus on the asymptotic properties of $(\widehat{\alpha},\widehat{c})$ when $X$ and $Y$ are cointegrated. In what follows, we let $W = (W^1,W^2)$ be a standard Brownian motion on $[0,1]$. Moreover, defining
 \bea \label{defL} 
 L = \l(\begin{matrix} \sqrt{\omega_{11} } &0\\ \omega_{11}^{-1/2}\omega_{12} & \sqrt{\omega_{22} - \omega_{11}^{-1}\omega_{12}^2}  \end{matrix}\r) = \l(\begin{matrix} \sqrt{\omega_{11} } &0\\ \sqrt{\omega_{22}} r_\infty & \sqrt{\omega_{22}} \sqrt{1-r_\infty^2} \end{matrix} \r),
 \eea 
with $r_\infty = \omega_{12}/\sqrt{\omega_{11} \omega_{22}}$, we also let $B = (B^1,B^2) = L W$ be a two dimensional Brownian motion on $[0,1]$, with covariance matrix $L L^T = \Omega$. Below, for a process $V$ on $[0,1]$, we let $\overline{V} = \int_0^1 V_udu$. 

\begin{theorem*} \label{thmH1}
Assume that $X$ and $Y$ are cointegrated, that is $\rho < 1$ and is fixed. Assume further that $1/2 \leq \gamma < 1$ and $0< \gamma^{'} < \gamma$. Then we have 
$$ \widehat{\alpha} - \alpha_0\to^\proba0  \textnormal{ and } T^{-1/2}(\widehat{c} -c_0) \to^\proba 0.$$
Moreover, under the additional assumption $\textnormal{\textbf{[D]}}$, we have the convergence in distribution
\beas
\l(\begin{matrix} n(\widehat{\alpha} - \alpha_0) \\ nT^{-1/2}(\widehat{c} - c_0) \end{matrix}\r) \to^d \frac{1}{1-\rho}\l(\begin{matrix} \frac{\omega_{12} +\int_0^1(B_s^1-\overline{B}^1) dB_s^2}{\int_0^1 (B_s^1 - \overline{B}^1)^2ds} \\ \omega_{11}^{-1/2}B_1^2 - \frac{\omega_{12} +\int_0^1(B_s^1-\overline{B}^1) dB_s^2}{\int_0^1 (B_s^1 - \overline{B}^1)^2ds}\omega_{11}^{-1/2}\overline{B}^1  \end{matrix}\r).
\eeas 
\end{theorem*}

The fast rate $n$ for the estimation of $\alpha_0$ is in line with the literature on cointegration estimation, see for instance Proposition 1 of \cite{engle1987co} and Theorem 7 in \cite{kim2010cointegrating}. Similarly, the fact that $X_T^c$ and $Y_T^c$ are of order $T^{1/2}$ implies that one can consistently estimate only $T^{-1/2}c_0$, with the same rate $n$ as for $\alpha_0$. Note also that in the exogenous residual case $\omega_{12} = 0$, the limit distribution for $\widehat{\alpha}$ corresponds to the one of a classical OLS on an homoskedastic cointegration regression as in Lemma 2.1 of \cite{phillips1988asymptotic} up to the mean terms $\overline{W}^1$, due to the presence of the constant $c_0$ in our regression. In other words, the jumps and the heteroskedasticity coming from $\sigma^M$ do not impact the limit theory of $\widehat{\alpha}$ and $\widehat{c}$, and no efficiency is lost due to the truncation and the deflation. Finally, note that, as explained in the following remark, the above limit distribution is actually mixed normal. 

\begin{remark*} \label{rmkThmH1}
Rewriting $B$ as $LW$ and conditioning on $W^1$, we can specify the above limit as the following mixed normal distribution. Defining $I[W^1] = \int_0^1 (W_s^1 - \overline{W}^1)dW_s^1 = (W_1^1)^2/2-1/2-(\overline{W}^1)W_1^1$, $J[W^1] = \int_0^1 (W_s^1-\overline{W}^1)^2ds$, $K[W^1] = \int_0^1 (W_s^1)^2ds$, $r_\infty=\omega_{12}/\sqrt{\omega_{11}\omega_{22}}$, and $v_\epsilon := \frac{\omega_{22}}{\omega_{11}(1-\rho)^2}$, we have 
\beas
\l(\begin{matrix} n(\widehat{\alpha} - \alpha_0) \\ nT^{-1/2}(\widehat{c} - c_0) \end{matrix}\r) \to^d \sqrt{v_\epsilon} \calm\caln \l(  r_\infty \mathcal{B}[W^1], \frac{ (1 - r_\infty^2)}{ J[W^1]} \mathcal{V}[W^1]\r)
\eeas 
where 
$$ \mathcal{B}[W^1] = \l(\begin{matrix}  \frac{1+I[W^1]}{ J[W^1]} \\ W_1^1 - \frac{1+ I[W^1]}{J[W^1]} \overline{W}^1  \end{matrix}\r) \textnormal{ and } \mathcal{V}[W^1] = \l(\begin{matrix} 1  &   \overline{W}^1  \\  \overline{W}^1   &  K[W^1]   \end{matrix}\r).$$ 
Moreover, when $B^1$ and $B^2$ are independent, $r_\infty=0$, and the above limit becomes
\beas
\l(\begin{matrix} n(\widehat{\alpha} - \alpha_0) \\ nT^{-1/2}(\widehat{c} - c_0) \end{matrix}\r)\to^d \sqrt{\frac{v_\epsilon}{ J[W^1]}}\calm\caln \l(0,  \l(\begin{matrix} 1  &   \overline{W}^1  \\  \overline{W}^1   &  K[W^1]   \end{matrix}\r)\r).
\eeas 
\end{remark*}

Remark \ref{rmkThmH1} suggests that it is possible to construct a studentized version of Theorem \ref{thmH1}, essential for the computation of confidence intervals and significance tests. To do so, we need to estimate the different quantities appearing in the bias and the variance of the mixed normal distribution. We construct the estimated residuals 
\bea \label{defResiduals}
\widehat{\epsilon}_i =  \calt (Y)_i^{def} -  \widehat{c} - \widehat{\alpha} \calt (X)_i^{def}, 
\eea 
and then estimate $v_\epsilon$, $\rho$, and $r_\infty$ as follows.
\bea
\widehat{v}_\epsilon = T^{-1} \sum_{i=1}^n \widehat{\epsilon}_i^2,
\eea

\bea \label{estimRho}
\widehat{\rho} = \frac{\sum_{i=2}^{n} \Delta \widehat{\epsilon}_i \widehat{\epsilon}_{i-2}}{\sum_{i=1}^{n} \Delta \widehat{\epsilon}_i \widehat{\epsilon}_{i-1}},
\eea 

\bea 
\widehat{r}_\infty = \frac{\sum_{i=2}^n \Delta \calt(X)_i^{def} (\widehat{\epsilon}_i - \widehat{\rho} \widehat{\epsilon}_{i-1})}{\sqrt{\sum_{i=2}^n (\widehat{\epsilon}_i - \widehat{\rho} \widehat{\epsilon}_{i-1})^2}}.
\eea 
Note that for the estimation of $\rho$, we have preferred the formula of (\ref{estimRho}) over the more classical estimator $\tilde{\rho} = \sum_{i=1}^n\widehat{\epsilon}_i \widehat{\epsilon}_{i-1}/\sum_{i=1}^n \widehat{\epsilon}_i^2$, because the former is robust to truncation under $\textbf{[A]-[C]}$ whereas it is not theoretically clear whether the latter is consistent without Assumption \textbf{[D]}. We also substitute $\calt(X)^{def}$ to all the quantities involving $W^1$ appearing in the central limit theorem. That is, letting
\beas 
\overline{\calt(X)}^{def} &=& T^{-1/2}n^{-1} \sum_{i=1}^n \calt(X)_i^{def}\\
I[\calt(X)^{def}]&=& \frac{(\calt(X)_n^{def})^2 - T}{2T}  - T^{-1/2}\overline{\calt(X)}^{def}  \calt(X)_n^{def}\\
J[\calt(X)^{def}]&=& n^{-1} \sum_{i=1}^n (T^{-1/2}\calt(X)_i^{def} - \overline{\calt(X)}^{def})^2\\ 
K[\calt(X)^{def}] &=& J[\calt(X)^{def}] + (\overline{\calt(X)}^{def} )^2,
\eeas
we introduce the estimators for the asymptotic biases and variances of $\widehat{\alpha}$ and $\widehat{c}$:
\beas  
B_{\widehat{\alpha}} = n^{-1} \sqrt{\widehat{v}_\epsilon} \widehat{r}_\infty\frac{1+I[\calt(X)^{def}]}{J[\calt(X)^{def}]}\textnormal{ and } B_{\widehat{c}} = n^{-1}T^{1/2}\sqrt{\widehat{v}_\epsilon} \widehat{r}_\infty \l(T^{-1/2}\calt(X)_n^{def} - \frac{1+I[\calt(X)^{def}]}{J[\calt(X)^{def}]} \overline{\calt(X)}^{def}\r), 
\eeas
and
\beas 
V_{\widehat{\alpha}} = \widehat{v}_\epsilon\frac{1-\widehat{r}_\infty^2}{J[\calt(X)^{def}]} \textnormal{ and } V_{\widehat{c}} = \widehat{v}_\epsilon\frac{(1-\widehat{r}_\infty^2) K[\calt(X)^{def}]}{J[\calt(X)^{def}]}.
\eeas 
we have the following studentized version of the central limit theory for $\widehat{\alpha}$ and $\widehat{c}$. 
\begin{proposition*}\label{propStud}
Assume that $X$ and $Y$ are cointegrated, that is $\rho < 1$ and is fixed. Assume that $1/2 \leq \gamma < 1$ and $0< \gamma^{'} < \gamma$. We have 
$$ \widehat{\rho} \to^\proba \rho.$$
Moreover, if we assume $\textnormal{\textbf{[D]}}$, then we also have 
$$ (\widehat{v}_\epsilon, \  \widehat{r}_\infty) \to^\proba (v_\epsilon,  r_\infty),$$
\beas 
\frac{n(\widehat{\alpha} - \alpha_0 - B_{\widehat{\alpha}})}{\sqrt{V_{\widehat{\alpha}}}} \to^d \caln(0,1),
\eeas 
and
\beas 
\frac{nT^{-1/2}(\widehat{c} - c_0 - B_{\widehat{c}})}{\sqrt{V_{\widehat{c}}}} \to^d \caln(0,1).
\eeas 
\end{proposition*}

\section{Residual based test for cointegration} \label{sectionTest}

\subsection{Construction of the test and limit theory} \label{sectionConstructionTest}
It is perhaps of even more crucial importance to infer from the data whether $X$ and $Y$ are cointegrated or not in the first place prior to analyzing the estimated cointegration coefficients. Consequently, we now give a test for the null hypothesis that there is no cointegration against the alternative of cointegration. We let  

\begin{eqnarray*}
\calh_0 & : & \rho = 1.\\
\calh_1 & : & 0 \leq \rho < 1 \textnormal{ (independent of }n \textnormal{)}.
\end{eqnarray*}
Recall that both hypotheses induce respectively the following models on the continuous parts of $X$ and $Y$: 
\begin{eqnarray*}
\calh_0 & : & Y_i^c = c_0 + \alpha_0 X_i^c + Z_i  \label{H0}\\
\calh_1 & : & Y_i^c = c_0+\alpha_0 X_i^c + \epsilon_i \textnormal{ with } \epsilon_i = \rho \epsilon_{i-1} + \Delta Z_i,\textnormal{  } \rho < 1, \label{H1}
\end{eqnarray*}
and that the parameter $\rho$ controls how far $\calh_1$ is from $\calh_0$. As it is standard in the literature on tests for unit roots and cointegration (see e.g. \cite{pesavento2004analytical}, \cite{beare2018unit}) and useful to derive the local power of our test, we embed $\calh_0$ in the family of local alternatives $\widetilde{\calh}_1^{n,\beta}$ defined as
\bea 
\widetilde{\calh}_1^{n, \beta} & : & \rho = 1 - \frac{\beta}{n}, \textnormal{ with } \beta \geq 0, \label{H1local}
\eea
which implies the following model on the continuous parts of $X$ and $Y$,
\beas 
\widetilde{\calh}_1^{n, \beta} & : & Y_i^c = c_0+\alpha_0 X_i^c + \epsilon_i \textnormal{ with } \epsilon_i = \rho \epsilon_{i-1} + \Delta Z_i,\textnormal{  } \rho = 1 - \frac{\beta}{n}, \textnormal{ and } \beta \geq 0, 
\eeas
that corresponds to the notion of weak cointegration introduced at the end of Section \ref{sectionCointModel} when $\beta >0$, and is simply $\calh_0$ when $\beta = 0$. The canonical test in the unit root literature is the so-called Dickey-Fuller test on residuals of \cite{dickey1981likelihood}. It has been extended to many directions, such as, among others, the augmented Dickey-Fuller (ADF) test, robust to a residual following an AR(p) specification, and the $Z_t$ and $Z_\alpha$ tests of \cite{phillips1987time}, robust to autocorrelated returns under the null hypothesis of a unit root process. These tests have been later adapted to cointegration, and their asymptotic properties derived in \cite{phillips1990asymptotic}. In the present work, we choose to focus on the DF approach, performed on the estimated residuals resulting from the OLS estimation of (\ref{alphahatexplicitdef}). Before we state the main result of this section, we briefly recall the construction of the test statistic. Recall that the estimated residuals are defined for $i \in \{1,\ldots,n\}$ as 
\beas 
\widehat{\epsilon}_i =  \calt (Y)_i^{def} -  \widehat{c}- \widehat{\alpha} \calt (X)_i^{def}. 
\eeas 
 
Then, the associated DF statistic $\Psi$ is the $t$-statistic of the coefficient $\phi$ in the linear regression
$$ \Delta \widehat{\epsilon}_i = \phi \widehat{\epsilon}_{i-1} + \eta_i, \textnormal{ } i \in \{1,\ldots,n\},$$
that is, first estimating $\phi$ with
$$ \widehat{\phi} = \frac{\sum_{i=1}^n\Delta \widehat{\epsilon}_i \widehat{\epsilon}_{i-1}}{\sum_{i=1}^n \widehat{\epsilon}_i^2},$$
and estimating the standard deviation of $\widehat{\phi}$ with
$$s_{\widehat{\phi}} = \sqrt{\frac{n^{-1} \sum_{i=1}^n (\Delta \widehat{\epsilon}_i - \widehat{\phi} \widehat{\epsilon}_{i-1})^2}{\sum_{i=1}^n \widehat{\epsilon}_{i-1}^2}},$$
$\Psi$ is defined as
\bea \label{defPsi}
\Psi = \frac{\widehat{\phi}}{s_{\widehat{\phi}}}.
\eea 
We now proceed to derive the asymptotic distribution of $\Psi$ under $\widetilde{\calh}_1^{n,\beta}$, for any $\beta \geq 0$. In particular, recall that $\calh_0$ is covered by Theorem \ref{thmH0DF} below, since $\calh_0 = \widetilde{\calh}_1^{n,0}$. We need to define a few quantities before we state the main result. As in the previous section, we consider $W = (W^1,W^2)$ a standard Brownian motion on $[0,1]$, and we define the two dimensional process on $[0,1]$ and for $\beta \geq 0$
$$ J(\beta)_u = \int_0^u e^{-\beta(u-s)}\sigma_s^M dW_s, \textnormal{ }u\in [0,1].$$
Next, letting $\lambda = (r_\infty/\sqrt{1-r_\infty^2},1)^T$, we consider 
$$ \xi(\beta)_u = W^2 -\beta\int_0^u (\sigma_s^M)^{-1} \lambda^TJ(\beta)_sds, \textnormal{ }u\in [0,1]$$
and finally 
$$ H(\beta) = (W^1 - \overline{W}^1, \xi(\beta) - \overline{\xi(\beta)}),$$
where we recall that for a process $(V_u)_{u \in [0,1]}$, $\overline{V} = \int_0^1 V_udu$. Note that under $\calh_0$, $\beta=0$ and $H(0)$ is simply $W-\overline{W}$. We finally introduce 
$$ \kappa(\beta) := \left(\frac{\int_0^1 H(\beta)_u^1 H(\beta)_u^2du}{\int_0^1 (H(\beta)_u^1)^2du},-1\right)^T.$$
The next proposition shows that the OLS estimator (and therefore the associated estimated residual process) is inconsistent under  $\widetilde{\calh}_{1}^{n,\beta}$.
\begin{proposition*} \label{propOLSH0}

Let $ \beta \geq 0$. Let $L_{\alpha_0} = \sqrt{\frac{ \omega_{22}}{\omega_{11}}}\l( \sqrt{1-r_\infty^2}, r_\infty\r)^T$, and $L_{c_0} = -\sqrt{\frac{\omega_{22}}{\omega_{11}}}\l(\overline{W}^1, \overline{\xi(\beta)}\r)^T$. Then, under $\widetilde{\calh}_1^{n,\beta}$, we have the joint convergences
$$ \widehat{\alpha} - \alpha_0 \to^d  L_{\alpha_0}^T \kappa(\beta),$$
and
$$ T^{-1/2}(\widehat{c} - c_0) \to^d L_{c_0}^T \kappa(\beta).$$
\end{proposition*}

We are now ready to state the limit distribution of $\Psi$ under any local alternative $\widetilde{\calh}_{1}^{n,\beta}$. Define 
$$ Q(\beta) = \kappa(\beta)^T H(\beta) = \frac{\int_0^1 H(\beta)_u^1 H(\beta)_u^2du}{\int_0^1 (H(\beta)_u^1)^2du} H(\beta)^1 - H(\beta)^2.$$
\begin{theorem*} 
\label{thmH0DF}
Let $ \beta \geq 0$. Under $\widetilde{\calh}_1^{n,\beta}$,

$$ \Psi \to^d \frac{\int_0^1 Q(\beta)_s dQ(\beta)_s}{\sqrt{\kappa(\beta)^T \kappa(\beta)\int_0^1 Q(\beta)_s^2ds} }.$$
In particular, under $\calh_0$, we have 
$$ \Psi \to^d \frac{\int_0^1 Q_s dQ_s}{\sqrt{\kappa^T \kappa\int_0^1 Q_s^2ds} }$$
where 
$$ \kappa:= \kappa(0)=  \left(\frac{\int_0^1 (W_u^1 - \overline{W}^1)  (W_u^2 - \overline{W}^2)du}{\int_0^1 (W_u^1 - \overline{W}^1)^2du},-1\right)^T$$
and
$$ Q := Q(0) = \kappa^T(W-\overline{W}).$$
\end{theorem*}
The next proposition gives the behavior of $\Psi$ (and proves the consistency of the test) under $\calh_1$, and provides an equivalent of the statistics under the stronger assumption $\textnormal{\textbf{[D]}}$.  
\begin{proposition*} \label{propH1} 
Under $\calh_1$, we have
$$  \Psi \to^\proba - \infty.$$
Moreover, under $\textnormal{\textbf{[D]}}$, we have 
$$\Psi \sim^\proba - n^{1/2} \sqrt{\frac{1-\rho}{1+\rho}}.$$
\end{proposition*}
When $\beta = 0$, the limit distribution of $\Psi$ in Theorem \ref{thmH0DF} is the same as the one of the ADF statistic in Theorem 4.2 of \cite{phillips1990asymptotic}, up to the mean component $\overline{W}$ coming from the fact that an intercept is present in the regression. Therefore, under $\calh_0$, as in the previous section, the truncation and the deflation completely cancel the impact of jumps and that of the non ergodic volatility $\sigma^M$ that may affect $\Psi$. Moreover, in Proposition \ref{propH1}, the divergence rate of $\Psi$ also corresponds to the standard one (see Theorem 5.1 in \cite{phillips1990asymptotic}). However, under the local alternative $\widetilde{\calh}_1^{n,\beta}$ with $\beta >0$, the limit of $\Psi$ depends on the shape of $\sigma^M$, so that the local power of the test may be affected by a non ergodic volatility component. This feature was already present in the time-varying variance robust unit root tests of \cite{beare2018unit}. More importantly, without jumps and if $\sigma^M = 1$, a careful examination of Theorem 1 in \cite{pesavento2004analytical} shows that the limit distribution of the standard DF and our own modified test coincide for any $\beta \geq 0$: no local power is lost when applying the truncation and the deflation even in the absence of those features. Finally, as a direct corollary of Theorem \ref{thmH0DF} and Proposition \ref{propH1}, we conclude this section with the consistency of the modified DF test.
\begin{corollary*}
Let $\delta \in (0,1)$ and $q_\delta$ be the $\delta$-quantile of $\int_0^1 Q_sdQ_s/\sqrt{\kappa^T\kappa \int_0^1 Q_s^2ds}$. Then the test statistic $\Psi$ satisfies
$$ \proba(\Psi < q_\delta | \calh_0) \to \delta \textnormal{ and } \proba(\Psi <q_\delta | \calh_1 ) \to 1.$$
\end{corollary*}

\subsection{Testing for cointegration with drifting \ito-semimartingales} \label{SectionDrift}

We now examine how the testing procedure can be adapted if the processes $X^c$ and $Z$ feature drift terms. We only partially address the problem, and restrict ourselves to the simple case of linear trends. Dealing simultaneously with general drifts, even ergodic ones, and a non ergodic volatility component is a difficult matter (at least to us) that we set aside in this work. As a matter of fact, we show in this section that even with linear drifts a natural adaptation of our DF statistic already yields a complex limit distribution that depends on the curve $u \to \sigma_u^M$ even under the null hypothesis (so that critical values must be estimated everytime the test is run). The new model for $X$ and $Y$ is   

\begin{eqnarray*}
X_t^c & = & X_0 +  b^X t +  \int_0^t \sigma_{s/T}^M\sigma_s^X dW_s^X  \textnormal{ and } Z_t  =  b^Z t  + \int_0^t\sigma_{s/T}^M \sigma_s^Z dW_s^Z, t \in [0,T],
\end{eqnarray*}
 where $b^X,b^Z \in \reels$. 

The testing procedure can be modified as follows to be drift robust. We first truncate the returns, then detrend the processes by subracting to each return the quantity $\cald(U) = n^{-1}  (\calt(U)_n - \calt(U)_0)$, and finally deflate by $\sqrt{C_i}$ each truncated and detrented return. This yields the new process for $U \in \{X,Y\}$
\beas 
\breve{\calt}(U)_i^{def} = U_0 + \sum_{j=1}^i C_j^{-1/2}(\Delta U_j \mathbb{1}_{\{ |\Delta X_j| \leq a 
\Delta^{\overline{\omega}}\} \cap \{ |\Delta Y_j| \leq a 
\Delta^{\overline{\omega}}\}} - \cald(U)).
\eeas 
Next, we apply the testing procedure of the previous section to $\breve{\calt}(X)_i^{def}$ and $\breve{\calt}(Y)_i^{def}$ in lieu of $\calt(X)_i^{def}$ and $\calt(Y)_i^{def}$. We denote by $\breve{\Psi}$ the associated statistic. In the following theorem, for $(V_u)_{u \in[0,1]}$ a process, we define $(\breve{V})_{u \in [0,1]}$ such that for $u\in[0,1]$, 

$$\breve{V}_u = \int_0^1 V_sds + \l(\int_0^1 (\sigma_s^M)^{-1} sds - \int_{u}^1 (\sigma_s^M)^{-1} ds\r) \int_0^1 \sigma_s^M dV_s$$
whenever the integrals make sense.
\begin{theorem*} \label{thmRobustDrift}
Let $ \beta \geq 0$. Under $\widetilde{\calh}_1^{n,\beta}$, $\breve{\Psi}$ converges to the same limit as $\Psi$ in Theorem \ref{thmH0DF} except that 
$\overline{W}$ and $\overline{\xi(\beta)}$ are respectively replaced by $\breve{W}$ and $\breve{\xi}(\beta)$ in the expression of $H(\beta)$, $\kappa(\beta)$ and $Q(\beta)$.
\end{theorem*}
In particular, note that even in the case of a constant drift, the limit distribution of $\breve{\Psi}$ now depends on $\sigma^M$ even under $\calh_0$. Therefore, one needs to estimate the curve of $\sigma^M$ and then compute the related critical values by, for instance, Monte-Carlo simulations. This sheds light on the lack of applicability of the above procedure, and also indicates that dealing with a drift and time-varying volatility at the same time is a complex procedure. 

\smallskip
Since the drift seems to be having a negligible impact in our numerical studies, at least in a realistic model and when it is calibrated to values usually encountered in empirical data, it seems more reasonable to use the simpler statistic $\Psi$ whose critical values are known and independent of the model at hand. Correspondingly, we will focus entirely on that statistic in the following finite sample experiment. In particular, we will not implement Monte-Carlo method to preestimate the curve of $\sigma^M$.

\section{Finite sample}

In this section, we conduct a Monte Carlo experiment in two steps. First, we investigate that the deflated and truncated based OLS method to estimate the cointegrated relations performs reasonably well, and that it outperforms the classical OLS procedure in a general model incorporating all the features of high frequency data in case of cointegration. Very related to estimation of relation methods is that of autocorrelation of the residuals' level, which we also look at. Second, we examine the size and power properties of the modified Dickey-Fuller residual based test for the null of no cointegration against the alternative of cointegration. In addition, we explore how the modified test performs relative to four standard residual based tests from the literature on cointegration in a variety of models, and more specifically in the presence of which feature the new test outperforms the standard procedures.

\subsection{Setup}
Overall, eight different models are generated. An overview is reported on Table \ref{tableoverview}. One model (i.e. Model 8) is general and includes all the aforementioned features of high frequency data, whereas each remaining model (i.e. Model 1-7) includes one specific feature. In what follows, and for the sake of brevity, we sharpen our focus on Model 3, Model 7 and Model 8. Additional tables and comments related to the other models can be found in the Appendix. We simulate $M=1,000$ Monte Carlo paths of high-frequency returns, where each path consists of $T=2$ years of generated returns. A year is divided into 252 working days, each of them being set to 6.5 hours of trading activity, i.e. 23,400 seconds. Each path is simulated via an Euler scheme with related step set to 10 seconds.
\subsubsection{Sampling gap}
In accordance with our empirical examples, we consider the gap between two observations $\Delta$ ranging from 10 minutes, i.e. 600 seconds which sets the number of observations to $n=19,656$ across the two simulated years, to 2 days, yielding $n=252$ observations. With one observation every 10 minutes, we are enough in the high frequency regime so that the limit theory related to the truncation method kicks in, but not too much into it so that we prevent as much as possible from market microstructure effects. When the gap is two days, this is purely low frequency setting.

\subsubsection{Simulation mechanism}
We simulate $X_t^c$ and $Z_t$ as:
\begin{eqnarray*}
X_t^c & = & X_0 + \int_0^t b_t^X dt +  \int_0^t \sigma_{s/T}^M\sigma_s^X dW_s^X  \textnormal{ and } Z_t  = \int_0^t b_t^Z dt + \int_0^t\sigma_{s/T}^M \sigma_s^Z dW_s^Z, t \in [0,T],
\end{eqnarray*}
where the correlation between $W_t^X$ and $W_t^Z$ is set to $\overline{\rho} = 0.2$, i.e. $d\langle W^X, W^Z \rangle_t = \overline{\rho} dt$. Here, contrary to the theoretical setting in (\ref{XtcZt}), the two processes can incorporate non-zero drifts which are set to $b_t^X = 0.03 (1 + W_t^{X,b})$, and $b_t^Z = 0.02 (1+ W_t^{Z,b})$, where $W_t^{X,b}$ and $W_t^{Z,b}$ are two independent Brownian motions. Depending on the model at hands, the market volatility can be constant, linear, or including 1 jump, and may respectively take the following forms:
\begin{eqnarray}
\sigma_t^M & = &\tilde{\sigma},\\
\sigma_t^M & = &\tilde{\sigma}(1 - 3t/4),\\
\sigma_t^M & = &\tilde{\sigma}(\mathbf{1}_{\{x<0.2\}} + 1/3 \times \mathbf{1}_{\{x\geq0.2\}}),
\end{eqnarray}
where we fix $\tilde{\sigma} = \sqrt{0.1}$. For $V \in 
\{ X,Z\}$, the idiosyncratic component of the volatility is split into a U-shape intraday seasonality component and Heston model with jumps specified as
\begin{eqnarray*}
\sigma_t^V &=& \sigma_{t-,U}^V \sigma_{t,SV}^V,
\end{eqnarray*}
where 
\begin{eqnarray*} 
\sigma_{t,U}^V & = & C + Ae^{-at/T} + De^{-c(1-t/T)} + J_t^{\sigma,V},\\
d(\sigma_{t,SV}^{V})^2 & = & \alpha(\bar{\sigma}^2 - (\sigma_{t,SV}^V)^2)dt + \delta \sigma_{t,SV}^Vd\bar{W}_t^V, 
\end{eqnarray*}
with $C=0.75$, $A=0.25$, $D=0.89$, $a=10$, $c=10$, the volatility jump process is defined as $dJ_t^{\sigma,V} = M_t^{\sigma,V} S_t^{\sigma,V} dN_t^{\sigma,V}$, where the volatility jump magnitude $M_t^{\sigma,V}$ is distributed as $\mathcal{N}(0.5,0.1)$, the signs of the jumps  $S_t^{\sigma,V} = \pm 1$ are i.i.d symmetric, $N_t^{\sigma,V}$ is a homogeneous Poisson process with parameter $\bar{\lambda} = 10T/252$ (with that setting volatility jumps occur randomly on average ten times a year), $\alpha = 5$, $\bar{\sigma}^2 = 1$, $\delta = 0.4$, $\bar{W}_t^V$ is a standard Brownian motion correlated to $W^V$ with $d\langle W^V,\bar{W}^V \rangle_t = \overline{\phi} dt$, $\overline{\phi} = -0.75$, $(\sigma_{0,SV}^V)^2$ is sampled from a Gamma distribution of parameters $(2\alpha\bar{\sigma}^2/\delta^2,\delta^2/2\alpha)$, which corresponds to the stationary distribution of the CIR process. To obtain more information about the model one can consult \cite{clinet2018efficient}. The model is inspired directly from \cite{andersen2012jump} and \cite{ait2019hausman}.

\smallskip
In addition, for $V\in \{X,Y\}$ the price jumps are generated via $dJ_t^{V} = M_t^{V} S_t^{V} dN_t^{V}$, where the price jump magnitude $M_t^{V}$ is distributed as $\mathcal{N}(\tilde{\sigma}/\sqrt{10} ,\tilde{\sigma}/10^{3/2})$, the signs of the jumps  $S_t^{V} = \pm 1$ are i.i.d symmetric, $N_t^{\sigma,V}$ is a homogeneous Poisson process with parameter $\bar{\lambda} = 10T/252$ (with that setting jumps occur on average 10 times a year and the contribution of jumps to the total quadratic variation of the price process is around 50\%, both of which are roughly in line with empirical findings in  \cite{huang2005relative}).

\smallskip
Finally, the parameter related to the autocorrelation of residuals' level introduced in (\ref{cointContinuous}) is obviously set to $\rho=1$ in case of no cointegration (i.e. null hypothesis) and chosen equal to $\rho=0.8,0.9$ when there is cointegration (i.e. in the alternative).

\subsubsection{Concurrent methods}
We implement four concurrent leading methods, all of which have already been mentioned: the DF test and the ADF test, and the Phillips-Perron tests $Z_\alpha$ and $Z_\tau$, which are tuned to cointegration in \cite{phillips1990asymptotic}.

\subsubsection{Remaining tuning parameters}
We choose $\overline{\omega} = 0.48$, $a = a_0 \widehat{\sigma}_{MLE}$, $a_0 = 4$, where $\widehat{\sigma}_{MLE}$ is the daily volatility MLE, consistently with the parameter values of the numerical study (Section 5, p. 301) in \cite{clinet2019testing} and \cite{ait2019hausman} (except for $a_0 = 4$, because the original value ($a_0=3$) was yielding too many jumps detection in our case). Parameters related to the deflation are set to $\gamma= 1/2$ and $\gamma'=0.01$.
\begin{table}
\centering
\caption{Overview of models}
\label{tableoverview}
\begin{tabular}{lllll}
\toprule
\toprule
No & Market volatility & Idiosyncratic part of volatility & Drift & Price jumps \\
\toprule
1 & constant & constant & no & no\\
2 & linear & constant & no & no \\
3 & constant + 1 jump & constant & no & no\\
4 & constant & Heston & no & no\\
5 & constant & daily U-shape + jumps & no & no\\
6 & constant & constant & yes & no\\
7 & constant & constant & no & yes\\
8 & linear + 1 jump & Heston + daily U-shape + jumps & yes & yes\\
\bottomrule
\end{tabular}

\end{table}

\begin{table}
\caption{Summary statistics of cointegration relations estimation in case when there is cointegration in Model 8}
\label{tableestcointrel}
\centering
\begin{tabular}{lrrrr}

\toprule
n &  \multicolumn{2}{c}{Modified OLS} &  \multicolumn{2}{c}{Standard OLS} \\
\midrule
\midrule & Bias & Std & Bias & Std\\ 
\midrule
\multicolumn{5}{c}{\emph{Estimation of $c_0$}}\\
\midrule
\multicolumn{5}{l}{\emph{True model: $c_0=1$, $\alpha_0=2$, $\rho=0.8$}}\\
\midrule 
19,656 & 0.005 & 0.202 & 9.075 & 3.087\\
6,552& -0.005& 0.361 & 9.037 & 3.175
\\3,276 & 0.326& 4.202 & 9.106 & 3.007
\\ 1,638 & 9.178& 4.484 & 8.950 & 3.167
\\ 504 & 10.062& 5.782 & 9.163 & 3.062
\\ 252 & 9.570& 4.818 & 9.041 & 3.206\\
\midrule
\multicolumn{5}{l}{\emph{True model: $c_0=1$, $\alpha_0=2$, $\rho=0.9$}}\\
\midrule 
19,656 & 0.002& 0.206 & 9.097 & 3.205\\
6,552& -0.003& 0.348 & 9.102 & 3.020
\\3,276 & 0.286& 3.119 & 8.927 & 2.880
\\ 1,638 & 9.074& 4.263 & 9.061 & 3.091
\\ 504 & 9.700& 5.780 & 9.213 & 3.147
\\ 252 & 9.890& 4.937 & 9.053 & 3.214\\
\midrule
\multicolumn{5}{c}{\emph{Estimation of $\alpha_0$}}\\
\midrule
\multicolumn{5}{l}{\emph{True model: $c_0=1$, $\alpha_0=2$, $\rho=0.8$}}\\
\midrule 
19,656 & -0.001& 0.045 & -1.962 & 0.670 \\
6,552& 0.001& 0.081 & -1.953 & 0.688
\\3,276 & -0.060& 0.829 & -1.969 & 0.650
\\ 1,638 & -1.831& 0.816 & -1.935 & 0.691
\\ 504 & -2.022& 0.733 & -1.979 & 0.665
\\ 252 & -2.001& 0.753 & -1.955 & 0.692\\
\midrule
\multicolumn{5}{l}{\emph{True model: $c_0=1$, $\alpha_0=2$, $\rho=0.9$}}\\
\midrule 
19,656 & -0.001& 0.047 &-1.965 & 0.694\\
6,552& 0.001& 0.078 & -1.969 & 0.657
\\3,276 & -0.056& 0.713 & -1.929 & 0.625
\\ 1,638 & -1.893& 0.786 & -1.962 & 0.670
\\ 504 & -1.992& 0.707 & -1.996 & 0.683
\\ 252 & -2.027& 0.733 & -1.959 & 0.689\\
\bottomrule
\end{tabular}
\end{table}

\begin{table}
\caption{Estimation of autocorrelation of the residuals' level $\rho$ in Model 8}
\label{tableestrho}
\centering
\begin{tabular}{lrr}

\toprule
n & Modified est. $\widehat{\rho}$ & Standard est. $\widetilde{\rho}$ \\
\midrule 
\midrule
\multicolumn{3}{l}{\emph{True model: $\rho=1$}}\\
\midrule 
19,656 & 1.010 & 0.987\\
6,552& 0.994& 0.976
\\3,276 & 0.985& 0.963
\\ 1,638 & 0.996& 0.950\\ 504 & 1.015& 0.913
\\ 252 & 0.934& 0.869\\
\midrule
\multicolumn{3}{l}{\emph{True model: $\rho=0.8$}}\\
\midrule 
19,656 & 0.799 & 0.986\\
6,552& 0.808& 0.976
\\3,276 & 0.824& 0.963
\\ 1,638 & 0.685& 0.950
\\ 504 & 0.942 & 0.906
\\ 252 & 0.846 & 0.870\\
\midrule
\multicolumn{3}{l}{\emph{True model: $\rho=0.9$}}\\
\midrule 
19,656 & 0.899 & 0.986 \\
6,552& 0.905 & 0.975
\\3,276 & 0.910 & 0.967
\\ 1,638 & 1.004 & 0.951
\\ 504 & 0.980 & 0.906
\\ 252 & 0.962& 0.866\\
\bottomrule
\end{tabular}
\end{table}

\begin{figure}
\includegraphics[width=\linewidth]{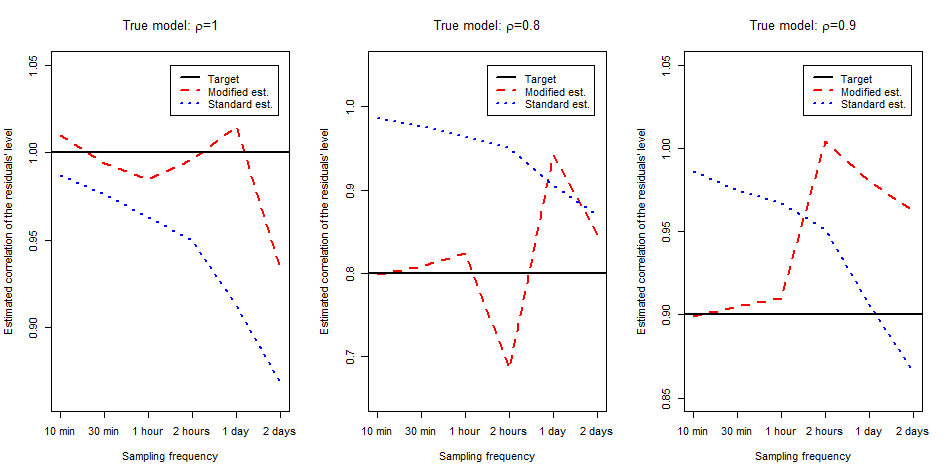}
\centering
\caption{Signature plot of estimated autocorrelation of the residuals' level $\rho$ in Model 8}
\label{sigplotnumstudy}
\end{figure}

\begin{table}
\caption{Size and power properties tabulated to 5\% quantile of several cointegration tests in Model 3}
\label{sizepowermod3}
\centering
\begin{tabular}{lrrrrr}

\toprule
n &  Modified DF &  DF &  ADF &  $Z_\alpha$ &  $Z_\tau$ \\
\midrule
\midrule
\multicolumn{6}{c}{\emph{Size}}\\
\midrule
19,656&0.046&0.235&0.233&0.237&0.174\\
6,552&0.058&0.233&0.228&0.236&0.175\\3,276&0.047&0.229&0.227&0.246&0.176\\
1,638&0.062&0.225&0.231&0.242&0.180\\
504&0.063&0.236&0.225&0.251&0.189\\
252&0.061&0.233&0.228&0.254&0.186\\
\midrule
\multicolumn{6}{c}{\emph{Power, $\rho=0.8$}}\\
\midrule 
19,656&1.000&1.000&1.000&1.000&1.000\\
6,552&1.000&1.000&1.000&1.000&1.000\\
3,276&1.000&1.000&1.000&1.000&1.000\\
1,638&1.000&1.000&1.000&1.000&1.000\\
504&0.999&1.000&1.000&1.000&1.000\\
252&0.978&0.995&0.988&0.915&1.000\\
\midrule
\multicolumn{6}{c}{\emph{Power, $\rho=0.9$}}\\
\midrule 
19,656&1.000&1.000&1.000&1.000&1.000\\
6,552&1.000&1.000&1.000&1.000&1.000\\
3,276&1.000&1.000&1.000&1.000&1.000\\
1,638&1.000&1.000&1.000&1.000&1.000\\
504&0.987&0.998&0.993&0.998&1.000\\
252&0.627&0.859&0.824&0.915&0.935\\
\bottomrule
\end{tabular}
\end{table}

\begin{table}
\caption{Size and power properties tabulated to 5\% quantile of several cointegration tests in Model 7}
\label{sizepowermod7}
\centering
\begin{tabular}{lrrrrr}

\toprule
n &  Modified DF &  DF &  ADF &  $Z_\alpha$ &  $Z_\tau$ \\
\midrule
\midrule
\multicolumn{6}{c}{\emph{Size}}\\
\midrule
19,656&0.048&0.046&0.047&0.044&0.047\\
6,552&0.061&0.054&0.052&0.058&0.057\\
3,276&0.052&0.061&0.060&0.070&0.073\\
1,638&0.039&0.059&0.060&0.071&0.070\\
504&0.050&0.051&0.052&0.056&0.056\\
252&0.048&0.065&0.065&0.066&0.066\\
\midrule
\multicolumn{6}{c}{\emph{Power, $\rho=0.8$}}\\
\midrule 
19,656&1.000&0.103&0.072&0.054&0.075\\
6,552&0.998&0.097&0.087&0.058&0.072\\
3,276&0.641&0.097&0.085&0.068&0.073\\
1,638&0.093&0.104&0.095&0.085&0.093\\
504&0.094&0.101&0.097&0.109&0.122\\
252&0.093&0.102&0.098&0.104&0.113\\
\midrule
\multicolumn{6}{c}{\emph{Power, $\rho=0.9$}}\\
\midrule 
19,656&1.000&0.093&0.078&0.050&0.071\\
6,552&0.994&0.095&0.090&0.062&0.081\\
3,276&0.636&0.096&0.092&0.075&0.083\\
1,638&0.091&0.087&0.083&0.075&0.082\\
504&0.084&0.082&0.082&0.090&0.102\\
252&0.078&0.080&0.080&0.092&0.089\\
\bottomrule
\end{tabular}
\end{table}

\begin{table}
\caption{Size and power properties tabulated to 5\% quantile of several cointegration tests in Model 8}
\label{sizepowermod8}
\centering
\begin{tabular}{lrrrrr}

\toprule
n &  Modified DF &  DF &  ADF &  $Z_\alpha$ &  $Z_\tau$ \\
\midrule
\midrule
\multicolumn{6}{c}{\emph{Size}}\\
\midrule
19,656&0.040&0.063&0.065&0.067&0.053\\
6,552&0.037&0.047&0.048&0.055&0.048\\
3,276&0.039&0.061&0.064&0.070&0.067\\
1,638&0.058&0.057&0.058&0.066&0.050\\
504&0.065&0.070&0.069&0.077&0.072\\
252&0.101&0.066&0.066&0.078&0.063\\
\midrule
\multicolumn{6}{c}{\emph{Power, $\rho=0.8$}}\\
\midrule 
19,656&1.000&0.052&0.052&0.054&0.047\\
6,552&0.998&0.061&0.062&0.065&0.064\\
3,276&0.913&0.061&0.059&0.063&0.067\\
1,638&0.058&0.077&0.075&0.075&0.070\\
504&0.078&0.068&0.069&0.077&0.067\\
252&0.088&0.066&0.063&0.076&0.071\\
\midrule
\multicolumn{6}{c}{\emph{Power, $\rho=0.9$}}\\
\midrule 
19,656&1.000&0.060&0.060&0.058&0.056\\
6,552&0.997&0.046&0.046&0.049&0.044\\
3,276&0.925&0.059&0.061&0.067&0.062\\
1,638&0.045&0.057&0.059&0.067&0.065\\
504&0.074&0.070&0.073&0.075&0.071\\
252&0.113&0.068&0.072&0.072&0.071\\
\bottomrule
\end{tabular}
\end{table}

\subsection{Results}
\subsubsection{Estimation of cointegrated relations via modified OLS}
Table \ref{tableestcointrel} reports the bias and standard deviation of the two parameters of cointegrated relations in the case of the modified OLS and standard OLS in a general model when there is cointegration. It is clear that the standard OLS is defectively biased for any level of subsampling. On the contrary, the modified OLS works well when the frequency of subsampling is high enough, but is equally biased when the frequency decreases. This is due to the truncation method which performs more poorly when the frequency decreases. Thus, a limitation of our method when there is a price jump component is that it requires to sample at reasonable high frequencies, i.e. up to one hour. 

\smallskip
Table \ref{tableestrho} reports the estimated autocorrelation of the residuals' level $\rho$. The corresponding signature plots can be found on Figure \ref{sigplotnumstudy}. The standard estimator is off, notably in the presence of cointegration (i.e. $\rho<1$). The modified estimator is quite reliable when subsampling up to one hour, but insufficient with lower frequencies. The standard and adapted DF can be seen as testing respectively $\widetilde{\rho}=1$ and $\widehat{\rho}=1$.

\subsubsection{Validity of modified DF tests}
We turn now to the behavior of the size and power of the tests. Table \ref{sizepowermod3}-\ref{sizepowermod8} report the size and power of the tests for a variety of models.

\smallskip
Table \ref{sizepowermod3} report the size and power of modified DF and that of the alternative methods when there is one break in market volatility. It is clear that sizes of the concurrent methods are distorted when market volatility is non constant. Reversely, sizes of the modified DF are satisfactory at any level of sampling and for both configurations. The powers of all the methods are not affected. This indicates that the deflation provides a real advantage in practice when market volatility is non constant.

\smallskip
Table \ref{sizepowermod7} reports the statistical properties in case of breaks in price process. We can see that the power of the concurrent methods is distorted when price features jumps. In case of the modified DF, the power is adequate when the sampling frequency is high enough, but not suitable when the frequency decreases. This is what to be expected using the truncation method, and definitely a limitation of our method. Nonetheless, we can see that the truncation is beneficial for whoever implements standard residual based tests for no cointegration with high frequency data.

Finally, Table \ref{sizepowermod8} is concerned with a general model featuring all the aforementioned high frequency features. Mostly, the idiosyncratic effects add to each other, although the sizes of the concurrent methods are somehow not as badly impacted as in the pure non constant market volatility case.
\section{Empirical examples}
We illustrate our methodology by studying two empirical examples where in particular the modified tests results deviate from that of standard tests. The two pairs of stocks considered are Action Construction Equipment Limited (ACE) - Alexion Pharmaceuticals (ALXN) and CMS Energy Corporation (CMS) - Eversource Energy (ES), all of which traded on the S\&P500.\footnote{The data were obtained through Reuters and provided by the Chair of Quantitative Finance of Ecole Centrale Paris.} In line with our numerical study, we consider a two-year-long period, i.e. 2012-2013, and subsample with frequency ranging from ten minutes to two days to conduct the tests.

\subsection{ACE-ALXN case}
Table \ref{empiricaltable1} reports the tests results. The corresponding signature plot of estimated autocorrelation of the residuals' level can be found on Figure \ref{sigplotempstudy1}. The modified DF rejects the null of no cointegration at the highest frequencies, with estimated autocorrelation level around 0.90. On the contrary, the concurrent tests do not reject the null hypothesis. This is an echo of the results available on Table \ref{sizepowermod8} in the case $\rho=0.9$. It seems that there is cointegration, and that due to price jumps, the alternative tests do not reject the null hypothesis. As in the numerical study, the tests results related to the modified DF are unstable when the subsample frequency is higher or equal to two hours. The signature plot in Figure \ref{sigplotempstudy1} is also a replica of that in Figure \ref{sigplotnumstudy} related to the case $\rho=0.9$, and corroborates the aforementioned analysis.
\begin{table}
\caption{Empirical test results and estimated autocorrelation of the residuals' level on the pair ACE-ALXN in period 2012-2013}
\label{empiricaltable1}
\centering
\begin{tabular}{lrrrrrrr}

\toprule
n &  Modified DF &  DF &  ADF &  $Z_\alpha$ &  $Z_\tau$ & $\widehat{\rho}$ & $\widetilde{\rho}$\\
\midrule
\midrule
& \multicolumn{7}{l}{\emph{0: no rejection, 1: rejection}}\\
19,656&1&0&0&0&0&0.940&0.991\\
6,552&1&0&0&0&0&0.894&0.984\\
3,276&1&0&0&0&0&0.896&0.976\\
1,638&1&0&0&0&0&0.895&0.966\\
504&0&0&0&0&0&1.057&0.905\\
252&1&0&0&0&0&0.830&0.782\\
\bottomrule
\end{tabular}
\end{table}

\begin{figure}
\includegraphics[width=\linewidth]{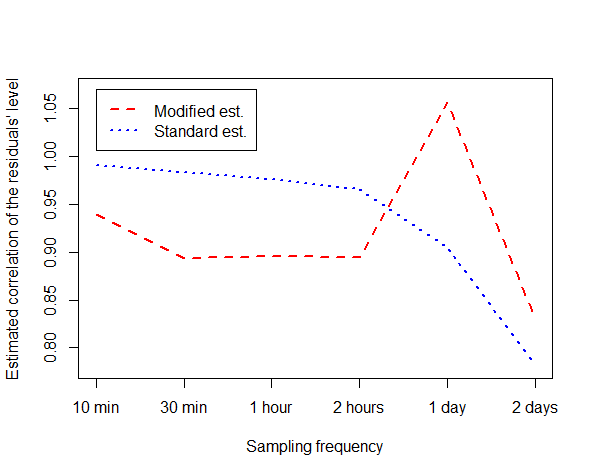}
\centering
\caption{Signature plot of estimated autocorrelation of the residuals' level $\rho$ on the pair ACE-ALXN in period 2012-2013}
\label{sigplotempstudy1}
\end{figure}

\subsection{CMS-ES case}
Table \ref{empiricaltable2} reports the tests results. The related signature plot of estimated autocorrelation of the residuals' level is available on Figure \ref{sigplotempstudy2}. This case is reverse from the previous case. The modified DF does not reject the null of no cointegration at the highest frequencies, while the concurrent tests do reject the null hypothesis. For this particular pair of stocks, results are to be compared with size results in Table \ref{sizepowermod3}. It seems that we should trust modified DF, which indicates no cointegration, whereas the concurrent tests are altered due to time-varying market volatility. Here again the test results related to modified DF are unstable when subsampling with lower frequencies.

\begin{table}
\caption{Empirical test results and estimated autocorrelation of the residuals' level on the pair CMS-ES in period 2012-2013}
\label{empiricaltable2}
\centering
\begin{tabular}{lrrrrrrr}

\toprule
n &  Modified DF &  DF &  ADF &  $Z_\alpha$ &  $Z_\tau$ & $\widehat{\rho}$ & $\widetilde{\rho}$\\
\midrule
\midrule
& \multicolumn{7}{l}{\emph{0: no rejection, 1: rejection}}\\
19,656&0&1&1&1&1&0.822&0.236\\
6,552&0&1&1&1&1&0.856&0.172\\
3,276&0&1&1&1&1&0.988&0.124\\
1,638&0&0&0&0&0&0.985&0.804\\
504&1&1&1&1&1&0.009&0.087\\
252&1&1&1&1&1&0.022&0.005\\
\bottomrule
\end{tabular}
\end{table}

\begin{figure}[H]
\includegraphics[width=\linewidth]{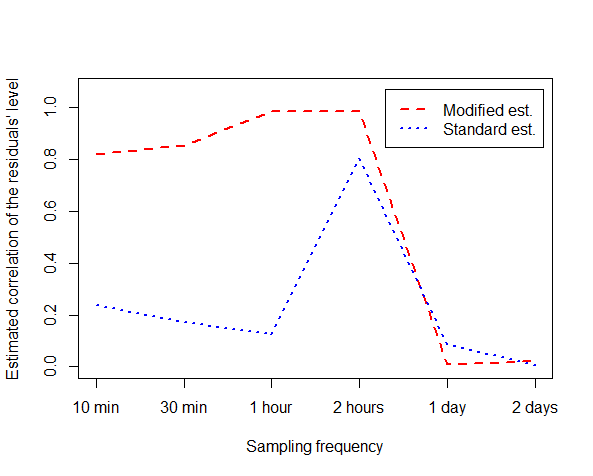}
\centering
\caption{Signature plot of estimated autocorrelation of the residuals' level $\rho$ on the pair CMS-ES in period 2012-2013}
\label{sigplotempstudy2}
\end{figure}

\section{Final remarks}
We have explored the challenges posed by the use of cointegration methods along with high frequency data. In terms of theoretical contribution, we have adapted the problem to the in-fill asymptotics case. We have provided a modified OLS to estimate cointegration relations when there is cointegration, together with its related central limit theory. We have also developed a (non ergodic) time-varying volatility and price-jump robust DF estimator, along with its limit theory.

\smallskip
In terms of applied contribution, we have seen in finite sample that some of the residual based concurrent methods to test for no cointegration are not sufficient when the model accommodates high frequency features, whereas our modified DF showed adequate size and reasonable power. Two empirical examples corroborated the fact that modified DF and standard tests can disagree in practice.

\appendix

\section{Additional finite sample results}
This section reports the size and power of the tests for additional models.

Table \ref{sizepowermod1} reports the size and power of modified DF and that of the concurrent methods in a pure time series environment, i.e. with no high frequency data feature and constant volatility. Evidently, the DF test, which was designed for such environment performs the best, but there is no substantial difference between the size and power of modified DF and that of the concurrent methods. This seems to indicate that the deflation and truncation do not degrade much the behavior of the statistic in this basic framework. Although we would not recommend to use our more sophisticated test given the setup, it is fairly reassuring to see that the modified test is yet substantially in line with the alternative methods.

\smallskip
Table \ref{sizepowermod2} reports the size and power of modified DF and that of the alternative methods when there is a linear trend in market volatility. The results are very comparable to the one break case.

\smallskip
Table \ref{sizepowermod4}-\ref{sizepowermod5} report the size and power properties respectively in the presence of U-shape and jumps in the idiosyncratic component of volatility. It seems that both configurations do not affect the size and power properties. It is not surprising as our assumption of ergodicity on the idiosyncratic part of volatility falls within the setting of \cite{phillips1987time}.

\smallskip
Table \ref{sizepowermod6} reports the size and power properties of modified DF and that of the concurrent methods in the presence of drift. This is considerably important as the theoretical setup related to modified DF does not accommodate with a non zero drift. Comparing to Table \ref{sizepowermod1}, there is no visible difference between the results obtained in case of drift inclusion or not. Actually, we have generated models including drift with parameter values ten times as big as the standard values we can find in the financial literature, and still there was no perceptible effect on the statistics. Our conclusion is that the drift does not seem to affect the size and power properties, at least in this specific stochastic model and within the range of parameter values we will come across on real stocks. Accordingly, we believe that it is reasonably safe to use our tests on stocks data including drift.

\begin{table}
\caption{Size and power properties tabulated to 5\% quantile of several cointegration tests in Model 1}
\label{sizepowermod1}
\centering
\begin{tabular}{lrrrrr}

\toprule
n &  Modified DF &  DF &  ADF &  $Z_\alpha$ &  $Z_\tau$ \\
\midrule
\midrule
\multicolumn{6}{c}{\emph{Size}}\\
\midrule
19,656&0.045&0.051&0.051&0.051&0.053\\
6,552&0.055&0.050&0.051&0.049&0.056\\3,276&0.048&0.049&0.052&0.050&0.060\\
1,638&0.061&0.051&0.051&0.054&0.063\\
504&0.063&0.056&0.061&0.064&0.068\\
252&0.059&0.049&0.056&0.058&0.060\\
\midrule
\multicolumn{6}{c}{\emph{Power, $\rho=0.8$}}\\
\midrule 
19,656&1.000&1.000&1.000&1.000&1.000\\
6,552&1.000&1.000&1.000&1.000&1.000\\
3,276&1.000&1.000&1.000&1.000&1.000\\
1,638&1.000&1.000&1.000&1.000&1.000\\
504&0.999&1.000&1.000&1.000&1.000\\
252&0.984&1.000&0.998&1.000&1.000\\
\midrule
\multicolumn{6}{c}{\emph{Power, $\rho=0.9$}}\\
\midrule 
19,656&1.000&1.000&1.000&1.000&1.000\\
6,552&1.000&1.000&1.000&1.000&1.000\\
3,276&1.000&1.000&1.000&1.000&1.000\\
1,638&1.000&1.000&1.000&1.000&1.000\\
504&0.993&1.000&1.000&1.000&1.000\\
252&0.737&0.849&0.838&0.884&0.933\\
\bottomrule
\end{tabular}
\end{table}

\begin{table}
\caption{Size and power properties tabulated to 5\% quantile of several cointegration tests in Model 2}
\label{sizepowermod2}
\centering
\begin{tabular}{lrrrrr}

\toprule
n &  Modified DF &  DF &  ADF &  $Z_\alpha$ &  $Z_\tau$ \\
\midrule
\midrule
\multicolumn{6}{c}{\emph{Size}}\\
\midrule
19,656&0.044&0.123&0.121&0.125&0.102\\
6,552&0.054&0.124&0.126 &0.133&0.109\\3,276&0.048&0.122&0.121&0.135&0.110\\
1,638&0.062&0.124&0.123&0.136&0.116\\
504&0.064&0.123&0.132&0.139&0.118\\
252&0.062&0.137&0.139&0.148&0.113\\
\midrule
\multicolumn{6}{c}{\emph{Power, $\rho=0.8$}}\\
\midrule 
19,656&1.000&1.000&1.000&1.000&1.000\\
6,552&1.000&1.000&1.000&1.000&1.000\\
3,276&1.000&1.000&1.000&1.000&1.000\\
1,638&1.000&1.000&1.000&1.000&1.000\\
504&1.000&1.000&1.000&1.000&1.000\\
252&0.984&1.000&0.994&1.000&1.000\\
\midrule
\multicolumn{6}{c}{\emph{Power, $\rho=0.9$}}\\
\midrule 
19,656&1.000&1.000&1.000&1.000&1.000\\
6,552&1.000&1.000&1.000&1.000&1.000\\
3,276&1.000&1.000&1.000&1.000&1.000\\
1,638&1.000&1.000&1.000&1.000&1.000\\
504&0.992&1.000&1.000&1.000&1.000\\
252&0.709&0.885&0.867&0.914&0.935\\
\bottomrule
\end{tabular}
\end{table}

\begin{table}
\caption{Size and power properties tabulated to 5\% quantile of several cointegration tests in Model 4}
\label{sizepowermod4}
\centering
\begin{tabular}{lrrrrr}

\toprule
n &  Modified DF &  DF &  ADF &  $Z_\alpha$ &  $Z_\tau$ \\
\midrule
\midrule
\multicolumn{6}{c}{\emph{Size}}\\
\midrule
19,656&0.045&0.048&0.048&0.049&0.054\\
6,552&0.052&0.049&0.049&0.051&0.056\\3,276&0.051&0.050&0.050&0.050&0.061\\
1,638&0.060&0.047&0.049&0.056&0.065\\
504&0.062&0.054&0.060&0.065&0.071\\
252&0.056&0.049&0.056&0.063&0.057\\
\midrule
\multicolumn{6}{c}{\emph{Power, $\rho=0.8$}}\\
\midrule 
19,656&1.000&1.000&1.000&1.000&1.000\\
6,552&1.000&1.000&1.000&1.000&1.000\\
3,276&1.000&1.000&1.000&1.000&1.000\\
1,638&1.000&1.000&1.000&1.000&1.000\\
504&0.997&1.000&1.000&1.000&1.000\\
252&0.984&1.000&0.998&1.000&1.000\\
\midrule
\multicolumn{6}{c}{\emph{Power, $\rho=0.9$}}\\
\midrule 
19,656&1.000&1.000&1.000&1.000&1.000\\
6,552&1.000&1.000&1.000&1.000&1.000\\
3,276&1.000&1.000&1.000&1.000&1.000\\
1,638&1.000&1.000&1.000&1.000&1.000\\
504&0.990&1.000&1.000&1.000&1.000\\
252&0.736&0.849&0.839&0.886&0.937\\
\bottomrule
\end{tabular}
\end{table}

\begin{table}
\caption{Size and power properties tabulated to 5\% quantile of several cointegration tests in Model 5}
\label{sizepowermod5}
\centering
\begin{tabular}{lrrrrr}

\toprule
n &  Modified DF &  DF &  ADF &  $Z_\alpha$ &  $Z_\tau$ \\
\midrule
\midrule
\multicolumn{6}{c}{\emph{Size}}\\
\midrule
19,656&0.045&0.048&0.048&0.049&0.054\\
6,552&0.052&0.049&0.049&0.051&0.056\\3,276&0.051&0.050&0.050&0.050&0.061\\
1,638&0.060&0.047&0.049&0.056&0.065\\
504&0.062&0.054&0.060&0.065&0.071\\
252&0.056&0.049&0.056&0.063&0.057\\
\midrule
\multicolumn{6}{c}{\emph{Power, $\rho=0.8$}}\\
\midrule 
19,656&1.000&1.000&1.000&1.000&1.000\\
6,552&1.000&1.000&1.000&1.000&1.000\\
3,276&1.000&1.000&1.000&1.000&1.000\\
1,638&1.000&1.000&1.000&1.000&1.000\\
504&0.997&1.000&1.000&1.000&1.000\\
252&0.984&1.000&0.998&1.000&1.000\\
\midrule
\multicolumn{6}{c}{\emph{Power, $\rho=0.9$}}\\
\midrule 
19,656&1.000&1.000&1.000&1.000&1.000\\
6,552&1.000&1.000&1.000&1.000&1.000\\
3,276&1.000&1.000&1.000&1.000&1.000\\
1,638&1.000&1.000&1.000&1.000&1.000\\
504&0.990&1.000&1.000&1.000&1.000\\
252&0.736&0.849&0.839&0.886&0.937\\
\bottomrule
\end{tabular}
\end{table}

\begin{table}
\caption{Size and power properties tabulated to 5\% quantile of several cointegration tests in Model 6}
\label{sizepowermod6}
\centering
\begin{tabular}{lrrrrr}

\toprule
n &  Modified DF &  DF &  ADF &  $Z_\alpha$ &  $Z_\tau$ \\
\midrule
\midrule
\multicolumn{6}{c}{\emph{Size}}\\
\midrule
19,656&0.044&0.048&0.048&0.050&0.046\\
6,552&0.052&0.048&0.048&0.048&0.050\\3,276&0.048&0.047&0.048&0.050&0.055\\
1,638&0.057&0.047&0.047&0.052&0.061\\
504&0.062&0.051&0.055&0.064&0.063\\
252&0.056&0.049&0.056&0.062&0.060\\
\midrule
\multicolumn{6}{c}{\emph{Power, $\rho=0.8$}}\\
\midrule 
19,656&1.000&1.000&1.000&1.000&1.000\\
6,552&1.000&1.000&1.000&1.000&1.000\\
3,276&1.000&1.000&1.000&1.000&1.000\\
1,638&1.000&1.000&1.000&1.000&1.000\\
504&0.999&1.000&1.000&1.000&1.000\\
252&0.986&1.000&0.998&1.000&1.000\\
\midrule
\multicolumn{6}{c}{\emph{Power, $\rho=0.9$}}\\
\midrule 
19,656&1.000&1.000&1.000&1.000&1.000\\
6,552&1.000&1.000&1.000&1.000&1.000\\
3,276&1.000&1.000&1.000&1.000&1.000\\
1,638&1.000&1.000&1.000&1.000&1.000\\
504&0.994&1.000&1.000&1.000&1.000\\
252&0.736&0.846&0.837&0.878&0.935\\
\bottomrule
\end{tabular}
\end{table}

\section{Proofs}
 
\subsection{Notation}

For the sake of clarity, most quantities ($T$, $\Delta$, $C_i$,\ldots) introduced in the main body of the paper and which depend on $n$ are explicitly indexed by $n$ ($T_n$, $\Delta_n$, $C_{i,n}$,\ldots) to avoid confusion. In addition to (\ref{notation1})-(\ref{notation3}), we also often introduce for a process $A$ and $i\in \{1,\ldots,n\}$, $t \in [0,T_n]$ the notation $\Delta A_{i,t} = A_{t_i \wedge t} - A_{t_{i-1} \wedge t}$ where $x \wedge y$ is the minimum of $x$ and $y$. 

When $\rho < 1$, (\ref{cointContinuous}) defines $\epsilon$ only for the discrete times $t_0,\ldots,t_n$. For the proofs, it will be more convenient to embed those discrete observations in a process on $[0,T_n]$ as follows. For $t \in [0,T_n]$, we let
\beas 
\epsilon_t = \sum_{j=1}^n \rho^{\Delta_n^{-1}(t-t_j)} \Delta Z_{j,t}.
\eeas 
One immediately checks that for $i \in \{1,\ldots,n\}$, $\epsilon_i$ coincides with (\ref{cointContinuous}) (however, the above definition is different from the interpolation introduced in Remark \ref{rmkOU}). Moreover, since all the estimated quantities are based on the discrete observations only, there is no loss of generality in assuming that $\epsilon$ is defined as above.\\

For a \cadlag function $f$ on $[0,1]$, we write $w'_f$ (or simply $w'$ when there is no room for ambiguity) its associated modulus of continuity as defined in (12.6), p. 122, in \cite{billingsley1999convergence}. We will also often deal with convergence of sequences of \cadlag processes $X_n$ from $[0,1]$ to $\reels^k$, $k \in \naturels - \{0\}$. Accordingly, $X^n \to^{u.c.p} X$ means $\sup_{u \in  [0,1]} \|X_u^n-X_u\| \to^\proba 0$, and $X_n \to^d X$ is the weak convergence with respect to the associated Skorohod topology of the Skorohod space $D_{\reels^k}[0,1]$ (We also use $\to^d$ for the convergence in distribution of simple random variables). By Proposition VI.1.17 from \cite{JacodLimit2003} and Theorem 2.7 from \cite{billingsley1999convergence}, note that, if $X^n \to^d X$ and the limit $X$ is continuous, then for any mapping $f$ on $D_{\reels^k}[0,1]$ which is continuous with respect to the \textit{uniform} topology, then $f(X^n) \to^d f(X)$. For instance, if $X^n \to^d X$ and $X$ is continuous, then for any $s \in [0,1]$, $X_s^n \to^d X_s$, and also $\int_0^1 X_s^nds \to^d \int_0^1 X_sds$. In the proofs, we will often apply this to several mappings which are clearly continuous for the uniform norm. When we do so, we will simply say "by the continuous mapping theorem\ldots" with no further reference to this discussion. \\

Hereafter, $K$ stands for a positive constant which does not depend on $n$ or any other index but may vary from one line to the next. Finally, for an event $E$,  $E^c$ stands for the complementary event.

\subsection{Estimates and preliminary lemmas}

We recall that, under \textbf{[B]}, for any $q \leq 2p_0$ we have for $U \in \{X^c,Z\}$
\bea \label{estimateXc}
\sup_{t \in \reels_+} \esp \l[  \sup_{w \in [0,s]} |U_{t+w} - U_t|^q \r] \leq K s^{q/2}
\eea
as a consequence of the Burkholder-Davis-Gundy inequality.
 
We now proceed to derive a few useful estimates for the jump increments. For $i \in  \{1,\ldots,n\}$, $u \in [0,1]$, let
$$ A_{i,n,u} = \{ |\Delta X_{i,uT_n} |\leq a \Delta_n^{\overline{\omega}} \} \cap \{ |\Delta Y_{i,uT_n}| \leq a \Delta_n^{\overline{\omega}} \}.$$
\begin{lemma*} \label{lemmaJumps}
Let $p \in [r,8] - \{0\}$. There exists a constant $K >0$ such that for $U \in \{X,Y\}$
\bea \label{eqJump1}
\sup_{t \in \reels_+} \esp\l[ \sup_{0 \leq u \leq s } \l|J_{t+u}^U - J_{t}^U\r|^p \r] \leq K(s + s^{p \vee 1}). 
\eea

Moreover, assume $r>0$. Then, for $q \geq r$, there exists $K>0$ which does not depend on $i \in \{1,\ldots,n\}$, such that
\bea \label{eqJump2}
\esp \l[\sup_{u \in [0,1]} |\Delta J_{i,uT_n}^U|^q \mathbb{1}_{\{ |\Delta J_{i,uT_n}^U| < a\Delta_n^{\overline{\omega}}\}}\r] \leq K\Delta_n^{1+\overline{\omega}(q-r)},
\eea
and we have the deviation for $p,q \in \{1,2\}$, $U\in \{X,Y\}$
\bea 
 \esp \sup_{u \in [0,1]} \l| (\Delta U_{i,uT_n})^p\mathbb{1}_{A_{i,n,u}} -  (\Delta U_{i,uT_n}^c)^p\r|^q &\leq& K(\Delta_n^{1+\overline{\omega}(pq-r)} + \Delta_n^{1/2 + p_0(1/2-\overline{\omega})}). \label{eqJump3}
\eea
When $r=0$, (\ref{eqJump2}) and (\ref{eqJump3}) remain true replacing $r$ in the right-hand sides by any positive number arbitrarily close to $0$. 
\end{lemma*}

\begin{proof}
Assume $r>0$. The estimate (\ref{eqJump1}) is a direct consequence of (2.1.40) for $p \in [r ,1]$ and (2.1.41) for $p \in (1,8]$ from Lemma 2.1.7 in \cite{jacod2011discretization}, along with (\ref{assJumps}). To show (\ref{eqJump2}), we first remark that 
\beas  
|\Delta J_{i,uT_n}^U|^q \mathbb{1}_{\{ |\Delta J_{i,uT_n}^U| < a\Delta_n^{\overline{\omega}}\}} &\leq& a^{q-r}\Delta_n^{\overline{\omega}(q-r)}|\Delta J_{i,uT_n}^U|^r ,
\eeas 
so that taking the supremum over $[0,1]$ in $u$, applying expectation on both sides and applying (\ref{eqJump1}) with $p=r$, $s = a\Delta_n^{\overline{\omega}}$ yields the claimed result. Now we show (\ref{eqJump3}) for $U=X$.  First, note that 
\beas 
\l| (\Delta X_{i,uT_n})^p\mathbb{1}_{\{ |\Delta X_{i,uT_n} \leq a \Delta_n^{\overline{\omega}} \}}| -  (\Delta X_{i,uT_n}^c)^p\r|^q &\leq&K \l( |\Delta J_{i,uT_n}^X|^{pq} \mathbb{1}_{\{|\Delta X_{i,uT_n}| < a \Delta_n^{\overline{\omega}}\}} + |\Delta X_{i,uT_n}^c|^{pq} \mathbb{1}_{A_{i,n,u}^c} \r)\\
 &:=& I_u +II_u.
\eeas 
Now, remark that $\mathbb{1}_{\{|\Delta X_{i,uT_n}| < a \Delta_n^{\overline{\omega}}\}} \leq \mathbb{1}_{\{|\Delta J_{i,uT_n}^X| < 2a \Delta_n^{\overline{\omega}}\}} + \mathbb{1}_{\{|\Delta X_{i,uT_n}^c| > a \Delta_n^{\overline{\omega}}\}}$ so that $I_u$ can be further dominated up to a multiplicative constant by $I_u^A + I_u^B$ with
\bea \label{eqDevJumpI}
I_u^A =  \sup_{u \in [0,1]} \l|  \Delta J_{i,uT_n}^X \r|^{pq} \mathbb{1}_{\{|\Delta J_{i,uT_n}^X|<2a\Delta_n^{\overline{\omega}}\}}  
\eea 
and
\bea \label{eqDevJumpII}
I_u^B =  \sup_{u \in [0,1]} \l|  \Delta J_{i,uT_n}^X \r|^{pq} \mathbb{1}_{\{|\Delta X_{i,uT_n}^c|>a\Delta_n^{\overline{\omega}}\}}.
\eea 
By (\ref{eqJump2}), we have
\bea\label{estJump1} 
\esp I_u^A  &\leq&  K\Delta_n^{1+\overline{\omega}(pq-r)}.
\eea
As for the expected value of (\ref{eqDevJumpII}), it can be dominated by
\beas  
 \sqrt{\esp\sup_{u \in [0,1]} \l|\Delta J_{i,uT_n}^X\r|^{2pq} \esp \sup_{u \in [0,1]} \mathbb{1}_{\{|\Delta X_{i,uT_n}^c|>a\Delta_n^{\overline{\omega}}\}}} 
\eeas 
By Lemma \ref{lemmaJumps}, we have $\esp  \sup_{u \in [0,1]} \l|\Delta J_{i,uT_n}^X\r|^{2pq}  \leq K\Delta_n$. Moreover, by (\ref{estimateXc}),
\beas 
\esp \sup_{u \in [0,1]} \mathbb{1}_{\{|\Delta X_{i,uT_n}^c|>a\Delta_n^{\overline{\omega}}\}} &\leq& K \esp \sup_{u \in [0,1]} \Delta_n^{-2p_0\overline{\omega}}|\Delta X_{i,uT_n}^c|^{2p_0} \\
&\leq & K  \Delta_n^{2p_0(1/2 - \overline{\omega})},
\eeas 
so that overall, 
\bea \label{estJump2} 
  \esp \sup_{u \in [0,1]} \l|\Delta J_{i,uT_n}^X\r|^{pq} \mathbb{1}_{\{|\Delta X_{i,uT_n}^c|>a\Delta_n^{\overline{\omega}}\}} &\leq& K \Delta_n^{1/2+p_0(1/2-\overline{\omega})}.
\eea 
Now we deal with $II_u$. Since 
\beas 
\mathbb{1}_{A_{i,n,u}^c} &=& \mathbb{1}_{\{|\Delta X_{i,uT_n}|>a\Delta_n^{\overline{\omega}}\} \cup \{|\Delta Y_{i,uT_n}|>a\Delta_n^{\overline{\omega}}\}  } \\
&\leq& \mathbb{1}_{\{|\Delta X_{i,uT_n}^c| > \frac{a}{2} \Delta_n^{\overline{\omega}}\}} + \mathbb{1}_{\{|\Delta Y_{i,uT_n}^c| > \frac{a}{2} \Delta_n^{\overline{\omega}}\}} + \mathbb{1}_{\{|\Delta J_{i,uT_n}^X| > \frac{a}{2} \Delta_n^{\overline{\omega}}\}} + \mathbb{1}_{\{|\Delta J_{i,uT_n}^Y| > \frac{a}{2} \Delta_n^{\overline{\omega}}\}},
\eeas 
we derive separate estimates for
\bea\label{eqDevJumpIII}
 II_u^A = \sup_{u \in [0,1]}  |\Delta X_{i,uT_n}^c|^{pq}( \mathbb{1}_{\{|\Delta X_{i,uT_n}|^c > \frac{a}{2}\Delta_n^{\overline{\omega}} \} } + \mathbb{1}_{\{|\Delta Y_{i,uT_n}|^c > \frac{a}{2} \Delta_n^{\overline{\omega}}\}} ) 
\eea 
and
\bea \label{eqDevJumpIV}
 II_u^B = \sup_{u \in [0,1]}  |\Delta X_{i,uT_n}^c|^{pq}( \mathbb{1}_{\{|\Delta J_{i,uT_n}^X| > \frac{a}{2} \Delta_n^{\overline{\omega}}\}} + \mathbb{1}_{\{|\Delta J_{i,uT_n}^Y| > \frac{a}{2} \Delta_n^{\overline{\omega}}\}} )
\eea
Similar calculation as for $I_u^B$ yields for the expected value of $II_u^A$ the domination 
\bea \label{estJump3} 
\esp \sup_{u \in [0,1]} |\Delta X_{i,uT_n}^c|^{pq} ( \mathbb{1}_{\{|\Delta X_{i,uT_n}|^c > \frac{a}{2}\Delta_n^{\overline{\omega}} \} } + \mathbb{1}_{\{|\Delta Y_{i,uT_n}|^c > \frac{a}{2} \Delta_n^{\overline{\omega}}\}} ) \leq K \Delta_n^{pq/2+ p_0(1/2-\overline{\omega})}.
\eea 
Finally, we derive an estimate for (\ref{eqDevJumpIV}). For $U\in\{X,Y\}$, 
\bea 
 &&\nonumber \esp\sup_{u \in [0,1]}  |\Delta X_{i,uT_n}^c|^{pq} \mathbb{1}_{\{|\Delta J_{i,uT_n}^U| > \frac{a}{2} \Delta_n^{\overline{\omega}}\}}\\\nonumber &\leq& K   \esp\sup_{u \in [0,1]}|\Delta X_{i,uT_n}^c|^{pq} |\Delta J_{i,uT_n}^U|^{r(1-\frac{pq}{2p_0})} \Delta_n^{-(1-\frac{pq}{2p_0})\overline{\omega}r} \\
\nonumber &\leq& \frac{K}{ \Delta_n^{(1-\frac{pq}{2p_0})\overline{\omega}r}} \l(
 \esp \sup_{u \in [0,1]} |\Delta X_{i,uT_n}^c |^{2p_0}\r)^{\frac{pq}{2p_0}} \l( \esp \sup_{u \in [0,1]} |\Delta J_{i,uT_n}^U|^r\r)^{(1-\frac{pq}{2p_0})}\\
\nonumber&\leq& K \Delta_n^{pq/2+ (1-\frac{pq}{2p_0})(1-\overline{\omega}r)} \\
&\leq& K\Delta_n^{1+ \overline{\omega}(pq-r)}
\label{estJump4}
\eea
where the last estimate is a consequence of $\overline{\omega} < \frac{1}{2} - \frac{3}{2p_0} \leq \frac{p_0-1}{2p_0-r}$ by Assumption \textbf{[C]}, and where we have applied H\"{o}lder inequality at the second step. In view of (\ref{estJump1}), (\ref{estJump2}), (\ref{estJump3}) and (\ref{estJump4}), and using the fact that $pq \geq 1$, (\ref{eqJump3}) for $U=X$ readily follows. The case $U=Y$ is similar, using that under all the alternatives $\esp |\Delta \epsilon_{i,uT_n}|^{2p_0} < K\Delta_n^{p_0}$ for any $i\in\{1,\ldots,n\}$. Finally, if $r=0$, then note that Condition (\ref{assJumps}) is satisfied for any $r^{'} > 0$ arbitrary close to $0$, hence the claimed result.
\end{proof}

Now we devote the next three lemmas to prove the convergence and the uniform boundedness away from 0 of the local realized volatility used for the deflation. First, we need a technical lemma for the c\'{a}dl\'{a}g function $\sigma^M$. 
\begin{lemma*} \label{lemmaSigmaM}
Let $u_n \geq 0$ and $u_n \to 0$. For $i\in \{1,\ldots,n\}$, define
\beas 
D_{i,n} = \int_{\frac{t_i}{T_n}- u_n}^{\frac{t_i}{T_n}} |(\sigma_s^M)^2 - (\sigma_{t_i-}^M)^2|^2 ds.
\eeas 
Then, there exists $\cala_n \subset \{1,\ldots,n\}$, such that $\# \cala_n = o(n)$ and
\beas 
\sup_{i \in \{1,\ldots,n\} - \cala_n} u_n^{-1}D_{i,n} \to 0.
\eeas 
\end{lemma*}
\begin{proof}

For $\eta \geq 0$, set $\calb_\eta = \{i \in \{1,\ldots,n\} | \exists s \in (\frac{t_i}{T_n}- u_n,\frac{t_i}{T_n}) \textnormal{ s.t }  |\Delta (\sigma_s^M)^2| \geq \eta \}$. For $i \in \{1,\ldots,n\} - \calb_\eta$ we easily have that  
\beas
D_{i,n} \leq u_n (2 w'(u_n)^2 + \eta^2),  
\eeas 
where $w'$ is the modulus of continuity of $(\sigma^M)^2$ introduced in the notation section and moreover $\#\calb_\eta \leq N_{\eta} u_n\Delta_n^{-1} $ where $N_\eta$ is the number of jumps of $(\sigma^M)^2$ of size at least $\eta$. Let $a_n >0$ such that $a_n u_n \to 0$ and $a_n \to +\infty$. Since $N_\eta$ is left continuous in $\eta$, if we set $\eta_n = \sup \{ \eta \geq 0 | N_\eta \geq a_n \} +1/n,$  then $N_{\eta_n} \leq a_n$, and it is easy to see that $\eta_n$ must be finite since $N_\eta = 0$ as soon as $\eta$ is larger than the greatest jump of $(\sigma^M)^2$, and moreover $\eta_n \downarrow 0$ because otherwise $(\sigma^M)^2$ would have an infinite number of jumps of size larger than some $\eta_\infty > 0$. Therefore, setting $\cala_{n} = \calb_{\eta_n}$, we get $\# \cala_n  \leq a_n u_n \Delta_n^{-1}= o(n/T_n) = o(n)$ and $$ \sup_{i \in \{1,\ldots,n\} - \cala_n} u_n^{-1} D_{i,n} \leq (2w'(u_n)^2 + \eta_n^2) \to 0.$$ 
\end{proof}

Now we prove the uniform consistency of $C_{i,n}$ outside of the set $\cala_n'$ whose cardinality is negligible with respect to $n$. The following lemma is a stronger version of Proposition \ref{lemmaCiWeak}. 
\begin{lemma*} \label{lemmaRVLocal}
  Let $\gamma' < \gamma \in (0,1)$. Let $k_n = [T_n^\gamma \Delta_n^{-1}]$ and $l_n = [T_n^{\gamma'} \Delta_n^{-1}]$, and finally $\cala_n' = \cala_n \cup \{1,\ldots,2k_n\}$ where $\cala_n$ is as in Lemma \ref{lemmaSigmaM}. Then, uniformly in $i \in \{1,\ldots,n \} - \cala_n'$
$$ \esp |C_{i,n} - (\sigma^M_{(t_i/T_n)-})^2  \omega_{11}|^2 \to 0.$$
\end{lemma*}

\begin{proof}
The proof is conducted in two steps.\\
\textbf{Step 1.} We remove the truncation part. We prove the case $r>0$. The proof is simpler if $r=0$. Defining 
\bea \label{defRvTilde}
\widetilde{RV}_{i,k_n,l_n} = \sum_{j=i-k_n}^{i-l_n-1} (\Delta X_j^c)^2,
\eea 
We prove that uniformly in $i\in\{2k_n+1,\ldots,n\}$, $T_n^{-\gamma}(RV_{i,k_n,l_n} - \widetilde{RV}_{i,k_n,l_n}) \to^{\mathbb{L}^2} 0$. By Jensen's inequality,  (\ref{eqJump3}) applied with $p=q=2$, and using $k_n - l_n - 1 < k_n$, we have that
\beas
T_n^{-2\gamma}\esp|RV_{i,k_n,l_n} - \widetilde{RV}_{i,k_n,l_n}|^2 &\leq& K \Delta_n^{-2}(\Delta_n^{1+\overline{\omega}(4-r)} + \Delta_n^{1/2 + p_0(1/2-\overline{\omega})}).
\eeas 
Note that for $p_0 \geq 8$ and $\overline{\omega} < \frac{p_0-1}{2p_0-r}$ we automatically have $-3/2 + p_0(1/2 - \overline{\omega}) > 0$, and moreover, $\overline{\omega}(4-r) - 1> 0$ since $\overline{\omega} > 1/(4-r)$, so that $T_n^{-2\gamma}\esp|RV_{i,k_n,l_n} - \widetilde{RV}_{i,k_n,l_n}|^2 \to 0$ uniformly in $i \in\{2k_n+1,\ldots,n\}$.

\textbf{Step 2.}  Introducing $A_{i,k_n} =  \sum_{j=i-k_n}^{i  - 1} \int_{t_j}^{t_{j+1}} (\sigma_{s/T_n}^M)^2(\sigma_s^X)^2ds  $, we have by \ito's formula that $\widetilde{RV}_{i,k_n,l_n} - A_{i,k_n} = M_{i,k_n,l_n}   $ with 
$$M_{i,k_n,l_n} = 2\sum_{j=i-k_n}^{i-l_n-1}   \int_{t_j}^{t_{j+1}} ( X_s^{c} -  X_{t_j}^{c} ) \sigma_{s/T_n}^M \sigma_s^X dW_s^X - \sum_{j=i-l_n}^{i-1} \int_{t_j}^{t_{j+1}}(\sigma_{s/T_n}^M)^2(\sigma_s^X)^2ds .$$  
Now, 
\beas
\esp M_{i,k_n,l_n}^2 &\leq& 4   \sum_{j=i-k_n}^{i-l_n-1}  \esp \int_{t_j}^{t_{j+1}} ( X_s^{c} -  X_{t_j}^{c} )^2 (\sigma_{s/T_n}^M)^2 (\sigma_s^X)^2 ds + K \Delta_n^2 l_n^2\\
&\leq& K\Delta_n^2 (k_n+l_n^2) = K(\Delta_n T_n^\gamma + T_n^{2\gamma'}) \\
\eeas 
  
so that uniformly in $i \in \{2k_n+1,\ldots,n\}$, $\esp|\widetilde{RV}_{i,k_n,l_n} - A_{i,k_n} |^2 = O(\Delta_n T_n^{\gamma} + T_n^{2\gamma'})$. Now defining $\widetilde{A_{i,k_n}} = (\sigma^M_{(t_i/T_n)-})^2 \int_{t_{i-k_n}}^{t_{i-1}} (\sigma_s^X)^2 ds,$ we have for $i\in\{k_n+1,\ldots,n\} - \cala_n$:
\beas 
 \esp \left|A_{i,k_n}-\widetilde{A_{i,k_n}}\right|^2 &=&  \esp \left| \sum_{j=i-k_n}^{i-1}    \int_{t_j}^{t_{j+1}} ((\sigma^M_{s/T_n})^2 - (\sigma^M_{(t_i/T_n)-})^2))(\sigma_s^X)^2 ds \right|^2\\
 &\leq& k_n \Delta_n \sum_{j=i-k_n}^{i-1} \int_{t_j}^{t_{j+1}} \left|(\sigma^M_{s/T_n})^2 - (\sigma^M_{(t_i/T_n)-})^2)\right|^2\esp (\sigma_s^X)^4 ds\\
 &\leq& K T_n^{\gamma+1} \sup_{i \in \{k_n+1,\ldots,n\} - \cala_n}D_{i,n} = o(T_n^{2\gamma}),\\
\eeas 
where we have applied Jensen's inequality at the second step and Lemma \ref{lemmaSigmaM} with $u_n = T_n^{\gamma-1}$ at the last step. Finally, By \textbf{[B]}, we immediately deduce that $\esp |\widetilde{A_{i,k_n}} - (\sigma^M_{t_i-/T})^2 T_n^\gamma \omega_{11}|^2 \leq K T_n^{2\gamma} \epsilon(T_n^\gamma) = o(T_n^{2\gamma})$. 
Combining all those inequalities we get that for any $i \in \{1,\ldots,n\} - \cala_n'$
\beas 
\esp |T_n^{-\gamma}\widetilde{RV}_{i,k_n,l_n} - (\sigma^M_{t_i-/T})^2 \omega_{11}|^2 \leq K (a_n + \Delta_nT_n^{-\gamma} + T_n^{2(\gamma' - \gamma)}) \to 0  
\eeas 
for some sequence $a_n \to 0$, and we are done.
 
\end{proof} 
Next, we prove that $C_{i,n}$ are uniformly bounded from below in probability, which will allow us to greatly simplify the subsequent proofs.

\begin{lemma*} \label{lemmaLocalization}
There exists $ 0 < c < \frac{1}{2}(\underline{\sigma}^X)^2 \min_{u \in [0,1]}(\sigma_u^M)^2   $ such that 
\beas
\proba\l[\min_{i \in \{1,\ldots,n\}} C_{i,n} < c\r] \to 0.
\eeas 
\end{lemma*}

\begin{proof}
The proof is conducted in two steps.\\
\textbf{Step 1.} Let $d_n = [\Delta_n^{-2\overline{\omega}}]$. We prove that 
\bea \label{approxCnBlock} 
\min_{j \in \{1, 1+d_n,\ldots, 1+[(n-1)/d_n]d_n\}} C_{j,n} - \min_{i \in \{1,\ldots,n\}} C_{i,n}   \to^\proba 0.
\eea
Indeed, note that for $2k_n+1 \leq j \leq k \leq n$,
$$ |C_{k,n} - C_{j,n}| \leq T_n^{-\gamma}\l( \sum_{l=(k-k_n)\vee (j-l_n)}^{k-l_n-1} \Delta X_l^2 \mathbf{1}_{\{|\Delta X_l| \leq a\Delta_n^{\overline{\omega}}\}} + \sum_{l=(j-k_n)}^{(k-k_n)\wedge (j-l_n-1)} \Delta X_l^2 \mathbf{1}_{\{|\Delta X_l| \leq a\Delta_n^{\overline{\omega}}\}}\r),$$
and since $l_n = o(k_n)$, we immediately deduce that then
\bea \label{estimateCkCj}
|C_{k,n} - C_{j,n}| \leq KT_n^{-\gamma}\Delta_n^{2\overline{\omega}} [(k-j) \wedge k_n]
\eea 
so that, since $d_n \leq k_n$, we have almost surely for $k$ and $j$ larger than $2k_n+1$
\bea \label{inegaliteKJ} 
\max_{|k-j| \leq d_n} |C_{k,n} - C_{j,n}| \leq KT_n^{-\gamma}.
\eea
Now, let $i_n$ be the random index such that $C_{i,n}$ is minimal. There exists $j_n$ of the form $1+k d_n$ such that $|j_n - i_n| \leq d_n$ and clearly $i_n$ and $j_n$ are larger than $2k_n+1$. Then (\ref{approxCnBlock}) can be rewritten as 
\beas 
\min_{j \in \{1, 1+d_n,\ldots, 1+[(n-1)/d_n]d_n\}} C_{j,n} -  C_{i_n,n} &\leq& C_{j_n,n} - C_{i_n,n} \\
&\leq& KT_n^{-\gamma} \to^\proba 0,
\eeas
by (\ref{inegaliteKJ}) and we are done.\\
\textbf{Step 2.} In view of Step 1 and since $C_{i,n} = +\infty$ if $i \leq 2k_n$, we only need to prove that the claimed result holds when the minimum is taken over the subset $\calb_n =\{1, 1+d_n,\ldots, 1+[(n-1)/d_n]d_n\} - \{  1,\ldots,2k_n\}$. Moreover, letting $E_n = \{ \inf_{t \geq t_{k_n}} (\sigma_{t}^X)^2 > 2c/\min_{u \in [0,1]}(\sigma_u^M)^2\}$ with $c$ as in the lemma, by \textbf{[B]} we have $\proba(E_n) \to 1$, so that it is sufficient to prove 
$$\proba\l[\l\{\min_{i \in \calb_n} C_{i,n} < c \r\} \cap E_n \r] \to 0.$$
We have
\beas 
 &&\l\{ \min_{j \in \calb_n} C_{j,n} < c \r\} \cap E_n \\&\subset& \bigcup_{j \in \calb_n} \l\{ \l|C_{j,n} - T_n^{-\gamma} \int_{t_{j-k_n} }^{t_{j-l_n}} (\sigma_{s/T_n}^M)^2 (\sigma_s^X)^2ds\r| >  T_n^{-\gamma} \int_{t_{j-k_n} }^{t_{j-l_n}} (\sigma_{s/T_n}^M)^2 (\sigma_s^X)^2ds - c \r\} \cap E_n\\
 &\subset&  \bigcup_{j \in \calb_n}\l\{ \l|C_{j,n} - T_n^{-\gamma} \int_{t_{j-k_n} }^{t_{j-l_n}} (\sigma_{s/T_n}^M)^2 (\sigma_s^X)^2ds\r| >  c \r\}
\eeas 
so that, using $\mathbb{1}_{\{ |a+b| > c \}} \leq \mathbb{1}_{\{ |a| > c/2\}} + \mathbb{1}_{\{ |b| > c/2\}} \leq 4|a|^2/c^2 + 2|b|/c $ with $a = T_n^{-\gamma}(\widetilde{RV}_{j,k_n,l_n} - \int_{t_{j-k_n} }^{t_{j-l_n}} (\sigma_{s/T_n}^M)^2 (\sigma_s^X)^2ds)$, and $b = RV_{j,k_n,l_n} - \widetilde{RV}_{j,k_n,l_n}$ where $\widetilde{RV}_{j,k_n,l_n}$ was defined in (\ref{defRvTilde}), we get 
\beas 
\proba\l[\min_{j \in \calb_n} C_{j,n} < c\r] &\leq& K\sum_{j\in \calb_n}  \proba \l[\l|C_{j,n} - T_n^{-\gamma} \int_{t_{j-k_n} }^{t_{j-l_n}} (\sigma_{s/T_n}^M)^2 (\sigma_s^X)^2ds\r|^2 > c \r] \\
&\leq& I + II
\eeas
with 
\beas 
I &=& K T_n^{-\gamma}  \sum_{j\in \calb_n}  \esp \l| \sum_{k=l_n}^{k_n - 1} \int_{t_{j-k -1} }^{t_{j-k}} (X_s^c - X_{t_{j-k-1}}) \sigma_{s/T_n}^M \sigma_s^XdW_s^X \r|^2
\eeas 
and 
\beas
II &=&  K T_n^{-\gamma}  \sum_{j\in \calb_n} \esp \l| RV_{j,k_n,l_n} - \widetilde{RV}_{j,k_n,l_n}  \r|.
\eeas 
Application of Burkholder-Davis-Gundy inequality and (\ref{estimateXc}) yields 
\beas 
I \leq Knd_n^{-1} T_n^{-\gamma} k_n \Delta_n^2 &\leq& K \l(T_n^{1+ \frac{1-\gamma}{2\overline{\omega}}}/n\r)^{2\overline{\omega}} \to 0
\eeas 
by Assumption \textbf{[C]}. Moreover, application of (\ref{eqJump3}) for $p=q=1$ yields for $r>0$
\beas 
II \leq K \l(T_n^{1+ \frac{1}{\overline{\omega}(4-r)-1}}/n\r)^{\overline{\omega}(4-r)-1} + K \l(T_n^{1+ \frac{1}{p_0(1/2-\overline{\omega}) + 2\overline{\omega} - 3/2}}/n\r)^{ p_0(1/2-\overline{\omega}) + 2\overline{\omega} - 3/2} \to 0
\eeas 
again by Assumption \textbf{[C]}. Similar reasoning yields $II \to 0$ for $r=0$.
\end{proof}

In view of the previous lemma and a standard localization argument, from now on we will always prove the convergences in probability and in distribution assuming that we are on the event $\{\min_{i\in\{1,\ldots,n\} }C_{i,n} \geq c\}$, which is asymptotically of probability $1$. This amounts to assuming without loss of generality that $C_{i,n}$ is bounded away from $0$ uniformly in $i$:

\smallskip 

\par\noindent \textbf{[H]} There exists $c > 0 $ such that for any $i\in \{2k_n+1,\ldots,n\}$, $C_{i,n} = (T_n^{-\gamma}RV_{i,k_n,l_n}) \vee c$, where $x \vee y$ is the maximum of $x$ and $y$. \\

We now proceed to show that $X$, $Z$ and $\epsilon$, when properly scaled (both in space and time) and seen as \cadlag processes on $[0,1]$, converge in distribution. Let $\widetilde{X}^{c,def,n}$, $\widetilde{Z}^{def,n}$, and $\widetilde{\epsilon}$ be the processes such that for $u \in [0,1]$, we have
\bea \label{defXtilde}
\widetilde{X}_{u}^{c,def,n} = T_n^{-1/2}\l(X_0 + \sum_{i=1}^n C_{i,n}^{-1/2} \Delta X_{i,uT_n}^c\r),
\eea 
a similar definition for $\widetilde{Z}^{def,n}$, and
\bea  \label{defepsilontilde}
\widetilde{\epsilon}_u^n = T_n^{-1/2} \epsilon_{uT_n} = T_n^{-1/2} \sum_{i=1}^n \rho_n^{n\l(u- \frac{t_i}{T_n}\r)} \Delta Z_{i,uT_n}.  
\eea 
Moreover, let $\widetilde{\epsilon}^{def,n}$ be the process such that
$$\Delta \widetilde{\epsilon}_{\frac{t_i}{T_n}}^{def,n} =  \widetilde{\epsilon}_{\frac{t_i}{T_n}}^{def,n} -  \widetilde{\epsilon}_{\frac{t_{i-1}}{T_n}}^{def,n} =  T_n^{-1/2}\frac{\Delta \epsilon_{i}}{\sqrt{C_{i,n}}}$$
for any $i \in \{1,\ldots,n\}$, that is, for $u \in [0,1]$,
\bea 
\widetilde{\epsilon}_{u}^{def,n} = \sum_{i=1}^n C_{i,n}^{-1/2} \Delta \widetilde{\epsilon}_{\frac{t_i}{T_n}, u}^n 
\eea 
where for a process $V$ on $[0,1]$, we write $\Delta V_{\frac{t_i}{T_n}, u} := V_{\frac{t_i}{T_n} \wedge u} - V_{\frac{t_{i-1}}{T_n} \wedge u}$. We also naturally define the scaled process $\widetilde{Y}^{c,def,n}$ as
\bea \label{defYtilde}
\widetilde{Y}_{u}^{c,def,n} = c_0 + \alpha_0 \widetilde{X}_{u}^{c,def,n} + \widetilde{\epsilon}_u^{def,n}, \textnormal{ } u\in[0,1].
\eea 
Finally, consider for $u \in [0,1]$ 
\bea \label{defTX}
\widetilde{\calt}(X)_{u}^{def,n} = T_n^{-1/2}\l(X_0 + \sum_{i=1}^n C_{i,n}^{-1/2}\Delta X_{i,uT_n}\mathbb{1}_{A_{i,n,u}}\r) 
\eea
and
\bea \label{defTY}
\widetilde{\calt}(Y)_{u}^{def,n} = T_n^{-1/2}\l(Y_0 + \sum_{i=1}^n C_{i,n}^{-1/2}\Delta Y_{i,uT_n} \mathbb{1}_{A_{i,n,u}}\r). 
\eea 
We first prove in the following Lemma that $\widetilde{\calt}(X)^{def,n}$ (resp. $\widetilde{\calt}(Y)^{def,n}$) is well approximated by its continuous counterpart $\widetilde{X}^{c,def,n}$ (resp. $\widetilde{Y}^{c,def,n}$).

\begin{lemma*}\label{lemmaTruncationUCP}
We have 
$$ \widetilde{\calt}(X)^n - \widetilde{X}^{c,def,n} \to^{u.c.p} 0,$$
 
and
$$ \widetilde{\calt}(Y)^n - \widetilde{Y}^{c,def,n} \to^{u.c.p} 0.$$
\end{lemma*}

\begin{proof}
We prove the convergence for $X$ and $r>0$. First, recall that $C_{i,n}^{-1/2} < c^{-1/2} < +\infty$ so that 
\beas 
 \widetilde{\calt}(X)_u^{def,n} - \widetilde{X}_u^{c,def,n} &\leq& KT_n^{-1/2} \sum_{i=1}^n   \l| \Delta X_{i,uT_n}\mathbb{1}_{A_{i,n,u}} -  \Delta X_{i,uT_n}^c\r|
\eeas 
and applying (\ref{eqJump3}) with $p=q=1$ yields 
\beas 
\esp \sup_{u \in  [0,1]} |\widetilde{\calt}(X)_u^{def,n} - \widetilde{X}_u^{c,def,n}| &\leq& KT_n^{-1/2}n(\Delta_n^{1+\overline{\omega}(1-r)} + \Delta_n^{1/2+p_0(1/2-\overline{\omega})}) \\
&\leq& K \l(\l(\frac{T_n^{1+\frac{1}{2\overline{\omega}(1-r)}}}{n}\r)^{\overline{\omega}(1-r)} + \l(\frac{T_n^{1+\frac{1}{2p_0(1/2 -\overline{\omega}) -1}}}{n}\r)^{p_0(1/2 - \overline{\omega}) - 1/2 } \r) \to 0
\eeas 
by Assumption \textbf{[C]}. The convergence for $Y$ and for the case $r=0$ can be proven the same way.

\end{proof}

We now show that under the local alternative $\widetilde{\calh}_{1}^{n,\beta}$, the process
$$(\calu_u^n)_{u \in [0,1]} = \left(\widetilde{X}_u^{c,def,n}, \widetilde{Z}_u^{def,n},\widetilde{\epsilon}_u^n\right)_{u \in [0,1]} $$ 
converges in distribution toward a limit which depends on the matrix 
$$ \Omega = \proba-\lim_{T \to + \infty} \frac{1}{T} \int_0^T \Sigma_t dt =\left(\begin{matrix} \omega_{11} & \omega_{12}\\ \omega_{12}& \omega_{22} \end{matrix}\right)$$
and on $B = L W$, where we recall that $W$ is a two dimensional standard Brownian motion on $[0,1]$, and $L$ was defined in (\ref{defL}).
\begin{lemma*}\label{lemmaLimitIntX}
Under $\widetilde{\calh}_{1}^{n,\beta}$, we have the convergence in distribution with respect to the Skorokhod topology of $D_{\reels^3}[0,1]$
$$\left(\calu_u^n\right)_{u \in [0,1]} \to^d \left(\calu_u\right)_{u \in [0,1]}$$
such that
$$ \left(\calu_{u,1},\calu_{u,2}\right)_{u \in [0,1]} = \omega_{11}^{-1/2} (B_u)_{u \in [0,1]},$$
and
$$ \left(\calu_{u,3}\right)_{u \in [0,1]} =  (I(\beta)_u)_{u \in [0,1]} := \left(\int_0^u  e^{-\beta(u-s)}\sigma_s^MdB_s^2\right)_{u \in [0,1]}.$$
 
Under $\calh_1$, we only have the convergence of the subvector $(\calu_{u,1}^n,\calu_{u,2}^n)_{u \in [0,1]} \to^d (\calu_{u,1},\calu_{u,2})_{u \in [0,1]}$.
\end{lemma*}

\begin{proof}
We first prove the functional convergence in distribution for $(\calu_1^n,\calu_2^n, G^n)$, where $G^n$ is the process such that for $u \in [0,1]$, $G_{u}^n = \rho_n^{-un}\widetilde{\epsilon}_u^n = T_n^{-1/2}\l(Z_0 +  \sum_{j=1}^{n} \rho_n^{n \frac{t_j}{T_n}  }\Delta Z_{j,uT_n}\r)$, which is a  martingale. We apply Corollary VIII.3.24 p. 476 from \cite{JacodLimit2003} to the continuous martingale $\calu^n$. We need to prove for any $u \in [0,1]$
\bea \label{bracketXZ}
\left[ (\calu_{1}^n,\calu_{2}^n),(\calu_1^n,\calu_2^n) \right]_u \to^\proba  u\omega_{11}^{-1}\Omega,
\eea
\bea 
\label{bracketEpsilon}
\left[ G^n,G^n \right]_u \to^\proba  \omega_{22}\int_0^u e^{2\beta s}(\sigma_s^M)^2ds,
\eea
\bea \label{bracketEpsilon1}
\left[ \calu_1^n, G^n \right]_u \to^\proba  \omega_{11}^{-1/2} \omega_{12}\int_0^u e^{\beta s}\sigma_s^Mds,
\eea 
and 
\bea \label{bracketEpsilon2}
\left[ \calu_2^n,G^n \right]_u \to^\proba  \omega_{11}^{-1/2} \omega_{22} \int_0^u e^{\beta s }\sigma_s^Mds.
\eea 

We first prove (\ref{bracketXZ}). Note that 
 
$$ \left[ (\calu_1^n,\calu_2^n), (\calu_1^n,\calu_2^n) \right]_u = T_n^{-1} \int_0^{uT_n} C_{t/T_n,n}^{-1} (\sigma_{t-/T_n}^M)^2 \Sigma_t dt,$$
since $\sigma_{t-/T_n}^M = \sigma_{t/T_n}^M$ Lebesgue almost everywhere, and where $C_{.,n}$ is the \cadlag piecewise constant process on $[0,1]$ such that $C_{t_i/T_n,n} = C_{i,n}$. Now, we have for $u\in [0,1]$:
\beas 
&&\esp \left|[\calu_1^n,\calu_1^n]_u - \omega_{11}^{-1}T_n^{-1}\int_0^{uT_n} (\sigma_s^X)^2ds\right| \\&\leq& \omega_{11}^{-1} T_n^{-1}  \esp \sum_{i=1}^{[un]} \int_{t_i}^{t_{i+1}} \left|C_{i,n}^{-1}(\sigma_{s-/T_n}^M)^2\omega_{11} -1\right|(\sigma_s^X)^2ds \\
&\leq& c^{-1}\omega_{11}^{-1}T_n^{-1} \esp \sum_{i=1}^{[un]} \int_{t_i}^{t_{i+1}} \left| (\sigma_{s-/T_n}^M)^2\omega_{11} -C_{i,n}\right|\mathbb{1}_{\{C_{i,n} < +\infty\}}(\sigma_s^X)^2ds \to 0,
\eeas 
because on the one hand, as $\#\cala_n' = o(n)$, we have 
\beas 
\esp \sum_{i \in \cala_n'} \int_{t_i}^{t_{i+1}} \left| (\sigma_{s-/T_n}^M)^2\omega_{11} -C_{i,n}\right|\mathbb{1}_{\{C_{i,n} < +\infty\}}(\sigma_s^X)^2ds = o(T_n) 
\eeas 
by application of (\ref{estimateXc}), and on the other hand, using Cauchy-Schwarz inequality along with Lemma \ref{lemmaRVLocal}, we have
\beas 
\esp \sum_{i \notin \cala_n'} \int_{t_i}^{t_{i+1}} \left| (\sigma_{s-/T_n}^M)^2\omega_{11} - C_{i,n}\right|\mathbb{1}_{\{C_{i,n} < +\infty\}}(\sigma_s^X)^2ds = o(T_n). 
\eeas 
Finally, by \textbf{[B]}, $\omega_{11}^{-1}T_n^{-1}\int_0^{uT_n} (\sigma_s^X)^2ds \to^\proba u$, and this proves (\ref{bracketXZ}) for $\left[ \calu_1^n , \calu_1^n \right]$. The other components are proved similarly. Now we prove (\ref{bracketEpsilon}). Introducing $\tilde{\rho}_n =e^{-\beta/n}$, we have for some constant $L >0$ and for $n$ large enough, and for any $u \in [0,1]$, $|\rho_n^{2un} - \tilde{\rho}_n^{2un}| \leq L/n$. Moreover, recall that $G_{u}^n = \rho_n^{-un}\widetilde{\epsilon}_u^n = T_n^{-1/2}\l(Z_0 +  \sum_{j=1}^{n} \rho_n^{n \frac{t_j}{T_n}  }\Delta Z_{j,uT_n}\r)$, and thus

\beas 
&&\esp\l|\left[ G^n,G^n \right]_u- T_n^{-1}\sum_{j=1}^n \int_{t_{j-1} \wedge uT_n}^{t_{j} \wedge uT_n} e^{2\beta \frac{t_j}{T_n}}(\sigma_{t/T_n}^M)^2 (\sigma_t^Z)^2dt \r|\\&\leq& T_n^{-1} \sum_{j=1}^n\int_{t_{j-1} \wedge uT_n}^{t_{j} \wedge uT_n} |\rho_n^{2n\frac{t_j}{T_n}} - \tilde{\rho}_n^{2n\frac{t_j}{T_n}} | (\sigma_{t/T_n}^M)^2 \esp(\sigma_t^Z)^2dt\\
&\leq& Kn^{-1}.
\eeas 
This yields
\beas 
&&\esp\l|\left[ G^n,G^n \right]_u- T_n^{-1} \int_0^{uT_n} e^{2\beta t}(\sigma_{t/T_n}^M)^2 (\sigma_t^Z)^2 dt  \r|\\&\leq& T_n^{-1} \sum_{j=1}^n\int_{t_{j-1} \wedge uT_n}^{t_{j} \wedge uT_n} |\tilde{\rho}_n^{2n\frac{t_j}{T_n}} - e^{-\beta t} | (\sigma_{t/T_n}^M)^2 \esp(\sigma_t^Z)^2dt +Kn^{-1} \\
&\leq& Kn^{-1} 
\eeas
so that 
\beas 
\left[ G^n,G^n \right]_u &=& T_n^{-1} \int_0^{uT_n} e^{2\beta t}(\sigma_{t/T_n}^M)^2 (\sigma_t^Z)^2 dt + o_\proba(1)\\
&=& T_n^{-1} \sum_{i = 1}^{[\sqrt{uT_n}]} \int_{(i-1)\sqrt{uT_n}}^{i \sqrt{uT_n}} e^{2\beta \frac{t}{T_n}}(\sigma_{t/T_n}^M)^2 (\sigma_t^Z)^2 dt + o_\proba(1)\\
\eeas 
 
Therefore, letting 
$$B_u^n = \esp \left| [G^n,G^n]_u- T_n^{-1}  \sum_{i = 1}^{[\sqrt{uT_n}]} e^{2\beta (i-1)\sqrt{u/T_n}}(\sigma_{(i-1)\sqrt{u/T_n}}^M)^2 \int_{(i-1)\sqrt{uT_n}}^{i \sqrt{uT_n}}  (\sigma_t^Z)^2 dt  \right|,$$
we have
\beas 
 B_u^n &\leq& T_n^{-1} \sum_{i = 1}^{[\sqrt{uT_n}]}  \int_{(i-1)\sqrt{uT_n}}^{i \sqrt{uT_n}}  \left|  e^{2\beta \frac{t}{T_n}}(\sigma_{t/T_n}^M)^2-e^{2\beta (i-1)\sqrt{u/T_n}}(\sigma_{(i-1)\sqrt{u/T_n}}^M)^2 \right| \esp(\sigma_t^Z)^2 dt  \\
&\leq & K \sum_{i = 1}^{[\sqrt{uT_n}]} \int_{(i-1)\sqrt{u/T_n}}^{i \sqrt{u/T_n}}  \left|  e^{2\beta \frac{t}{T_n}}(\sigma_{t}^M)^2-e^{2\beta (i-1)\sqrt{u/T_n}}(\sigma_{(i-1)\sqrt{u/T_n}}^M)^2 \right|dt \to 0,
\eeas 
where we have used that the above integrals are all $o(T_n^{-1/2})$ except for a number of indices which is negligible with respect to $[\sqrt{uT_n}]$, by the same argument as for Lemma \ref{lemmaSigmaM} along with the fact that $e^{2\beta .}(\sigma^M)^2$ is \cadlag. Moreover, by a Riemann sum argument we also have 
$$ \int_0^u e^{2\beta t}(\sigma_t^M)^2dt =  \sqrt{\frac{u}{T_n}} \sum_{i=1}^{[\sqrt{uT_n}]}e^{2\beta (i-1)\sqrt{u/T_n}} (\sigma_{(i-1)\sqrt{u/T_n}}^M)^2 + o(1),$$ 
so that 
\beas 
&& \esp\l|[G^n,G^n]_u - \omega_{22} \int_0^u e^{2\beta t}(\sigma_t^M)^2dt\r|\\
&\leq&  T_n^{-1/2}\sum_{i = 1}^{[\sqrt{uT_n}]} e^{2\beta (i-1)\sqrt{u/T_n}}(\sigma_{(i-1)\sqrt{u/T_n}}^M)^2 \esp  \left|T_n^{-1/2} \int_{(i-1)\sqrt{uT_n}}^{i \sqrt{uT_n}} (\sigma_t^Z)^2 dt - \sqrt{u} \omega_{22}  \right| \\
&\leq& K \max_{1 \leq i\leq [\sqrt{uT_n}] } \esp  \left|T_n^{-1/2} \int_{(i-1)\sqrt{uT_n}}^{i \sqrt{uT_n}} (\sigma_t^Z)^2 dt - \sqrt{u} \omega_{22}  \right| \to 0 
\eeas 
by the triangle inequality and the ergodicity assumption \textbf{[B]}, since $T_n \to +\infty$, which concludes the proof of (\ref{bracketEpsilon}). (\ref{bracketEpsilon1}) and (\ref{bracketEpsilon2}) are proved similarly. 
This gives the functional convergence in distribution for $(\calu_1^n,\calu_2^n,G^n)$. Finally, since for any $u \in [0,1]$, $\calu_{u,3}^n = \rho_n^{un} G_u^n$ and $\rho_n^{un} \to^{u.c.p} e^{-\beta u}$, we get by the continuous mapping theorem the convergence in distribution of $\calu^n$ toward the desired limit. Finally, under $\calh_1$, $\calu_{1}^n$ and $\calu_2^n$ have the same distribution as under $\widetilde{\calh}_1^{n,\beta}$ so that the convergence of the subcomponent $(\calu_{1}^n,\calu_{2}^n)$ remains true. 
\end{proof}

Finally, we prove that several Riemann sums are uniformly convergent.  
\begin{lemma*} \label{lemmaRiemann}
We have 
\bea \label{Riemann1}
\l(n^{-1}\sum_{i=1}^n \widetilde{X}_{\frac{t_i}{T_n} \wedge u}^{c,def,n} - \int_0^u \widetilde{X}_v^{c,def,n}dv\r)_{u \in [0,1]} \to^{u.c.p} 0,
\eea 
and
\bea \label{Riemann2}
\l(n^{-1}\sum_{i=1}^n \widetilde{Z}_{\frac{t_i}{T_n} \wedge u}^{def,n} - \int_0^u \widetilde{Z}_v^{def,n}dv\r)_{u \in [0,1]} \to^{u.c.p} 0.
\eea 
Moreover, under $\widetilde{\calh}_1^{n,\beta}$  
\bea\label{Riemann3}
\l(n^{-1}\sum_{i=1}^n C_{i,n}^{-1/2}\widetilde{\epsilon}_{\frac{t_i}{T_n} \wedge u}^{n} - \omega_{11}^{-1/2}\int_0^u (\sigma_{v}^M)^{-1}\widetilde{\epsilon}_v^{n}dv\r)_{u \in [0,1]} \to^{u.c.p} 0
\eea
and
\bea \label{Riemann4}
\l(n^{-1}\sum_{i=1}^n \widetilde{\epsilon}_{\frac{t_i}{T_n} \wedge u}^{def,n} - \int_0^u \widetilde{\epsilon}_v^{def,n}dv\r)_{u \in [0,1]} \to^{u.c.p} 0.
\eea 
\end{lemma*}

\begin{proof}
(\ref{Riemann1}) and (\ref{Riemann2}) are direct consequences of the fact that uniformly in $i \in \{1,\ldots,n\}$, we have $\esp \sup_{u \in [0,1], v \in [t_{i-1} \wedge u, t_i \wedge u]}|\widetilde{X}_{\frac{t_i}{T_n} \wedge u}^{c,def,n} -  \widetilde{X}_{v}^{c,def,n}| \leq K T_n^{-1/2}\Delta_n^{1/2} \to 0$ and a similar estimate for $\widetilde{Z}^{def,n}$. For (\ref{Riemann3}), note that 
\beas 
&&\esp|n^{-1}\sum_{i=1}^n C_{i,n}^{-1/2}\widetilde{\epsilon}_{\frac{t_i}{T_n} \wedge u}^{n} -\omega_{11}^{-1/2} n^{-1}\sum_{i=1}^n (\sigma_{(t_i/T_n)-}^M)^{-1}\widetilde{\epsilon}_{\frac{t_i}{T_n} \wedge u}^{n}| \\&\leq& n^{-1}\sum_{i=1}^n \sqrt{\esp|C_{i,n}^{-1/2} - \omega_{11}^{-1/2}(\sigma_{(t_i/T_n)-}^M)^{-1}|^2 \underbrace{\esp|\widetilde{\epsilon}_{\frac{t_i}{T_n} \wedge u}^{n}|^2}_{\leq K}} \to 0  
\eeas 
by Lemma \ref{lemmaRVLocal} (separating the cases where $i \in \cala_n'$ and $i \notin \cala_n'$), since outside $\cala_n'$
\beas 
 |C_{i,n}^{-1/2} - \omega_{11}^{-1/2}(\sigma_{(t_i/T_n)-}^M)^{-1}|^2 &\leq&  \frac{|C_{i,n} -\omega_{11} (\sigma_{(t_i/T_n)-}^M)^{2}|^2}{\omega_{11} (\sigma_{(t_i/T_n)-}^M)^{2}C_{i,n}|C_{i,n}^{-1/2} + \omega_{11}^{-1/2}(\sigma_{(t_i/T_n)-}^M)^{-1}|^{2}} \\
 &\leq& K|C_{i,n} -\omega_{11} (\sigma_{(t_i/T_n)-}^M)^{2}|^2
\eeas 
where we have used that $C_{i,n} \geq c >0$. Moreover, we also have 
\beas 
\esp \sup_{u \in [0,1]} |\omega_{11}^{-1/2} n^{-1}\sum_{i=1}^n (\sigma_{(t_i/T_n)-}^M)^{-1}\widetilde{\epsilon}_{\frac{t_i}{T_n} \wedge u}^{n} - \omega_{11}^{-1/2}\int_0^u (\sigma_{v}^M)^{-1}\widetilde{\epsilon}_v^{n}dv| \to 0 
\eeas 
separating again the cases $i \in \cala_n'$ and $i \notin \cala_n'$ and applying Lemma \ref{lemmaSigmaM} with $u_n = \Delta_n$ along with the fact that under $\widetilde{\calh}_1^{n,\beta}$ uniformly in $i \in \{1,\ldots,n\}$, we have $\esp \sup_{u \in [0,1], v \in [t_{i-1} \wedge u, t_i \wedge u]}|\widetilde{\epsilon}_{\frac{t_i}{T_n} \wedge u}^{ n} -  \widetilde{\epsilon}_{v}^{n}| \leq K T_n^{-1/2}\Delta_n^{1/2} \to 0$, and we are done. Finally (\ref{Riemann4}) is proved similarly.
\end{proof}

\subsection{Proofs of Proposition \ref{propOLSH0}, Theorem \ref{thmH0DF} and Theorem \ref{thmRobustDrift}}
We first derive the limit distribution of $\widetilde{\epsilon}^{def,n}$ under any local alternative $\widetilde{\calh}_1^{n,\beta}$.
\begin{lemma*}\label{lemmaEpsilonDef}
Under $\widetilde{\calh}_{1}^{n,\beta}$, Jointly with $\calu^n$, we have the convergence in distribution
\beas 
(\widetilde{\epsilon}_{u}^{def,n})_{u \in [0,1]} \to^d \omega_{11}^{-1/2}(\zeta(\beta)_u)_{u \in [0,1]}:= \omega_{11}^{-1/2}\left(B_u^2 - \beta\int_0^u (\sigma_v^M)^{-1} I(\beta)_vdv\right)_{u \in [0,1]}.
\eeas 

\end{lemma*}

\begin{proof}
Simple algebraic manipulations give 
\beas 
\Delta \widetilde{\epsilon}_{\frac{t_i}{T_n},u}^{def,n}  = \Delta \widetilde{Z}_{\frac{t_i}{T_n},u}^{def,n} - \frac{\beta}{n\sqrt{C_{i,n}}} \widetilde{\epsilon}_{\frac{t_{i-1}}{T_n} \wedge u}^{n},
\eeas 
so that for $i\in\{1,\ldots,n\}$
\beas 
\widetilde{\epsilon}_{\frac{t_i}{T_n} \wedge u}^{def,n}  =   \widetilde{Z}_{\frac{t_i}{T_n} \wedge u}^{def,n} - \frac{\beta}{n} \sum_{j=0}^{i-1} C_{j,n}^{-1/2} \widetilde{\epsilon}_{\frac{t_{j} }{T_n} \wedge u}^{n} 
\eeas 
so that by (\ref{Riemann3}) from Lemma \ref{lemmaRiemann}, 
\bea \label{convEpsilonDefUCP} 
\widetilde{\epsilon}^{def,n}  -   \widetilde{Z}^{def,n} + \beta\omega_{11}^{-1/2} \int_0^{.} (\sigma_{v}^M)^{-1} \widetilde{\epsilon}_{v}^{n} dv \to^{u.c.p} 0.
\eea 
Finally, by the continuous mapping theorem we also have jointly with $\calu^n$
\beas 
\widetilde{Z}^{def,n} - \beta\omega_{11}^{-1/2} \int_0^{.} (\sigma_{v}^M)^{-1} \widetilde{\epsilon}_{v}^{n} dv \to^d \omega_{11}^{-1/2}\zeta(\beta)
\eeas 
which, combined with (\ref{convEpsilonDefUCP}) yields the claimed result. 
\end{proof}

We are now ready to prove Proposition \ref{propOLSH0}.
\begin{proof}[Proof of Proposition \ref{propOLSH0}.]
We first deal with $\widehat{\alpha}$. For two processes $(U_u)_{u \in [0,1]}$, $(V_u)_{u \in [0,1]}$, let us define 
$$C[U,V] = \frac{1}{n} \sum_{i=1}^n ( U_{t_i/T_n} - \mu(U))  ( V_{t_i/T_n} - \mu(V)),$$
where $\mu(U) = n^{-1} \sum_{i=1}^n U_{t_i/T_n}$ and $\mu(V) = n^{-1} \sum_{i=1}^n V_{t_i/T_n}$. Note that by (\ref{alphahatexplicitdef}), (\ref{defTX}) and (\ref{defTY}), we have the representation
\bea \label{representationAlphaTauXTauY}
\widehat{\alpha} = \frac{C[\widetilde{\calt}(X)^{def,n},\widetilde{\calt}(Y)^{def,n}]}{C[\widetilde{\calt}(X)^{def,n},\widetilde{\calt}(X)^{def,n}]}.
\eea 
By the uniform convergences of Lemma \ref{lemmaTruncationUCP} we easily obtain  
\bea \label{Calpha11}
C[\widetilde{\calt}(X)^{def,n},\widetilde{\calt}(Y)^{def,n}] - C[\widetilde{X}^{c,def,n}, \widetilde{Y}^{c,def,n}] \to^\proba 0
\eea 
and
\bea \label{Calpha12}
C[\widetilde{\calt}(X)^{def,n},\widetilde{\calt}(X)^{def,n}] - C[\widetilde{X}^{c,def,n}, \widetilde{X}^{c,def,n}] \to^\proba 0.
\eea
Using now the Riemann sum approximations (\ref{Riemann1}) and (\ref{Riemann4}) from Lemma \ref{lemmaRiemann} yields 
\bea \label{Calpha21}
C[\widetilde{X}^{c,def,n}, \widetilde{Y}^{c,def,n}] - \int_0^1 \left(\widetilde{X}_u^{c,def,n} - \overline{\widetilde{X}^{c,def,n}} \right)  \left(\widetilde{Y}_u^{c,def,n} - \overline{\widetilde{Y}^{c,def,n}} \right)du \to^\proba 0
\eea 
and
\bea \label{Calpha22}
C[\widetilde{X}^{c,def,n}, \widetilde{X}^{c,def,n}] - \int_0^1 \left(\widetilde{X}_u^{c,def,n} - \overline{\widetilde{X}^{c,def,n}} \right)^2du \to^\proba 0
\eea
where we recall that for a process $(V_u)_{u \in [0,1]}$, $\overline{V} = \int_0^1 V_udu$. Next, by lemmas \ref{lemmaLimitIntX} and \ref{lemmaEpsilonDef} we have the convergence in distribution 
\bea \label{convtildeXtildeY} 
\l(\widetilde{X}^{c,def,n},\widetilde{Y}^{c,def,n}\r) \to^d \omega_{11}^{-1/2}\l( B^1, \alpha_0 B^1 + \zeta(\beta)\r),
\eea
so that by application of the continuous mapping theorem along with the convergences (\ref{Calpha11})-(\ref{Calpha22}) we get 
\bea 
\left( \begin{matrix} C[\widetilde{\calt}(X)^{def,n},\widetilde{\calt}(Y)^{def,n}]\\ C[\widetilde{\calt}(X)^{def,n},\widetilde{\calt}(X)^{def,n}] \end{matrix}\right) \to^d \l( \begin{matrix} \int_0^1 F(\beta)_u^1 (\alpha_0F(\beta)_u^1 + F(\beta)_u^2)du \\ \int_0^1(F(\beta)_u^1)^2du \end{matrix}\r)
\eea 
where 
\bea 
F(\beta) = (B^1-\overline{B}^1, \zeta(\beta)-\overline{\zeta(\beta)}).
\eea
By (\ref{representationAlphaTauXTauY}) and another application of the continuous mapping theorem, we obtain
\bea \label{convAlpha}
\widehat{\alpha} \to^d \alpha_0 + \frac{\int_0^1 F(\beta)_u^1 F(\beta)_u^2du}{\int_0^1(F(\beta)_u^1)^2du}.
\eea 
To get the claimed form for the right hand side we now recall the following definitions,
\bea 
J(\beta) = \int_0^.e^{-\beta(u-s)}\sigma_s^M dW_s, 
\eea
\bea 
\xi(\beta) = W^2 -\beta\int_0^. (\sigma_s^M)^{-1} \lambda^TJ(\beta)_sds 
\eea
where $\lambda = (r_\infty/\sqrt{1-r_\infty^2},1)^T$ and $r_\infty = \omega_{12}/\sqrt{\omega_{11}\omega_{22}}$, and
\bea 
H(\beta) = (W^1 - \overline{W}^1, \xi(\beta) - \overline{\xi(\beta)}). 
\eea 
Then, direct calculation gives the linear relation $F(\beta) = LH(\beta)$, so that $\frac{\int_0^1 F(\beta)_u^1 F(\beta)_u^2du}{\int_0^1(F(\beta)_u^1)^2du} = \frac{\omega_{12}}{\omega_{11}} + \sqrt{\frac{\omega_{22}(1-r_{\infty}^2)}{\omega_{11}}} \frac{\int_0^1 H(\beta)_u^1 H(\beta)_u^2du}{\int_0^1 (H(\beta)_u^1)^2du} $. This yields, letting $$\kappa(\beta) = \l(\frac{\int_0^1 H(\beta)_u^1 H(\beta)_u^2du}{\int_0^1 (H(\beta)_u^1)^2du},-1\r)^T,$$
the convergence 
\bea
\widehat{\alpha} - \alpha_0 \to^d  L_{\alpha_0}^T \kappa(\beta)
\eea
with $L_{\alpha_0} = \sqrt{\frac{ \omega_{22}}{\omega_{11}}}\l( \sqrt{1-r_\infty^2}, r_\infty\r)^T$. We now turn our attention to $\widehat{c}$. Recall that we have 
\bea 
T_n^{-1/2}\widehat{c} = \mu\l(\widetilde{\calt}(Y)^{def,n}\r) -\widehat{\alpha}\mu\l(\widetilde{\calt}(X)^{def,n}\r)
\eea 
so that similar arguments as above show that, jointly with $\widehat{\alpha}$, we have the convergence 
\bea \label{convC}
T_n^{-1/2}(\widehat{c}-c_0) \to^d \omega_{11}^{-1/2}\overline{\zeta(\beta)} - \omega_{11}^{-1/2}\frac{\int_0^1 F(\beta)_u^1 F(\beta)_u^2du}{\int_0^1(F(\beta)_u^1)^2du} \overline{B}^1.
\eea 
Replacing once again $F(\beta)^1$ and $F(\beta)^2$ by their respective expressions in terms of $H(\beta)$ in the above expression gives
\bea 
 T^{-1/2}(\widehat{c} - c_0) \to^d L_{c_0}^T \kappa(\beta)
\eea
where $L_{c_0} = -\sqrt{\frac{\omega_{22}}{\omega_{11}}}\l(\overline{W}^1, \overline{\xi(\beta)}\r)^T$, and we are done.
\end{proof}

We can now prove Theorem \ref{thmH0DF}.
\begin{proof}[Proof of Theorem \ref{thmH0DF}.]
By lemmas \ref{lemmaTruncationUCP}, \ref{lemmaLimitIntX} and \ref{lemmaEpsilonDef} along with the continuous mapping theorem, recall that we have
\bea \label{convTauXTauY} 
\l(\widetilde{\calt}(X)^{def,n},\widetilde{\calt}(Y)^{def,n}\r) \to^d \omega_{11}^{-1/2}\l( B^1, \alpha_0 B^1 + \zeta(\beta)\r).
\eea 
Let us now define the scaled estimated residual process $r^n$ as the \cadlag process  such that for $u\in[0,1]$, 
\bea \label{estimRes}
r_{u}^n =  \widetilde{\calt} (Y)_{u}^{def,n} -  T_n^{-1/2}\widehat{c} - \widehat{\alpha}\widetilde{\calt} (X)_{u}^{def,n}.   
\eea 
Note that for $i \in \{1,\ldots,n\}$,  $r_{\frac{t_i}{T_n}}^n$ coincides with $T_n^{-1/2} \widehat{\epsilon}_{t_i}$. Then, by the continuous mapping theorem, (\ref{convTauXTauY}), (\ref{convAlpha}) and (\ref{convC}) we get that with respect to the Skorohod topology of $D_\reels[0,1]$ 
\bea \label{estimRes2}
r^n \to^d -\omega_{11}^{-1/2}\eta(\beta)^T F(\beta)
\eea 
where $\eta(\beta) = ((\int_0^1 (F(\beta)_u^1)^2du)^{-1}\int_0^1 F(\beta)_u^1 F(\beta)_u^2du,-1)^T$. To reformulate (\ref{estimRes2}) in terms of $H(\beta)$, recall that $F(\beta) = LH(\beta)$, so that  
\beas 
L^T\eta(\beta) = \l(\begin{matrix} l_{11}\frac{\int_0^1 F(\beta)_u^1 F(\beta)_u^2du}{\int_0^1(F(\beta)_u^1)^2du} - l_{21} \\ -l_{22}  \end{matrix}\r) = \l(\begin{matrix} l_{22}\frac{\int_0^1 H(\beta)_u^1 H(\beta)_u^2du}{\int_0^1(H(\beta)_u^1)^2du}   \\ -l_{22}  \end{matrix}\r) = l_{22}\kappa(\beta),
\eeas 
where for $i,j \in \{1,2\}$, $l_{ij}$ is the coefficient in position $(i,j)$ of $L$. We thus get
$$ r_n \to^d -\omega_{11}^{-1/2}\eta(\beta)^T F(\beta) = -\omega_{11}^{-1/2}\eta(\beta)^T L H(\beta) = -\omega_{11}^{-1/2}l_{22} \kappa(\beta)^T  H(\beta).$$
Therefore, letting $Q(\beta) = \kappa(\beta)^T H(\beta)$, we deduce from the above results, and the continuous mapping theorem that
$$n \widehat{\phi} = n \frac{\sum_{i=1}^n\widehat{r}_{t_{i-1}/T_n}^n \Delta \widehat{r}_{t_{i}/T_n}^n}{\sum_{i=1}^n(\widehat{r}_{t_{i}/T_n}^n)^2} \to^d \frac{\int_0^1 Q(\beta)_s dQ(\beta)_s}{\int_0^1 Q(\beta)_s^2ds}.$$
Now we derive the limit of $s_{\widehat{\phi}}$, where
\beas 
s_{\widehat{\phi}} = \sqrt{\frac{n^{-1} \sum_{i=1}^n (\Delta \widehat{\epsilon}_i - \widehat{\phi} \widehat{\epsilon}_{i-1})^2}{\sum_{i=1}^n \widehat{\epsilon}_{i-1}^2}}.
\eeas 
We have that 
\begin{eqnarray*}
n^{-1}\sum_{i=1}^n (\Delta \widehat{\epsilon}_i - \widehat{\phi} \widehat{\epsilon}_{i-1})^2 &=& n^{-1}\sum_{i=1}^n \Delta \widehat{\epsilon}_i^2 - 2 n^{-1}\widehat{\phi}\sum_{i=1}^n \Delta \widehat{\epsilon}_i \widehat{\epsilon}_{i-1} + n^{-1}\widehat{\phi}^2 \sum_{i=1}^n \widehat{\epsilon}_{i-1}^2 \\
&=& I+II+III.
\end{eqnarray*}
Moreover
\beas 
\Delta_n^{-1}I &=& T_n^{-1}\sum_{i=1}^n  (\Delta \calt(Y)_i^{def})^2 - \widehat{\alpha}T_n^{-1}\sum_{i=1}^n  \Delta \calt(Y)_i^{def}\Delta \calt(X)_i^{def} + \widehat{\alpha}^2 T_n^{-1}\sum_{i=1}^n (\Delta \calt(X)_i^{def})^2.
\eeas 
Next, since $C_{i,n}^{-1} < c^{-1} < +\infty$ by \textbf{[H]}, and using (\ref{eqJump3}) with $p=2$ and $q=1$, we have
\beas 
T_n^{-1} \esp \sum_{i=1}^n |(\Delta \calt(X)_i^{def})^2 - C_{i,n}^{-1}(\Delta X_i^{c})^2| &\leq& K (\Delta_n^{\overline{\omega}(2-r)} + \Delta_n^{p_0(1/2-\overline{\omega}) - 1/2}) \to 0
\eeas
by Assumption \textbf{[C]}. Similarly we have 
\beas 
T_n^{-1}\sum_{i=1}^n  (\Delta \calt(Y)_i^{def})^2 &=& T_n^{-1}\sum_{i=1}^n C_{i,n}^{-1} (\Delta Y_i^c)^2 + o_\proba(1)\\
\eeas 
and
\beas 
T_n^{-1}\sum_{i=1}^n  \Delta \calt(Y)_i^{def}\Delta \calt(X)_i^{def} = T_n^{-1}\sum_{i=1}^n  C_{i,n}^{-1}\Delta Y_i^c\Delta X_i^c +o_\proba(1). 
\eeas 
Combined with (\ref{bracketXZ}), the fact that $T_n^{-1}\sum_{i=1}^n C_{i,n}^{-1} \Delta \epsilon_i^2 \to^\proba \omega_{11}^{-1}\omega_{22}$, (\ref{convAlpha}), (\ref{convC}) and Slutsky's lemma, this yields the convergence $$ \Delta_n^{-1} I \to^d \omega_{11}^{-1}l_{22}^2\kappa(\beta)^T \kappa(\beta). $$
Moreover straightforward calculation shows that $ II = O_\proba(\Delta_n n^{-1}),$ and $III = O_\proba(\Delta_n n^{-1}),$ using $\widehat{\phi} = O_\proba(n^{-1})$,
so that jointly with $n\widehat{\phi}$ we have
$$ n s_{\widehat{\phi}} \to^d \sqrt{\frac{\kappa(\beta)^T\kappa(\beta)}{\int_0^1 Q(\beta)_s^2ds}},$$
and thus
$$ \Psi \to^d \frac{\int_0^1 Q(\beta)_s dQ(\beta)_s}{\sqrt{\kappa(\beta)^T \kappa(\beta)\int_0^1 Q(\beta)_s^2ds} }.$$
\end{proof}

Finally we prove Theorem \ref{thmRobustDrift}.

\begin{proof}[Proof of Theorem \ref{thmRobustDrift}]
Since for $\Delta_n \to 0$ the increments of the drift term are negligible with respect to the increments of the Brownian integral, we immediately deduce that the lemmas \ref{lemmaJumps}, \ref{lemmaRVLocal}, and \ref{lemmaLocalization} remain true. Next, we replace $C_{i,n}^{-1/2}\Delta X_{i,uT_n}\mathbb{1}_{A_{i,n,u}}$ and $C_{i,n}^{-1/2}\Delta Y_{i,uT_n}\mathbb{1}_{A_{i,n,u}}$ in definitions (\ref{defTX})-(\ref{defTY}) by $C_{i,n}^{-1/2}(\Delta X_{i,uT_n}\mathbb{1}_{A_{i,n,u}} - \Delta_n^{-1}\Delta_{i,n,u}\cald(X))$ and $C_{i,n}^{-1/2}(\Delta Y_{i,uT_n}\mathbb{1}_{A_{i,n,u}} -\Delta_n^{-1}\Delta_{i,n,u}\cald(Y))$ respectively, where $\Delta_{n,i,u} = (t_i \wedge uT_n - t_{i-1} \wedge uT_n)$. Similarly, defining for discrete observations $V_i$, $i\in \{1,\ldots,n\}$
$$ \cald(V) = n^{-1}(V_n -V_0),$$
we replace $\Delta X_{i,uT_n}^c$ and $\Delta Z_{i,uT_n}$ by $\Delta X_{i,uT_n}^c - b^X \Delta_{n,i,u} - \Delta_n^{-1}\Delta_{i,n,u}\cald(X^c - b^X \cdot)$ and  $ \Delta Z_{i,uT_n}  - b^Z \Delta_{n,i,u} - \Delta_n^{-1}\Delta_{i,n,u}\cald(Z - b^Z \cdot)$ in definitions (\ref{defXtilde}) and (\ref{defepsilontilde}). A straightforward application of Lemma \ref{lemmaJumps} as in the previous proofs shows that with these new definitions, Lemma \ref{lemmaTruncationUCP} and Lemma \ref{lemmaRiemann} remain also true. Moreover, by the continuous mapping theorem and Lemma \ref{lemmaLimitIntX} in the case without drift, we have jointly with $T_n^{-1/2}(X^c - b^X \cdot, Z - b^Z \cdot)$ the convergences $nT_n^{-1/2}\cald(X^c - b^X \cdot) \to^d \int_0^1 \sigma_s^M dB_s^1 $ and $nT_n^{-1/2}\cald(Z - b^Z \cdot) \to^d \int_0^1 \sigma_s^M dB_s^2 $, so that by another application of the continuous mapping theorem we deduce that Lemma \ref{lemmaLimitIntX} remains true if in the limits $B$ is replaced by $B - \int_0^\cdot (\sigma_s^M)^{-1} ds \int_0^1 \sigma_s^M dB_s$. Finally, following closely the proofs of Lemma \ref{lemmaEpsilonDef} and Theorem \ref{thmH0DF}, we easily prove that everything holds if in the limits $W-\overline{W}$ and $\xi - \overline{\xi}$ are replaced by $W- \breve{W}$ and $\xi - \breve{\xi}$.  
\end{proof}

\subsection{Proofs of Theorem \ref{thmH1} and Proposition \ref{propH1}}

We begin this section with a technical lemma. Let 
$$ g_n := \sum_{i=1}^n \l(\widetilde{X}_{\frac{t_i}{T_n}}^{c,def,n}-n^{-1}\sum_{j=1}^n\widetilde{X}_{\frac{t_j}{T_n}}^{c,def,n}\r)\l( \widetilde{\epsilon}_{\frac{t_i}{T_n}}^{def,n} - n^{-1}\sum_{j=1}^n\widetilde{\epsilon}_{\frac{t_j}{T_n}}^{def,n} \r).$$

\begin{lemma*} \label{lemGn}
Under $\calh_1$, jointly with $(\calu_1^n,\calu_2^n)$, we have the convergence in distribution $$ g_n \to^d \frac{\omega_{11}^{-1}}{1-\rho} \left(\omega_{12}+\int_0^1 (B_s^1 - \overline{B}^1) dB_s^2 \right).$$
\end{lemma*}

\begin{proof}
We rewrite 
\beas 
g_{n} &:=& \sum_{j=1}^n  \widetilde{X}_{\frac{t_j}{T_n}}^{c,def,n}   \widetilde{\epsilon}_{\frac{t_j}{T_n}}^{def,n}  -  n^{-1}\sum_{j=1}^n\widetilde{X}_{\frac{t_j}{T_n}}^{c,def,n}  \sum_{j=1}^n   \widetilde{\epsilon}_{\frac{t_j}{T_n}}^{def,n}  \\
&=& I+II.
\eeas
The limit for $I$ and $II$ is derived in two steps.\\
\textbf{Step 1.} Let us define for all $i \in \{1,\ldots,n\}$ 
\bea 
q_{\frac{t_i}{T_n}} = \sum_{j=1}^i \rho^{i-j} \Delta \widetilde{Z}_{\frac{t_j}{T_n}}^{def,n}
\eea  
where $\Delta \widetilde{Z}_{\frac{t_j}{T_n}}^{def,n} =  \widetilde{Z}_{\frac{t_j}{T_n}}^{def,n} -  \widetilde{Z}_{\frac{t_{j-1}}{T_n}}^{def,n}$. Note that for $i \leq 2k_n$, $q_{\frac{t_i}{T_n}} = 0$. Accordingly, define also 
\beas 
\widetilde{I} = \sum_{j=1}^n  \widetilde{X}_{\frac{t_j}{T_n}}^{c,def,n}   q_{\frac{t_j}{T_n}} \textnormal{ and } \widetilde{II} =  \sum_{j=1}^n   q_{\frac{t_j}{T_n}}.
\eeas We show that 
\bea \label{convI}
I- \widetilde{I} \to^\proba 0
\eea 
and
\bea \label{convII}
II + n^{-1}\sum_{j=1}^n\widetilde{X}_{\frac{t_j}{T_n}}^{c,def,n}\widetilde{II} \to^\proba 0.
\eea 
Note that elementary algebraic manipulations yield the representation
\bea
\widetilde{\epsilon}_{\frac{t_i}{T_n}}^{def,n} = \sum_{j=1}^i a_{j,n}^i \Delta \widetilde{Z}_{\frac{t_j}{T_n}}^{def,n}
\eea 
where $a_{j,n}^i = 1 + \sqrt{C_{j,n}}(\rho-1)\sum_{k=j+1}^i \frac{\rho^{k-1-j}}{\sqrt{C_{k,n}}}$ with the convention $+\infty \times 0 = 0$ in the last expression, so that if $i \leq 2k_n$ we have $\widetilde{\epsilon}_{\frac{t_i}{T_n}}^{def,n} - q_{\frac{t_i}{T_n}} = 0$ and for $i \geq 2k_n+1$
\beas 
\widetilde{\epsilon}_{\frac{t_i}{T_n}}^{def,n} - q_{\frac{t_i}{T_n}} &=& \sum_{j=2k_n+1}^i (\delta_{j,n}^{i,(1)} + \delta_{j,n}^{i,(2)})  \Delta \widetilde{Z}^{def,n}_{\frac{t_j}{T_n}}
\eeas 
with 
$$ \delta_{j,n}^{i,(1)} = \frac{1-\rho}{\sqrt{C_{j,n}}} \sum_{k=j+1}^{i \wedge (j+l_n)}  \rho^{k-1-j}(\sqrt{C_{k,n}} + \sqrt{C_{j,n}})(C_{k,n} - C_{j,n}),$$
and
$$ \delta_{j,n}^{i,(2)} = \frac{1-\rho}{\sqrt{C_{j,n}}} \sum_{k=i \wedge (j+l_n)+1}^{i}  \rho^{k-1-j}(\sqrt{C_{k,n}} + \sqrt{C_{j,n}})(C_{k,n} - C_{j,n}).$$
Now, using $\esp |\Delta \widetilde{Z}_{\frac{t_j}{T_n}}^{def,n}|^2 \leq c^{-1}Kn^{-1}$, we get 
\beas 
\esp \sum_{j=2k_n+1}^i |\delta_{j,n}^{i,(2)}||\Delta \widetilde{Z}^{def,n}_{\frac{t_j}{T_n}}| &\leq& K n^{-1/2} \sum_{j=2k_n+1}^i \sqrt{\esp|\delta_{j,n}^{i,(2)}|^2}\\
&\leq & Kn \rho^{l_n/2}
\eeas 
since by straightforward calculations $\esp|\delta_{j,n}^{i,(2)}|^2 \leq K \rho^{l_n}$. Now, note that 
since $C_{k,n}$ is $\calf_{j-1}$ measurable for any $k \in \{j,\ldots,j+l_n\}$, we get that $\sum_{j=2k_n+1}^i \delta_{j,n}^{i,(1)}  \Delta \widetilde{Z}^{def,n}_{\frac{t_j}{T_n}}$ is a sum of martingale increments.
Moreover, since we have $|x \vee c - y\vee c| \leq |x-y|$ for any $x,y \in \mathbb{R}$, then for $k \geq j \geq 2k_n+1$, recall that 
$$ |C_{k,n} - C_{j,n}| \leq T_n^{-\gamma}\l( \sum_{l=(k-k_n)\vee (j-l_n)}^{k-l_n-1} \Delta X_l^2 \mathbf{1}_{\{|\Delta X_l| \leq a\Delta_n^{\overline{\omega}}\}} + \sum_{l=(j-k_n)}^{(k-k_n)\wedge (j-l_n-1)} \Delta X_l^2 \mathbf{1}_{\{|\Delta X_l| \leq a\Delta_n^{\overline{\omega}}\}}\r),$$
and since $l_n = o(k_n)$, we immediately deduce for any $q \geq 1$
\bea \label{estimateCkCj2}
|C_{k,n} - C_{j,n}|^q \leq KT_n^{-q\gamma}\Delta_n^{2q\overline{\omega}} [(k-j)^q \wedge k_n^q]
\eea 
which yields for $\delta_{j,n}^{i,(1)}$
\beas \esp |\delta_{j,n}^{i,(1)}|^4 &\leq& K\esp\l(\sum_{k=j+1}^i \rho^{k-1-j} (\sqrt{C_{k,n}} + \sqrt{C_{j,n}})(C_{k,n} - C_{j,n}) \r)^4\\
&\leq& K\l(\frac{1-\rho}{1-\rho^{i-j}}\r)^3 \sum_{k=j+1}^i \rho^{k-1-j} \esp \l[(\sqrt{C_{k,n}} + \sqrt{C_{j,n}})^4(C_{k,n} - C_{j,n})^4\r]\\
&\leq& K T_n^{-4\gamma}\Delta_n^{8\overline{\omega}}\underbrace{\sum_{k=j+1}^i \rho^{k-1-j} (k-j)^4}_{\leq K}\\
&\leq& KT_n^{-4\gamma}\Delta_n^{8\overline{\omega}}
\eeas
where, in the above calculation we have used Jensen's inequality at the second step, along with (\ref{estimateCkCj2}), and we have used at the third step Cauchy-Schwarz inequality along with (\ref{estimateXc}) and the fact that $p_0 \geq 8$. We thus obtain 
\beas 
\esp\l|\sum_{j=2k_n+1}^i \delta_{j,n}^{i,(1)}  \Delta \widetilde{Z}^{def,n}_{\frac{t_j}{T_n}}\r|^2 &=& \sum_{j=2k_n+1}^i \esp \l[(\delta_{j,n}^{i,(1)})^2   (\Delta \widetilde{Z}^{def,n}_{\frac{t_j}{T_n}})^2 \r]  \\
&\leq&  \sum_{j=2k_n+1}^i \sqrt{\esp (\delta_{j,n}^{i,(1)})^4   \esp(\Delta \widetilde{Z}^{def,n}_{\frac{t_j}{T_n}})^4} \\ 
&\leq& K  \Delta_n^{4\overline{\omega}} T_n^{-2\gamma},
\eeas
where we have used that $\esp(\Delta \widetilde{Z}^{def,n}_{\frac{t_j}{T_n}})^4 \leq c^{-2}K n^{-2}$. Finally this yields uniformly in $i \in \{1,\ldots,n\}$
\bea \label{estEpsilonQ}
  \esp |\widetilde{\epsilon}_{\frac{t_i}{T_n}}^{def,n} - q_{\frac{t_i}{T_n}}|^2 \leq K  \Delta_n^{4\overline{\omega}} T_n^{-2\gamma},
\eea
since $n\rho^{l_n/2} = o(\Delta_n^{4\overline{\omega}}T_n^{-\gamma})$, and by similar calculation, we also deduce $\esp |\sum_{i=1}^n (\widetilde{\epsilon}_{\frac{t_i}{T_n}}^{def,n} - q_{\frac{t_i}{T_n}})|^2 = O(\Delta_n^{4\overline{\omega} - 1} T_n^{1-2\gamma}) = o(1)$ since $\gamma \geq 1/2$ and $\overline{\omega} > 1/4$ by \textbf{[C]}. As $n^{-1}\sum_{j=1}^n\widetilde{X}_{\frac{t_j}{T_n}}^{c,def,n}  = O_\proba(1)$, this proves (\ref{convII}). To show (\ref{convI}), it suffices to note that 
\beas 
I - \widetilde{I} = \sum_{j=1}^n \widetilde{X}_{\frac{t_{j-1}}{T_n}}^{c,def,n}(\widetilde{\epsilon}_{\frac{t_j}{T_n}}^{def,n} - q_{\frac{t_j}{T_n}}) +  \sum_{j=1}^n \Delta \widetilde{X}_{\frac{t_{j}}{T_n}}^{c,def,n}(\widetilde{\epsilon}_{\frac{t_j}{T_n}}^{def,n} - q_{\frac{t_j}{T_n}}).
\eeas 
The first term can be treated following exactly the same path as for $II - \widetilde{II}$, multiplying $\delta_{j,n}^{i,(1)}$ and $\delta_{j,n}^{i,(2)}$ by $\widetilde{X}_{\frac{t_{j-1}}{T_n}}^{c,def,n}$ which is $\mathbb{L}_{2p_0}$ bounded and does not affect the estimates. As for the second term, using (\ref{estEpsilonQ}) we have 
\beas 
\esp \l|\sum_{j=1}^n \Delta \widetilde{X}_{\frac{t_{j}}{T_n}}^{c,def,n}(\widetilde{\epsilon}_{\frac{t_j}{T_n}}^{def,n} - q_{\frac{t_j}{T_n}})\r|&\leq&  \sum_{j=1}^n \sqrt{\esp |\Delta \widetilde{X}_{\frac{t_{j}}{T_n}}^{c,def,n}|^2 \esp|(\widetilde{\epsilon}_{\frac{t_j}{T_n}}^{def,n} - q_{\frac{t_j}{T_n}})|^2} \\
&\leq& \Delta_n^{2\overline{\omega} -1/2}T_n^{1/2-\gamma} \to 0
\eeas 
since $\gamma \geq 1/2$ and $\overline{\omega} >1/4$. \\
\textbf{Step 2.} We prove that jointly with $(\calu_1^n,\calu_2^n)$, 
\bea \label{limitIandII} \l(\begin{matrix}\widetilde{I} \\ \widetilde{II}\end{matrix}\r) \to^d \frac{\omega_{11}^{-1}}{1-\rho} \l( \begin{matrix} \omega_{12}+\int_0^1 B_s^1  dB_s^2  \\  \omega_{11}^{-1/2}B_1^2 \end{matrix}\r). 
\eea 
First, note that by definition of $q$, $\widetilde{I}$ and $\widetilde{II}$ can be rewritten as follows:
\beas 
\widetilde{I} &=& \sum_{i=1}^n \widetilde{X}_{\frac{t_{i}}{T_n}}^{c,def,n}q_{\frac{t_{i}}{T_n}}\\
&=& \sum_{j=1}^n \sum_{i=j}^n \sum_{k=j}^i \rho^{i-j}\Delta \widetilde{X}_{\frac{t_{k}}{T_n}}^{c,def,n} \Delta \widetilde{Z}^{def,n}_{\frac{t_j}{T_n}} + \sum_{j=1}^n\sum_{i=j}^n \rho^{i-j} \widetilde{X}^{c,def,n}_{\frac{t_{j-1}}{T_n}} \Delta \widetilde{Z}^{def,n}_{\frac{t_j}{T_n}}  \\
&=& \frac{1}{1-\rho} \l(\underbrace{\sum_{j=1}^n \sum_{k=j}^n \rho^{k-j}(1-\rho^{n-k}) \Delta \widetilde{X}_{\frac{t_{k}}{T_n}}^{c,def,n} \Delta \widetilde{Z}^{def,n}_{\frac{t_j}{T_n}}}_{A} + \underbrace{\sum_{j=1}^n (1-\rho^{n+1-j}) \widetilde{X}^{c,def,n}_{\frac{t_{j-1}}{T_n}} \Delta \widetilde{Z}^{def,n}_{\frac{t_j}{T_n}}}_{B}  \r),
\eeas 
and
\beas 
\widetilde{II} = \frac{1}{1-\rho}\sum_{j=1}^n (1-\rho^{n+1-j}) \Delta \widetilde{Z}^{def,n}_{\frac{t_j}{T_n}}.
\eeas 
Moreover, $A$ can be further decomposed as
\beas 
A &=& \sum_{j=1}^n (1-\rho^{n-j}) \Delta \widetilde{X}_{\frac{t_{j}}{T_n}}^{c,def,n} \Delta \widetilde{Z}^{def,n}_{\frac{t_j}{T_n}} + \sum_{j=1}^n \sum_{k>j}^n \rho^{k-j}(1-\rho^{n-k}) \Delta \widetilde{X}_{\frac{t_{k}}{T_n}}^{c,def,n} \Delta \widetilde{Z}^{def,n}_{\frac{t_j}{T_n}} \\
&=& A_1+A_2.
\eeas
Note that by Jensen's inequality
\beas 
\esp \l(\sum_{j=1}^n \rho^{n-j} \Delta \widetilde{X}_{\frac{t_{j}}{T_n}}^{c,def,n} \Delta \widetilde{Z}^{def,n}_{\frac{t_j}{T_n}}\r)^2  &\leq & \frac{1-\rho^n}{1-\rho}\sum_{j=1}^n \rho^{n-j} \esp(\Delta \widetilde{X}_{\frac{t_{j}}{T_n}}^{c,def,n} \Delta \widetilde{Z}^{def,n}_{\frac{t_j}{T_n}})^2  \\
&\leq & \l(\frac{1-\rho^n}{1-\rho}\r)^2 n^{-2} \to 0, 
\eeas 
and we have also $\sum_{j=1}^n  \Delta \widetilde{X}_{\frac{t_{j}}{T_n}}^{c,def,n} \Delta \widetilde{Z}^{def,n}_{\frac{t_j}{T_n}} \to^\proba \omega_{11}^{-1}\omega_{12},$ so that this proves that $A_1 \to^\proba \omega_{11}^{-1}\omega_{12}$.  Now, remark that $A_2$ can be represented as the sum of martingale increments $A_2 = \sum_{k=1}^n m_{k,n} \Delta \widetilde{X}_{\frac{t_{k}}{T_n}}^{c,def,n}$ with 
$$ m_{n,k} = \l(\sum_{j=1}^{k-1} \rho^{k-j} \Delta \widetilde{Z}^{def,n}_{\frac{t_j}{T_n}}\r)(1-\rho^{n-k}), $$
and thus
\beas  
\esp A_2^2 &\leq & \sum_{k=1}^n \sqrt{\esp [m_{n,k}^4 ] \esp(\Delta \widetilde{X}_{\frac{t_{k}}{T_n}}^{c,def,n})^4}.
\eeas
Again, using Jensen's inequality and $(1-\rho^{n-k}) \leq 1$, we have
\beas 
\esp[m_{n,k}^4] &\leq& \l(\frac{\rho-\rho^k}{1-\rho}\r)^3 \sum_{j=1}^{k-1} \sum_{j=1}^{k-j} \rho^{k-j} \esp (\Delta \widetilde{Z}_{\frac{t_{k}}{T_n}}^{def,n})^4\\
&\leq& K n^{-2},
\eeas 
so that $\esp A_2^2 \leq K n^{-1} \to 0$, and thus overall $A \to^\proba \omega_{11}^{-1}\omega_{12}$. Finally, since we have the immediate approximations
\beas
B = \sum_{j=1}^n \widetilde{X}^{c,def,n}_{\frac{t_{j-1}}{T_n}} \Delta \widetilde{Z}^{def,n}_{\frac{t_j}{T_n}} +o_\proba(1)
= \int_0^1 \widetilde{X}^{c,def,n}_{u} d\widetilde{Z}^{def,n}_{u} +o_\proba(1)
\eeas 
and
$$ \widetilde{II} = \frac{\widetilde{Z}^{def,n}_{1}}{1-\rho}  + o_\proba(1)$$
All we need is to show a joint central limit theorem for the extended process $ (\widetilde{X}^{c,def,n},\widetilde{Z}^{def,n},V^n)$ where $V^n$ is defined as
\bea 
V_{u}^n := \int_0^u \widetilde{X}^{c,def,n}_{s} d \widetilde{Z}^{def,n}_{s},
\eea 
which is a consequence of Lemma \ref{lemmaLimitIntX} along with Theorem 2.2 in \cite{kurtz1991weak}, with $\delta = \infty$ and Condition C2.2(i) being satisfied for any localizing sequence. We have thus with respect to the Skorohod topology of $D_{\mathbb{R}^3}[0,1]$ the convergence  $(\widetilde{X}^{c,def,n},\widetilde{Z}^{def,n},V^n) \to^d (\omega_{11}^{-1/2}B^1,\omega_{11}^{-1/2}B^2,\omega_{11}^{-1}\int_0^. B_s^1dB_s^2)$, which implies by the continuous mapping theorem along with Slutsky's Lemma the convergence (\ref{limitIandII}). Finally, combined with the fact that $n^{-1}\sum_{j=1}^n \widetilde{X}^{c,def,n}_{\frac{t_{j}}{T_n}} = \int_0^1 \widetilde{X}_u^{c,def,n}du + o_\proba(1)$, the continuous mapping theorem, and Step 1 of this proof, the convergence of $g_n$ readily follows.
\end{proof}

Next, we prove a technical lemma for the jump part.
\begin{lemma*}\label{lemmaJumpsH1}
Under $\calh_1$ and $\textnormal{\textbf{[D]}}$, we have 
\beas 
n(\widetilde{\calt}(Y)^{def,n} - T_n^{-1/2}c_0 - \alpha_0 \widetilde{\calt}(X)^{def,n} - \widetilde{\epsilon}^{def,n} ) \to^{u.c.p} 0. 
\eeas 

\end{lemma*}
\begin{proof}
By definition, we have for $u \in [0,1]$
\beas 
n(\widetilde{\calt}(Y)_u^{def,n} - T_n^{-1/2}c_0 - \alpha_0 \widetilde{\calt}(X)_u^{def,n} - \widetilde{\epsilon}_u^{def,n} ) &=& nT_n^{-1/2}\sum_{j=1}^n C_{j,n}^{-1/2}(\Delta J_{j,uT_n}^Y-\alpha_0 \Delta J_{j,uT_n}^X)\mathbb{1}_{A_{j,n,u}} \\
&+& nT_n^{-1/2}\sum_{j=1}^n C_{j,n}^{-1/2} \Delta \epsilon_{j,uT_n} \mathbb{1}_{A_{j,n,u}^c}\\
&=& I_u + II_u.
\eeas 
Using $C_{j,n}^{-1/2} \leq c^{-1/2}$, we have that 
\beas 
\sup_{u \in [0,1]} I_u &\leq& nT_n^{-1/2} \sum_{0<s\leq T_n}|\Delta J_{s}^Y - \alpha_0 \Delta J_s^X| =    o_\proba(1)
\eeas 
by \textbf{[D]}. As for $II$, recall that $\mathbb{1}_{A_{j,n,u}^c} \leq \mathbb{1}_{\{|\Delta X_{j,uT_n}^c| > \frac{a}{2} \Delta_n^{\overline{\omega}}\}} + \mathbb{1}_{\{|\Delta Y_{j,uT_n}^c| > \frac{a}{2} \Delta_n^{\overline{\omega}}\}} + \mathbb{1}_{\{|\Delta J_{j,uT_n}^X| > \frac{a}{2} \Delta_n^{\overline{\omega}}\}} + \mathbb{1}_{\{|\Delta J_{j,uT_n}^Y| > \frac{a}{2} \Delta_n^{\overline{\omega}}\}}$. Now, 
on the one hand for $U \in \{X,Y\}$ and any $q >0$,
\beas 
\sup_{j \in \{1,\ldots,n\}, u\in[0,1]} \mathbb{1}_{\{|\Delta J_{j,uT_n}^U| > \frac{a}{2} \Delta_n^{\overline{\omega}}\}} &\leq& K \sup_{j \in \{1,\ldots,n\}, u\in [0,1]}  |\Delta J_{j,uT_n}^U|^q \Delta_n^{-q\overline{\omega}} = O_\proba(\Delta_n^{q(1/2 -\overline{\omega})})
\eeas 
by \textbf{[D]}, and on the other hand, still for $U \in  \{X,Y\}$, 
\beas
\esp \sup_{u \in [0,1]}\mathbb{1}_{\{|\Delta U_{j,uT_n}^c| > \frac{a}{2} \Delta_n^{\overline{\omega}}\}}  \leq K\Delta_n^{-2p_0\overline{\omega}}\esp \sup_{u \in [0,1]}|\Delta U_{j,uT_n}^c|^{2p_0} \leq  K \Delta_n^{p_0(1-2\overline{\omega})}.
\eeas
Therefore, $\sup_{u \in [0,1]}|II_u| \leq II_A + II_B$ with
\beas
II_A &=& n\sum_{j=1}^nC_{j,n}^{-1/2}\sup_{u \in[0,1]}|\Delta \epsilon_{j,uT_n}|(\mathbb{1}_{\{|\Delta J_{j,uT_n}^X| > \frac{a}{2} \Delta_n^{\overline{\omega}}\}} + \mathbb{1}_{\{|\Delta J_{j,uT_n}^Y| > \frac{a}{2} \Delta_n^{\overline{\omega}}\}}) \\ 
&\leq& nK \underbrace{\sum_{j=1}^n \sup_{u \in[0,1]}|\Delta \epsilon_{j,uT_n}|}_{=O_\proba(n\Delta _n^{1/2})} \sup_{j \in \{1,\ldots,n\}, u\in[0,1]} (\mathbb{1}_{\{|\Delta J_{j,uT_n}^X| > \frac{a}{2} \Delta_n^{\overline{\omega}}\}}+\mathbb{1}_{\{|\Delta J_{j,uT_n}^Y| > \frac{a}{2} \Delta_n^{\overline{\omega}}\}})\\
&=& O_\proba(n^2 \Delta_n^{1/2 + q(1/2 -\overline{\omega})}),
\eeas 
and 
\beas 
II_B &=& n\sum_{j=1}^n \sup_{u \in [0,1]}|\Delta \epsilon_{j,uT_n}| (\mathbb{1}_{\{|\Delta X_{j,uT_n}^c| > \frac{a}{2} \Delta_n^{\overline{\omega}}\}} + \mathbb{1}_{\{|\Delta Y_{j,uT_n}^c| > \frac{a}{2} \Delta_n^{\overline{\omega}}\}}) 
\eeas 
so that by Cauchy-Schwarz inequality we get  
\beas 
\esp II_B &\leq& n^2 \Delta_n^{1/2 + p_0(1/2-\overline{\omega})}.
\eeas
Taking $q\geq p_0$, this yields 
\beas 
II = O_\proba(n^2 \Delta_n^{1/2 + p_0(1/2- \overline{\omega})}) = o_\proba(1)
\eeas 
by assumption \textbf{[C]}.
\end{proof}

We are now ready to prove Theorem \ref{thmH1}.
\begin{proof}[Proof of Theorem \ref{thmH1} and Remark \ref{rmkThmH1}.]
We first prove the consistency of the OLS estimator under \textbf{[A]-[C]}. 
By definition, we have the representation
\bea \label{repC}
T_n^{-1/2}(\widehat{c}-c_0) = \frac{1}{n}\sum_{i=1}^n (\widetilde{\calt}(Y)_{\frac{t_i}{T_n}}^{def,n} - T_n^{-1/2}c_0-\alpha_0 \widetilde{\calt}(X)_{\frac{t_i}{T_n}}^{def,n}) - \frac{\widehat{\alpha}-\alpha_0}{n} \sum_{i=1}^n\widetilde{\calt}(X)_{\frac{t_i}{T_n}}^{def,n}
\eea 
and
\bea \label{repAlpha}
\widehat{\alpha} - \alpha_0 = \frac{f_n}{n^{-1}\sum_{i=1}^n(\widetilde{\calt}(X)_{\frac{t_i}{T_n}}^{def,n} -n^{-1}\sum_{i=1}^n\widetilde{\calt}(X)_{\frac{t_i}{T_n}}^{def,n} )^2  }
\eea 
with
\beas
f_n &=& n^{-1}\sum_{i=1}^n (\widetilde{\calt}(Y)_{\frac{t_i}{T_n}}^{def,n} - T_n^{-1/2}c_0-\alpha_0 \widetilde{\calt}(X)_{\frac{t_i}{T_n}}^{def,n})\widetilde{\calt}(X)_{\frac{t_i}{T_n}}^{def,n}\\&&-n^{-1}\sum_{i=1}^n (\widetilde{\calt}(Y)_{\frac{t_i}{T_n}}^{def,n} - T_n^{-1/2}c_0-\alpha_0 \widetilde{\calt}(X)_{\frac{t_i}{T_n}}^{def,n})\sum_{i=1}^n\widetilde{\calt}(X)_{\frac{t_i}{T_n}}^{def,n}. 
\eeas
By Lemma \ref{lemmaTruncationUCP} (combined with Cauchy-Schwarz inequality for the first term) yields
\beas 
f_n &=&  n^{-1}\sum_{i=1}^n \widetilde{\epsilon}_{\frac{t_i}{T_n}}^{def,n}\widetilde{\calt}(X)_{\frac{t_i}{T_n}}^{def,n}-n^{-1}\sum_{i=1}^n \widetilde{\epsilon}_{\frac{t_i}{T_n}}^{def,n} \sum_{i=1}^n\widetilde{\calt}(X)_{\frac{t_i}{T_n}}^{def,n} + o_\proba(1)\\
&=& n^{-1}g_n +o_\proba(1) \\
&=& o_\proba(1)
\eeas 
by Lemma \ref{lemmaEpsilonDef}. Since the denominator in (\ref{repAlpha}) is stochastically bounded by the continuous mapping theorem and Lemma  \ref{lemmaTruncationUCP}, \ref{lemmaLimitIntX}, \ref{lemmaRiemann}, we get $\widehat{\alpha} \to^ \proba \alpha_0$.  Repeating the same argument in (\ref{repC}) and using the consistency of $\widehat{\alpha}$, we immediately deduce $T_n^{-1/2}(\widehat{c} - c_0) \to^\proba 0$. Now we prove the central limit theorem under the additional condition \textbf{[D]}. Note that now, Lemma \ref{lemmaJumpsH1} (combined with Cauchy-Schwarz inequality for the first term) yields
$$ n f_n =  \sum_{i=1}^n \widetilde{\epsilon}_{\frac{t_i}{T_n}}^{def,n}\widetilde{\calt}(X)_{\frac{t_i}{T_n}}^{def,n}-\sum_{i=1}^n \widetilde{\epsilon}_{\frac{t_i}{T_n}}^{def,n} \sum_{i=1}^n\widetilde{\calt}(X)_{\frac{t_i}{T_n}}^{def,n} + o_\proba(1),$$
which in turn, easily yields 
$$ nf_n = g_n + o_\proba(1)$$
by Lemma \ref{lemmaTruncationUCP} along with the fact that under $\calh_1$, $\esp|\widetilde{\epsilon}_{\frac{t_i}{T_n}}^{def,n}|^{2} \leq K n^{-1}$. Similarly, by Lemma \ref{lemmaTruncationUCP}, \ref{lemmaLimitIntX}, \ref{lemmaRiemann} and the continuous mapping theorem, we have jointly with $(g_n, \calu_1^n, \calu_2^n)$
$$ n^{-1}\sum_{i=1}^n(\widetilde{\calt}(X)_{\frac{t_i}{T_n}}^{def,n} -n^{-1}\sum_{i=1}^n\widetilde{\calt}(X)_{\frac{t_i}{T_n}}^{def,n} )^2 \to^d \omega_{11}^{-1} \int_0^1 (B_u^1 -\overline{B}^1)^2du,$$
and by Lemma \ref{lemGn}, the convergence of distribution of $n(\widehat{\alpha} - \alpha_0)$ toward the claimed distribution readily follows. By similar arguments as for $\widehat{\alpha}$, we also have 
\beas 
nT_n^{-1/2}(\widehat{c} - c_0) = \sum_{i=1}^n \widetilde{\epsilon}_{\frac{t_i}{T_n}}^{def,n} - n(\widehat{\alpha} - \alpha_0)n^{-1}\sum_{i=1}^n \widetilde{X}_{\frac{t_i}{T_n}}^{c,def,n} +o_\proba(1).
\eeas 
By (\ref{estEpsilonQ}) and the convergence of the second component in (\ref{limitIandII}) we have jointly with $(g_n,\calu_1^n,\calu_2^n)$ that $ \sum_{j=1}^n \widetilde{\epsilon}_{\frac{t_j}{T_n}}^{def,n} \to^d (1-\rho)^{-1}\omega_{11}^{-1/2} B_1^2$, and so combined with lemmas \ref{lemmaLimitIntX}, \ref{lemmaRiemann}, \ref{lemGn}, along with the continuous mapping theorem, we deduce that jointly with $n(\widehat{\alpha} - \alpha_0)$, $nT_n^{-1/2}(\widehat{c} - c_0)$ converges toward the claimed distribution. Finally, reformulating the limit as a function of $W$ (using $B= LW$) and conditioning on the first component of $W$ yields the mixed normal representation derived in Remark \ref{rmkThmH1}.  

\end{proof}

\begin{proof}[Proof of Proposition \ref{propH1}]
Defining as in (\ref{estimRes}) for any $u \in [0,1]$ the scaled estimated residual 

\bea \label{representationR}  
r_{u}^n = \widetilde{\calt} (Y)_{u}^{def,n} -  T_n^{-1/2}\widehat{c} - \widehat{\alpha}\widetilde{\calt} (X)_{u}^{def,n},
\eea  
we easily get by the first part of Theorem \ref{thmH1}, and Lemma \ref{lemmaTruncationUCP} that
\bea  \label{convREpsilon}
\sup_{u \in[0,1]} |r_u^n - \widetilde{\epsilon}_u^{def,n}| = o_\proba(1).
\eea
We first derive an estimate for the numerator and the denominator of
$$\widehat{\phi} = \frac{\sum_{i=1}^n \Delta r_{\frac{t_i}{T_n}}^{n} r_{\frac{t_{i-1}}{T_n}}^{n}}{\sum_{i=1}^n \l( r_{\frac{t_i}{T_n}}^{n}\r)^2}.$$ 
Note that we have the identity
\beas 
2\sum_{i=1}^n \Delta r_{\frac{t_i}{T_n}}^{n} r_{\frac{t_{i-1}}{T_n}}^{n} = (r_{1}^n)^2 - (r_0^n)^2 - \sum_{i=1}^n  \l(\Delta r_{\frac{t_i}{T_n}}^{n}\r)^2
\eeas 
which, combined with (\ref{representationR}), (\ref{eqJump3}), (\ref{estEpsilonQ}) and \textbf{[C]} gives
\bea \label{identityQ} 
2\sum_{i=1}^n \Delta r_{\frac{t_i}{T_n}}^{n} r_{\frac{t_{i-1}}{T_n}}^{n} = (q_{1}^n)^2 - (q_0^n)^2 - \sum_{i=1}^n  \l(\Delta q_{\frac{t_i}{T_n}}^{n}\r)^2 + o_\proba(1),
\eea 
where we recall that 
\beas 
q_{\frac{t_i}{T_n}}^n = \sum_{j=1}^i \rho^{i-j} \Delta \widetilde{Z}_{\frac{t_j}{T_n}}^{def,n},
\eeas  
and $\Delta q_{\frac{t_i}{T_n}}^n = q_{\frac{t_i}{T_n}}^n-q_{\frac{t_{i-1}}{T_n}}^n$. From the above representation it is straightforward to check that $q_1^n \to^\proba 0$, $q_0^n = 0$,
\beas 
\sum_{i=1}^n   q_{\frac{t_i}{T_n}}^{n} q_{\frac{t_{i-1}}{T_n}}^{n} \to^\proba \rho(1-\rho^2)^{-1}\omega_{11}^{-1}\omega_{22},
\eeas
and
\bea \label{quadResidual2}  
\sum_{i=1}^{n} \l(q_{\frac{t_{i}}{T_n}}^{n}\r)^2 \to^\proba (1-\rho^2)^{-1}\omega_{11}^{-1}\omega_{22},
\eea
so that, using $(\Delta q_{\frac{t_i}{T_n}}^n)^2 = (  q_{\frac{t_i}{T_n}}^n)^2 + (  q_{\frac{t_{i-1}}{T_n}}^n)^2 -2q_{\frac{t_{i}}{T_n}}^nq_{\frac{t_{i-1}}{T_n}}^n$, we immediately get that  (\ref{identityQ}) yields
\bea \label{estimNumerator} 
\sum_{i=1}^n \Delta r_{\frac{t_i}{T_n}}^{n} r_{\frac{t_{i-1}}{T_n}}^{n} \to^\proba (\rho-1)(1-\rho^2)^{-1}\omega_{11}^{-1}\omega_{22} < 0.
\eea
Now, in general, unfortunately, (\ref{convREpsilon}) is not sufficient to get $ \sum_{i=1}^{n} \l(r_{\frac{t_{i}}{T_n}}^{n}\r)^2 = \sum_{i=1}^{n} \l(q_{\frac{t_{i}}{T_n}}^{n}\r)^2 +o_\proba(1)$. However, we do have by (\ref{estEpsilonQ}) and (\ref{convREpsilon}) the weaker estimate
\bea \label{estimateR2} 
\sum_{i=1}^{n} \l(r_{\frac{t_{i}}{T_n}}^{n}\r)^2 = \sum_{i=1}^{n} \l(q_{\frac{t_{i}}{T_n}}^{n}\r)^2 +o_\proba(n)  = o_\proba(n).
\eea 
Moreover, note that by definition of $\widehat{\phi}$
\bea \label{estimCrossTerm}
\sum_{i=1}^n \l(\Delta r_{\frac{t_{i}}{T_n}}^{n} - \widehat{\phi} r_{\frac{t_{i-1}}{T_n}}^{n}\r)^2 \leq  2\sum_{i=1}^n \l(\Delta r_{\frac{t_{i}}{T_n}}^{n}\r)^2 + 2\l(\sum_{i=1}^n \Delta r_{\frac{t_i}{T_n}}^{n} r_{\frac{t_{i-1}}{T_n}}^{n}\r)^2 = O_\proba(1)
\eea 
by the above calculations. Therefore, by (\ref{estimNumerator}), (\ref{estimateR2}), and (\ref{estimCrossTerm}),  
\beas  
\Psi^{-1} = n^{-1/2} \frac{\sqrt{\sum_{i=1}^n \l(\Delta r_{\frac{t_{i}}{T_n}}^{n} - \widehat{\phi} r_{\frac{t_{i-1}}{T_n}}^{n}\r)^2 \sum_{i=1}^n \l( r_{\frac{t_i}{T_n}}^{n}\r)^2}}{\sum_{i=1}^n \Delta r_{\frac{t_i}{T_n}}^{n} r_{\frac{t_{i-1}}{T_n}}^{n}} \to^\proba 0,
\eeas  
and moreover with probability tending to $1$ we have $\sum_{i=1}^n \Delta r_{\frac{t_i}{T_n}}^{n} r_{\frac{t_{i-1}}{T_n}}^{n} < 0$ by (\ref{estimNumerator}) so that $\Psi \to^\proba -\infty$ which proves the first part of the proposition. Now, under \textbf{[D]}, we easily get by Theorem \ref{thmH1} and Lemma \ref{lemmaJumpsH1} that
\beas  
\sup_{u \in[0,1]} |r_u^n - \widetilde{\epsilon}_u^{def,n}| = O_\proba(n^{-1}),
\eeas  
and combined with (\ref{estEpsilonQ}) and \textbf{[C]} this easily yields
\bea \label{quadResidual1} 
\nonumber \sum_{i=1}^{n-1}  \l(r_{\frac{t_{i}}{T_n}}^{n}\r)^2 &=& \sum_{i=1}^{n-1} \l(q_{\frac{t_{i}}{T_n}}^{n}\r)^2 + o_\proba(1),\\
&\to^\proba& (1-\rho^2)^{-1}\omega_{11}^{-1}\omega_{22},
\eea 
so that 
\bea \label{convPhiH1} 
\widehat{\phi} \to^\proba \rho-1.
\eea 
Moreover, following a similar path as before, we also deduce that 
\beas
\sum_{i=1}^n (\Delta r_{\frac{t_{i}}{T_n}}^{n} - \widehat{\phi} r_{\frac{t_{i-1}}{T_n}}^{n})^2 = \sum_{i=1}^n (\Delta \widetilde{Z}_{\frac{t_i}{T_n}}^{def,n})^2 + o_\proba(1) \to^\proba \omega_{11}^{-1}\omega_{22},
\eeas 
and thus 
\beas 
n^{1/2}s_{\widehat{\phi}} \to^\proba \sqrt{1-\rho^2}
\eeas 
and so
\beas 
\Psi\sim -n^{1/2} \sqrt{\frac{1-\rho}{1+\rho}}.
\eeas 
\end{proof}

Finally we prove the studentized version of the central limit theorem.

\begin{proof}[Proof of Proposition \ref{propStud}]
 
Using the notation introduced in (\ref{estimRes}), and by similar calculations as for (\ref{estimNumerator}), we have
\beas 
\sum_{i=2}^n \Delta \widehat{\epsilon}_i \widehat{\epsilon}_{i-2} = \sum_{i=2}^n \Delta r_{\frac{t_i}{T_n}}^{n} r_{\frac{t_{i-2}}{T_n}}^{n} \to^\proba \rho(\rho-1)(1-\rho^2)^{-1}\omega_{11}^{-1}\omega_{22},
\eeas 
which, along with (\ref{estimNumerator}) proves the consistency of $\widehat{\rho}$. Under \textbf{[D]}, the consistency of $\widehat{v}_\epsilon$ is a direct consequence of (\ref{quadResidual1}) and (\ref{quadResidual2}) from the proof of Proposition \ref{propH1}. Finally, the consistency of $\widehat{r}_{\infty}$ is easily obtained following the same line of reasoning as for the proof of Proposition \ref{propH1}. Now, by Lemma \ref{lemmaTruncationUCP}, \ref{lemmaLimitIntX}, \ref{lemmaRiemann} and the continuous mapping theorem, we immediately deduce that, jointly with $(n(\widehat{\alpha} - \alpha_0), nT_n^{-1/2}(\widehat{c} - c_0))$, we have 
\beas 
 (\overline{\calt(X)}^{def},I[\calt(X)^{def}],J[\calt(X)^{def}],K[\calt(X)^{def}]) \to^d (\overline{W}^1, I[W^1], J[W^1],K[W^1]), 
\eeas 
which, combined with the consistency of $\widehat{v}_\epsilon$, $\widehat{\rho}$, and $\widehat{r}_{\infty}$ along with Slutsky's Lemma and the continuous mapping theorem yields the claimed result.
\end{proof}

\bibliography{biblio/biblio}
 
\bibliographystyle{apalike} 

\end{document}